\documentclass{amsart}
\usepackage[T1]{fontenc}
\usepackage{amsmath,a4wide}
\usepackage{amssymb}
\usepackage{amscd}
\usepackage{epsfig}
\usepackage[T1]{fontenc}
\usepackage{epsfig}
\usepackage{amsmath}
\usepackage{graphicx}
\usepackage{subfigure}
\usepackage{mathrsfs}
\usepackage{epstopdf}
\usepackage{color}
\usepackage{booktabs}
\usepackage{ifpdf}
\usepackage{cite}
\usepackage{amsthm}
\usepackage{amsmath}
\definecolor{myblue}{rgb}{0.2,0,0.9}
\definecolor{blue-violet}{rgb}{0.54, 0.17, 0.89}
\usepackage{enumitem}
\usepackage{algorithm,algorithmic}

\usepackage{pgfplots}
\usepgfplotslibrary{groupplots}
\pgfplotsset{compat=1.12}

\usepackage{tikz}
\usetikzlibrary{arrows,intersections}
\usepackage{boldline}
\usepackage{multirow}
\usepackage[margin=1.32in]{geometry}
\usepackage{varioref}
\usepackage{xr-hyper}
\usepackage{hyperref}
\hypersetup{
	colorlinks,linkcolor=blue,citecolor=blue,filecolor=red,urlcolor=blue}

\definecolor{myblue}{rgb}{0.2,0,0.9}
\definecolor{blue_violet}{rgb}{0.54, 0.17, 0.89}
\definecolor{darkgreen}{rgb}{0,0.35,0}

\newcommand{\EZ}{\text{EZ}}
\newcommand{\cR}{{\cal R}}

\allowdisplaybreaks

\makeatletter
\DeclareRobustCommand*\cal{\@fontswitch\relax\mathcal}
\makeatother

\newtheorem{thm}{Theorem}[section]
\newtheorem{pro}[thm]{Proposition}

\newtheorem{cor}[thm]{Corollary}
\newtheorem{lem}[thm]{Lemma}

\numberwithin{equation}{section}

\theoremstyle{definition}
\newtheorem{rem}[thm]{Remark}
\newtheorem{dfn}[thm]{Definition}
\newtheorem{as}[thm]{Assumption}
\usepackage{dsfont}

\DeclareMathOperator*{\esssup}{ess\,sup}

\usepackage{xparse}
\makeatletter
\RenewDocumentCommand{\title}{om}{%
	\IfNoValueTF{#1}
	{\gdef\shorttitle{}}
	{\gdef\shorttitle{#1}}%
	\gdef\@title{#2}%
}
\makeatother

\DeclareMathOperator*{\essinf}{ess\,inf}

\makeatother

\newcommand{\Romannum}[1]{\uppercase\expandafter{\romannumeral #1}}

\newcommand{\labitem}[2]{%
	\def\@itemlabel{{#1}}
	\item
	\def\@currentlabel{#1}\label{#2}}
\makeatother

\makeatletter
\def\@tocline#1#2#3#4#5#6#7{\relax
	\ifnum #1>\c@tocdepth % then omit
	\else
	\par \addpenalty\@secpenalty\addvspace{#2}%
	\begingroup \hyphenpenalty\@M
	\@ifempty{#4}{%
		\@tempdima\csname r@tocindent\number#1\endcsname\relax
	}{%
		\@tempdima#4\relax
	}%
	\parindent\z@ \leftskip#3\relax \advance\leftskip\@tempdima\relax
	\rightskip\@pnumwidth plus4em \parfillskip-\@pnumwidth
	#5\leavevmode\hskip-\@tempdima
	\ifcase #1
	\or\or \hskip 2em \or \hskip 2em \else \hskip 3em \fi%
	#6\nobreak\relax
	\hfill\hbox to\@pnumwidth{\@tocpagenum{#7}}\par% <---- \dotfill -> \hfill
	\nobreak
	\endgroup
	\fi}
\makeatother
\title{Robust dividend policy: Equivalence of \\ Epstein-Zin and Maenhout preferences}
\author{Kexin Chen}
\address{Department of Applied Mathematics, The Hong Kong Polytechnic University, Hong Kong}
\email{kexinchen@polyu.edu.hk}

\author{Kyunghyun Park}
\address{Division of Mathematical Sciences, Nanyang Technological University, Singapore}
\email{kyunghyun.park@ntu.edu.sg}

\author{Hoi Ying Wong}
\address{Department of Statistics and Data Science, The Chinese University of Hong Kong, Shatin, N.T., Hong Kong}
\email{hywong@cuhk.edu.hk}

\thanks{\textit{Funding:} K.~Chen is supported by the HKRGC grants (project numbers: 15305422, 15224923) and by the Research Centre for Quantitative Finance at The Hong Kong Polytechnic University (P0042708).
    K.~Park acknowledges financial support by the Presidential Postdoctoral Fellowship of Nanyang Technological University and the National Research Foundation of Korea (grant DOI: RS-2025-02633175). H.Y.~Wong acknowledges support by HKRGC grants (project number: 14305823).}
\date{\today.}

\begin{document}

\begin{abstract}
	In a continuous-time economy, this paper formulates the Epstein-Zin preference for discounted dividends received by an investor as an Epstein-Zin singular control utility.  We introduce a backward stochastic differential equation with an aggregator integrated with respect to a singular control, prove its well-posedness, and show that it coincides with the Epstein-Zin singular control utility. We then establish that this formulation is equivalent to a robust dividend policy chosen by the firm’s executive under the Maenhout’s ambiguity-averse preference.  In particular, the robust dividend policy takes the form of a threshold strategy on the firm’s surplus process, where the threshold level is characterized as the free boundary of a Hamilton–Jacobi–Bellman variational inequality. Therefore, dividend-caring investors can choose firms that match their preferences by examining stock's dividend policies and financial statements, whereas executives can make use of dividend to signal their confidence, in the form of ambiguity aversion, on realizing the earnings implied by their financial statements.
	%\medskip
	%\noindent {\it Key words.}  
	%
	%\medskip
	%\noindent {\it JEL Classification.} C61 · G35 · G33
%Primary 90C31, 60G65, 60H05; secondary 91G10
\end{abstract}

%\vspace{-3.em}
\maketitle

%\vspace{-1.em}
%{
%	\hypersetup{linkcolor=black}
%	\tableofcontents
%}

%\vspace{-2.5em}
\section{Introduction}
        Dividend (or cash payout) policy is an important topic in both corporate finance and asset pricing. The Miller-Modigliani~\cite{Miller1961} dividend irrelevance theory suggests that a firm's dividend policy does not affect its value and stock prices in a perfect market. However, the underlying incentives for stockholders to receive dividends and for the firm's executives to pay dividends are largely unknown in reality. For instance, a stockholder can liquidate some of his stocks for consumption instead of receiving cash dividend, which is subject to a higher tax rate in the USA. Hence,  Black \cite{Black1976} proposes the dividend puzzle which has generated a great of attentions in the literature. 

        A school of thought is the dividend signaling theory (DST) (see Bhattacharya~\cite{Bhattacharya1979}, John and Williams~\cite{John1985}, Miller and Rock~\cite{Miller1985}) that a firm's executive passes signals about the firm's information to the public through dividends. In particular, Black \cite{Black1986} wrote,
    \begingroup
    \addtolength\leftmargini{-1.5em}
    \begin{quote}
    {\it 
        ``The idea that dividends convey information beyond that conveyed by the firm’s financial
    statements and public announcements stretches the imagination.... I~think we must assume that investors care about dividends directly. We must put dividends
    into the utility functions.”}
    \end{quote}
    \endgroup

Although most empirical studies support DST, the
survey in Brav et al.~\cite{Brav2005} shows that executives viewed the mechanism behind DST as broadly misguided. However, dividend changes do contain information about future earnings according to Ham et al.~\cite{Ham2020}. Additionally, Michaely et al.~\cite{Michaely2021} shows that payout announcements signal changes in cash flow volatility (in the opposite direction) but not stock volatility, which the authors refer to as ``signaling safety''. A firm's cashflow is related to its earnings so that the cashflow volatility indicates the earnings' volatility. By modeling the firm's uncontrolled surplus process as diffusion models with a stochastic drift, there is a mathematical toolkit as in Reppen et al.~\cite{Reppen2020} to compute the optimal dividend policy under a random profitability. 

In this paper, we consider the suggestions by  Black \cite{Black1986}. We assume that the firm's financial reports contain information on the firm's expected earnings and risk. Thus, we model the firm's uncontrolled surplus by an diffusion process and put dividends into the utility functions for investors. Investors who care directly about dividends have a value function that considers the short-term cash inflow from discounted cash dividends and the certainty equivalence of long-term stock holdings subject to bankruptcy risk. This consideration requires us to separate risk aversion from elasticity of intertemporal substitutions (EIS) in the decision making process. 

The preference introduced by Epstein and Zin (EZ)~\cite{EpsteinZin1989} is a natural candidate designed for this purpose. 
In our formulation, we increase the degree of risk aversion of intertemporal preference ordering via a parameter~$R\in(0,1)$, without affecting ``certainty preferences'' of a dividend payment policy. We also follow the continuous-time dividend models where the discounted dividend stream is paid in a singular control manner (see Cadenillas et al.~\cite{Cadenillas2007MF}, Jeanblanc-Picqu{\'e} and Shiryaev~\cite{jeanblanc1995optimization}, Reppen et al.~\cite{Reppen2020}, Shreve et al.~\cite{shreve1984optimal}). This formulation turns out to be a singular control problem with recursive utilities for which we have to establish % the well-definedness 
that there exists a unique value function for the EZ investor within the recursion in singular control setting.

Alternatively, the firm's executive aims to maximize the expected discounted dividend stream subject to the bankruptcy risk. Although the financial reports indicate the firm's expected earnings and risk, the executive encounters {uncertainty} in the predicted earnings. We view this uncertainty as {\it ambiguity} so that the executive's value function is postulated as the Maenhout's \cite{maenhout2004robust} robust preference. Specifically, the executive %specifies ambiguity aversion on the expected earnings so that 
aims to make a robust dividend decision %is made within a maximin singular control problem 
subject to bankruptcy risk and Maenhout's regularity.

\color{black}
Our first aim is to establish the equivalence between the EZ singular control utility and Maenhout's counterpart %when the parameter $R$ of the former equals the ambiguity aversion parameter $\cR$ of the latter 
under the setting that the risk-aversion parameter in the EZ utility coincides with the ambiguity-aversion parameter in Maenhout’s preference (see Theorem~\ref{thm:rob_rec}).  While the connection between EZ and Maenhout preferences is initially addressed in \cite{maenhout2004robust} and formally established in Skiadas~\cite{skiadas2003robust} for consumption models, this paper is, to the best of our knowledge, the first to develop this equivalence within a {\it singular control setting}, even accounting for {\it the bankruptcy time}.
%Unlike the regular control situation discussed in \cite{maenhout2004robust}, our BSDE framework reveals that the singular control problem with Maenhout preferences cannot be reduced to a classic utility maximization with singular control. 

To establish the equivalence, a substantial portion of our efforts revolves around showing that the EZ singular control utility is well-defined. In Proposition \ref{pro:recurs_BSDE}, we prove the existence and uniqueness of the EZ singular control utility via backward stochastic differential equations (BSDEs). The analysis is %technically challenging 
nontrivial 
due to the interaction of a non-decreasing control with a non-Lipschitz aggregator, together with the inclusion of the bankruptcy time. %In order to establish the equivalence, a substantial portion of our efforts revolves around showing that the EZ singular control utility is well-defined. In particular, 
%This paper makes three theoretical contributions. The first contribution is to prove that the EZ singular control utility is well-defined. 
%We 
%we establish the existence and uniqueness of the value function satisfying the recursion using backward stochastic differential equations (BSDEs) in Proposition~\ref{pro:recurs_BSDE}. The corresponding proof is non-trivial due to the presence of a non-decreasing process (i.e., singular control) within a non-Lipschitz aggregator, combined with the bankruptcy time. 
In contrast to the EZ consumption literature (see, e.g., %Compared to the existing literature on the EZ consumption model (see, e.g., 
Aurand and Huang~\cite{Aurand2023}, Xing~\cite{xing2017consumption}), our framework explicitly incorporates the risk of ruin rather than solely considering on cash flow utility, %.
%The consideration of bankruptcy time captures the ruin risk, 
which results in a nonconvex set of dividend strategies. %that remain constant after ruin. 
This nonconvexity implies the lack of concavity in our EZ singular control utility. %, distinguishing it from typical EZ consumption models.
Moreover, we focus on the case $R\in (0,1)$, where standard a priori bounds for BSDEs suitable for the case 
$R\in(1,\infty)$ cannot be directly used (see, e.g., \cite{Aurand2023,xing2017consumption}, Popier~\cite{popier2007backward}).

Our next aim is to develop a shooting method in Proposition \ref{pro:shooting} to show that under a Markovian setting and certain regularities on the reference parameters of the surplus process, the robust dividend policy is a threshold strategy (see Theorem \ref{thm:verification}). When the firm's surplus hits the upside threshold, the dividend is paid. %The threshold depends on the expected earnings, volatility and the ambiguity aversion of the executive.%, which is nontrivial to be characterized.
%However, the development of this result is highly non-trivial. ]
When there is no bankruptcy, the shooting method has been used to a Hamilton-Jacobi-Bellman (HJB) variational inequality (VI) (see Cohen et al.~\cite{cohen2022optimal}). Our Maenhout's singular control problem with bankruptcy can be transformed into a HJB-VI. % subject to the bankruptcy boundary condition. 
%To characterize the threshold, we show that the Maenhout's singular control problem with bankruptcy can be transformed into a Hamilton-Jacobi-Bellman (HJB) variational inequality (VI) subject to the bankruptcy boundary condition. The nonlinear HJB differential operator appears in the VI formulation because of the Maehout's preference. %the robust control problem considers worst-case scenarios.  
%When there is no bankruptcy, the shooting method is shown useful in solving the HJB-VI in \cite{cohen2022optimal}. % through searching for a suitable threshold to fit the condition of the HJB-VI. 
However, the bankruptcy time complicates the analysis because it is involved within the dividend decision. In particular, we have to prove that there is a threshold on the surplus process which shoots two targets on the bankruptcy boundary condition and HJB-VI, simultaneously. 
\color{black}

Our theoretical results lead to the following interpretations for the dividend signaling theory: 
%Although investors are informed about the firm's expected earnings and risk from its financial statements, 
 While investors are informed with the firm's expected earnings from the firm's financial statements, executives set dividend threshold to showcase their confidence on realizing the expected earnings.  The level of confidence is reflected by the ambiguity aversion parameter in our model. 
 %For a positive excess expected earning, a higher ambiguity aversion results in a lower the firm's dividend value function. %Therefore, an investor who accepts an ambiguity aversion of the firm's executive, which is translated into the investor's EIS, also accepts stocks that the induced ambiguity aversion from executives are smaller for other things being fixed.\comment{I couldn't understand this line well.} 
Therefore, our model predicts that dividend policy conveys information about earnings uncertainty or, equivalently, the executive's confidence in realizing the expected earnings. We call this notion ``signaling confidence''.

We now proceed to a discussion of some related literature. In the context of regular stochastic control, most of the fundamental and technically demanding questions in recursive utility maximization problems are understood fairly well by now (see, e.g., Duffie and Epstein~\cite{DuffieEpstein1992}, Herdegen et al.~\cite{herdegen2021proper,herdegen2023infinite,herdegen2023infinite2}, Kraft et al.~\cite{kraft2017optimal}, Matoussi and Xing~\cite{matoussi2018convex}, Melnyk et al.~\cite{melnyk2020lifetime}, Monoyios and Mostovyi~\cite{monoyios2024stability}, Schroder and Skiadas~\cite{Schroder1999}, Xing~\cite{xing2017consumption}). While the dominant mathematical framework therein is the BSDE approach, it does not involve a non-decreasing process in the aggregator. In particular, our choice of singular control and random terminal time makes certain BSDEs arguments intricate.

On the other hand, considerable efforts have been made to characterize classical %(i.e., {\it ambiguity-neutral}) 
singular optimal control. %from both dynamic programming and probabilistic perspectives. 
Specifically, we refer to e.g. De Angelis et al.~\cite{de2017optimal}, Ferrari~\cite{ferrari2015integral}, Fleming and Soner~\cite{fleming2006controlled}, Shreve and Soner~\cite{shreve1991free}, Soner et al.~\cite{soner1991free} for corresponding free boundary problems, and to e.g. Bank~\cite{bank2005optimal}, Bank and El Karoui~\cite{bank2004stochastic}, Bank and Kauppila~\cite{bank2017convex},  Bank and Riedel~\cite{bank2001optimal} for stochastic representation theorem approach. Recently, there has been an intensive interest in robust analogue of the above references (e.g., Chakraborty et al.~\cite{chakraborty2023optimal}, Cohen et al.~\cite{cohen2022optimal}, Ferrari et al.~\cite{ferrari2022optimal,ferrari2022knightian}, Park et al.~\cite{park2023irreversible}). We also refer to Bayraktar and Huang~\cite{bayraktar2013multidimensional}, Bayraktar and Yao~\cite{bayraktar2011optimal_a,bayraktar2011optimal_b}, Nutz and Zhang~\cite{nutz2015optimal}, Park and Wong~\cite{PW23}, Park et al.~\cite{park2023robust}, Riedel~\cite{riedel2009optimal} 
for relations to optimal stopping time under ambiguity. We point out that the Maenhout preference, which is a form of relative entropy deriving ambiguity-aversion (see, e.g., Anderson et al.~\cite{anderson2003quartet}, Ben-Tal~\cite{ben1985entropic}, Csisz{\'a}r~\cite{csiszar1975divergence}, Laeven and Stadje~\cite{laeven2014robust}), has been utilized in robust optimization problems in finance and economics (see, e.g., Branger and Larsen~\cite{branger2013robust}, Jin et al.~\cite{jin2018dynamic}, Maenhout~\cite{maenhout2006robust}, Yi et al.~\cite{yi2013robust}). Finally, %for completeness, 
we refer to the {forward utility} preference (see, e.g., El Karoui and Mrad~\cite{ElKaroui2013SIFIN}, Liang and Zariphopoulou~\cite{Liang2017SIFIN}, Musiela and Zariphopoulou~\cite{Musiela2009QF,Musiela2010SIFIN,Musiela2011IJTAF}) as an alternative %important 
framework for dynamic utility preferences. 

    \color{black}
    
    The paper is organized as follows. Section \ref{sec:SDU}
    introduces the EZ singular control utility, %under admissible dividend policies with bankruptcy risk, 
    presents its BSDE characterization, and establishes well-posedness. Section~\ref{sec:3} formulates Maenhout’s robust dividend-value utility and proves its well posedness. Section~\ref{sec:main} establishes the equivalence between the EZ singular control utility and Maenhout’s robust dividend-value utility, and characterizes the robust dividend policy using HJB-VI. Section~\ref{sec:5} studies the sensitivity of Maenhout’s robust dividend-value utility with respect to the ambiguity-aversion parameter. All proofs are presented in Section \ref{sec:proof}.
    \color{black}

\subsection{Notation and preliminaries}\label{sec:notation}
    Consider a probability space $(\Omega,\mathcal{F},\mathbb{P})$ that supports a one-dimensional Brownian motion $W:=(W_t)_{t\ge0}$. 
	Let $\mathbb{F}:=(\mathcal{F}_{t})_{t\ge 0}$ be the augmented filtration generated by $W$. We assume that ${\cal F}_0$ is trivial. We write $\mathbb{E}[\cdot]$ for the expectation under $\mathbb{P}$ and $\mathbb{E}_t[\cdot]:=\mathbb{E}[\cdot|{\cal F}_t]$ for its conditional expectation given ${\cal F}_t$.%,~$t\geq 0$.  
    
    $\mathcal{P}$ is the set of all real-valued $\mathbb{F}$-progressively measurable processes, and ${\cal P}_+\subset\mathcal{P}$ is the subset of processes that take values in $\mathbb{R}_+:=[0,+\infty)$. Moreover, we denote for every $p\ge1$ and $t\geq0$ by $L^p(\mathcal{F}_t)$ the set of all $\mathcal{F}_t$-measurable random variable with norm $\|X\|_{L^p}^p:=\mathbb{E}[|X|^p]<\infty$. 
    %The following notations are useful in the rest of the paper. %\mycomment{Let us replace ${\cal A}\to{\cal D}$ to define optimization over ${\cal A}$.} 
    We introduce some notations for defining stochastic differential utility under singular dividend flows.
    \begin{enumerate}
        \item [$\bullet$] ${\cal T}$ is the set of all $\mathbb{P}$-almost\;surely ($\mathbb{P}$-a.s.)~finite, $\mathbb{F}$-stopping times $\tau$.
        \item [$\bullet$] $\mathcal{A}$ is the set of all $\mathbb{F}$-progressively measurable, 
        $\mathbb{P}$-a.s.\;continuous, and non-decreasing processes $D:=(D_t)_{t\geq 0}$ with $D_{0-}=0$. Here we write $D_{0-} = 0$ for indicating that $D_0>0$ can only be achieved by a jump of the process at time zero, followed by a continuous path for all $t\ge0$.
    \end{enumerate}
    % \begin{align*}
    %     % \mathcal{T}:=&\{\tau : \tau \text{ is a $\mathbb{P}$-almost\;surely ($\mathbb{P}$-a.s.)~finite $\mathbb{F}$-stopping time}\},\nonumber\\
    %     \mathcal{A}:=&\{D:=(D_t)_{t\geq0}:D \text{ is $\mathbb{F}$-progressively measurable, continuous $\mathbb{P}$-a.s.,} \nonumber\\
    %     &\qquad\qquad\quad\quad\quad \text{and non-decreasing with $D_{0-}=0$}\},
    % \end{align*}

    We define the stochastic integral with respect to (w.r.t.) $D\in {\cal A}$ as follows: for any $\mathbb{F}$-progressively measurable and locally bounded process $(g_t)_{t\ge 0}$, the stochastic integral, 
    \(
        \int_{0}^{\tau} g_{t} dD_t:= g_0D_0+\int_{(0,\tau]}g_{t}dD_t,
    \)
    is well-defined in the Stieltjes sense (see Protter~\cite[Theorem~IV.15]{protter2003}). Here, the jump of $D$ at time zero is accounted, so that $\int_0^{t}dD_s = D_t$ for $t\ge0$.

    %Then we lay out the setup of a stochastic differential utility (SDU) on singular dividend flows. 

    Finally, we introduce a set of processes that serve as Girsanov kernels, which will be used to define Maenhout’s ambiguity aversion in Section \ref{sec:3}.
    \begin{itemize}
        \item [$\bullet$] $\Theta $ is the set of all processes $\theta:=(\theta_t)_{t\geq 0}\in {\cal P}$ under which a stochastic exponential $\eta^\theta:=(\eta^{\theta}_t)_{t\ge0}$ defined as
		\begin{align}\label{eq:eta}
			\eta_t^{\theta} = \exp\bigg(\int_{0}^t \theta_s dW_s - \frac{1}{2}\int_0^t \theta_s^2 ds\bigg), \quad  t\ge0,
		\end{align}
		is a martingale. We call $\Theta$ the Girsanov kernels set. 
    \end{itemize}
    
 	For each $\theta \in \Theta$, we define a probability measure $\mathbb{Q}^{\theta}$ on $(\Omega,\mathcal{F},\mathbb{F})$ by 
 	\begin{align}\label{eq:measureQ}
 		\frac{d\mathbb{Q}^{\theta}}{d\mathbb{P}}\bigg\vert_{\mathcal{F}_T} = \eta^{\theta}_T\quad \mbox{for each $T>0$.}
 	\end{align}
 	Then $W_t^{\theta}: = W_t - \int_0^t \theta_s ds$, $t\in[0,T]$,	
    % \begin{align}\label{eq:QW}
 	% 	W_t^{\theta}: = W_t - \int_0^t \theta_s ds, \quad t\in[0,T],
 	% \end{align}
 	is a $\mathbb{Q}^{\theta}$-Brownian motion. We write $\mathbb{E}^\theta[\cdot]$ for the expectation under $\mathbb{Q}^\theta$ and $\mathbb{E}^\theta_t[\cdot]$ for its conditional expectation given~$\mathcal{F}_t$. %It follows from Skiadas~\cite[Propositions 2\;\&\;4]{skiadas2003robust} that for any $\tau \in {\cal T}$, the discounted relative entropy process, given as $\frac{1}{2}\mathbb{E}^{\theta}_t[ \int_{t\wedge\tau}^{\tau}  e^{-\rho  (s-t)} \theta_s^2 ds]$, $t>0$, is well-defined and takes values in $[0,\infty]$ $\mathbb{P}$-a.s.

\section{Main results}\label{sec:2}
%%%%%%%%%%%%%%%%%%%%%%%%%%%%%%%%%%%%%%%%%%%%%%%%%%%%%%%%%%%%%%%%%%%%%%%%%%%%%%%%
    \subsection{Epstein-Zin utility under singular dividend flows}\label{sec:SDU}
    Following the suggestion of Black~\cite{Black1976}, we incorporate dividends into the investor’s objective function through recursive utilities.
    
    We start by defining the surplus process and its controlled version under singular dividend flows. We let $\mathcal{X}:= (\ell,+\infty)$, $-\infty \le \ell < 0$, denote the state space. For each $x\in \mathcal{X}$, we denote by $X^x:=(X^x_t)_{t\ge 0}$ the surplus process governed by the $\mathcal{X}$-valued stochastic differential equation~(SDE)
	\begin{align}\label{eq:uncontrolX}
		X^x_t = x + \int_{0}^{t} \mu(X^x_s) ds + \int_{0}^{t}\sigma(X^x_s) dW_t, \quad t\geq 0,
	\end{align}
    where the drift $\mu:\mathcal{X}\rightarrow \mathbb{R}$ and volatility  $\sigma:\mathcal{X}\rightarrow \mathbb{R}\backslash\{0\}$ functions are Borel measurable. %The value $\mu(X_t^x)/X_t^x$ is the expected per surplus growth rate and $\sigma(X_t^x)/X^x_t$ stands for infinitesimal volatility of fluctuations in the per surplus growth rate.

    We impose the following condition on the uncontrolled surplus process~\eqref{eq:uncontrolX}.
    \begin{as}\label{as:recurrent}
            The SDE \eqref{eq:uncontrolX} admits a unique weak solution in law, which is a regular diffusion with unattainable boundaries $\{\ell, +\infty\}$; see, e.g., Karatzas and Shreve~\cite[Chapter 5.5]{KS1991}, Revuz and Yor~\cite[Chapter VII.3]{revuz2013continuous}.
    \end{as}

    \begin{rem}\label{rem:recurrent}
        Assumption \ref{as:recurrent} implies that \(X^x\) in \eqref{eq:uncontrolX} does not explode at infinity and is recurrent (see \cite[Proposition 5.22]{KS1991}, \cite[Chapter VII, Exercise~3.22]{revuz2013continuous}). Consequently, bankruptcy occurs as intended in our model (cf. Definition \ref{dfn:bankruptcy}). An example satisfying this assumption is provided in Remark~\ref{rem:OU}
    \end{rem}

    We recall the set ${\cal A}$ defined in Section \ref{sec:notation}. 
    For each $x\in \mathbb{R}_+\subset {\cal X}$ and $D\in\mathcal{A}$, we let $X^{x,D}:=(X_t^{x,D})_{t\ge 0}$ denote the controlled surplus process governed by
    \begin{align}\label{eq:controlX}
    	X_t^{x,D} = x + \int_{0}^{t} \mu(X_{s}^{x,D}) ds + \int_{0}^{t}\sigma(X_{s}^{x,D}) dW_t - D_t, \quad t\geq 0,
    \end{align}
    with its left limit $X_{0-}^{x,D} =x$ at time zero.

    We then introduce the bankruptcy risk of the investor under the controlled surplus process and the set of singular controls for admissible dividend flows.
    \begin{dfn}\label{dfn:bankruptcy}
        We define the bankruptcy (ruin) time~by
    \begin{align*}%\label{def:bankrupcytime}
 			\tau^{x,D} := \inf\{t\ge0: X_{t}^{x,D} \le 0  \}.
 		\end{align*} 
    Under Assumption \ref{as:recurrent},% implies the recurrence of the process \eqref{eq:uncontrolX}, 
    $\tau^{x,D}$ is finite $\mathbb{P}$-a.s. (see Remark \ref{rem:recurrent}); hence $\tau^{x,D}\in {\cal T}$.
    
    Moreover, we let $\xi_{\tau^{x,D}}$ denote the bankruptcy payout satisfying
    \begin{align}\label{def:bankrupcypayment}
        \xi_{\tau^{x,D}}\in L^{2}({\cal F}_{\tau^{x,D}})\quad\mbox{and}\quad \mbox{$\xi_{\tau^{x,D}}>0$ $\mathbb{P}$-a.s..}
    \end{align}
    \end{dfn}

    %To formulate a well-defined utility process for the investor, 
    %We consider the following set of singular control processes for admissible dividend flows. 
    \begin{dfn}\label{dfn:admissible} %Recall the set $\Theta$ defined in Section~\ref{sec:notation}.
    For any $x\in \mathbb{R}_+$, let $\mathcal{A}^x $ denote the set of all $D \in \mathcal{A}$ satisfying  
            \(
            \int_{0}^{\tau^{x,D}} e^{-\rho s} dD_s \in L^2(\mathcal{F}_{\tau^{x,D}})\) and \(D_t - D_{t-} \le X_{t-}^{x,D}\)  for $t \ge 0$, where $\rho>0$ is a constant discount rate. 
     %    \begin{itemize}
     %        \item [i.] %which represent accumulated dividend flow processes.
     %        \item [ii.] For any $(x,D)\in \mathbb{R}_+\times{\cal A}^x$, we denote by $\Theta^{x,D}$ the set of $\theta\in \Theta$~satisfying  
     % \begin{align}
     % &\mathbb{E}\bigg[\bigg(\int_{0}^{\tau^{x,D}} \eta_{t}^{\theta}e^{\int_{0}^{t} \frac{\theta^2_s}{2\cR}ds -\rho t} dD_t\bigg)^2\bigg]<\infty\quad \mbox{and}\nonumber\\
     % &\mathbb{E}\bigg[(\eta_{\tau^{x,D}}^{\theta})^2e^{\int_{0}^{\tau^{x,D}} \frac{\theta^2_u}{\cR}du } \max\{\xi_{\tau^{x,D}}^{2},1\} \bigg]<\infty,\label{eq:ThetaD}
     % \end{align}
     % where $\eta^{\theta}$ is defined in \eqref{eq:eta} and $\rho>0$ is a constant discount rate. %and \color{red}$\xi_{\tau^{x,D}}\in L^{2}({\cal F}_{\tau^{x,D}})$ is a bankruptcy payout such that $\xi_{\tau^{x,D}}>0$ $\mathbb{P}$-a.s.. \color{black}
     %    \end{itemize}
    \end{dfn}

    Having completed the investor's dividend model, we now introduce her utility under Epstein-Zin~(EZ) preferences. 
    
    Throughout this section, we fix a parameter $R\in (0,1)$, which corresponds to a {\it degree of risk-aversion}. Following Epstein and Zin~\cite{EpsteinZin1989}, we introduce the EZ aggregator $g_{\text{EZ}}:\mathbb{R}_+\times\mathbb{R}_+\to \mathbb{R}_+$ defined by
    	\begin{align}\label{dfn:g_EZ}
    		g_{\text{EZ}}(s,v): = e^{-\rho s}(1-R)v^{\frac{-R}{1-R}},\quad (s,v)\in \mathbb{R}_+\times\mathbb{R}_+.
    \end{align}
    %appearing in \eqref{eq:ThetaD} and \eqref{eq:object_J}.

    %Throughout this section, we fix a parameter \(\cR \in (0,1)\), which corresponds to a {\it degree of ambiguity-aversion}. In the presence of ambiguity, the executive prefers a family of unspecified alternative models that are in close proximity to the reference model given in \eqref{eq:uncontrolX}. In line with this, we define the executive’s admissible sets of singular dividend flows and of Girsanov kernels as follow. 

    For each $x\in \mathbb{R}_+$ and $D\in {\cal A}^x$ (see Definition \ref{dfn:admissible}), %we recall that $\tau^{x,D}\in {\cal T}$ is the bankruptcy time in \eqref{def:bankrupcytime} and $\xi_{\tau^{x,D}}$ is the bankruptcy payout in~\eqref{def:bankrupcypayment}.  
    we define the investor's EZ utility \(V^{x,D}:= (V^{x,D}_t)_{t \geq 0}\) by %as follows: for $t\geq 0$
    	\begin{align}\label{eq:EZ_V}
    		V^{x,D}_t = \mathbb{E}_t \bigg[ \int_{t\wedge \tau^{x,D}}^{\tau^{x,D}} %e^{-\rho s}(1-R) 	(V^{x,D}_s)^{\frac{-R}{1-R}}
            g_{\text{EZ}}(s,V^{x,D}_s)
            dD_s + \xi_{\tau^{x,D}}^{1-R}\bigg],\quad t\ge 0.
    	\end{align}

    We note that when $R=0$, $V^{x,D}$ in \eqref{eq:EZ_V} reduces to the classical dividend utility $K^{x,D}:=(K^{x,D}_t )_{t\geq 0}$ (see Alvarez and Virtanen~\cite{alvarez2006class}) given by% follows: for $t\ge 0$
    \begin{align}\label{eq:std_singular_V}
		K^{x,D}_t := \mathbb{E}_{t}\bigg[\int_{t\wedge \tau^{x,D}}^{\tau^{x,D}} e^{-\rho s} dD_s +\xi_{\tau^{x,D}} \bigg],\quad t\geq 0.
	\end{align}
    %, the dividend-value process $K^D:=(K^D_t)_{t\ge0}$ associated with $D$ is defined as
% \footnote{For an $\mathbb{F}$-progressively measurable and locally bounded process $(g_t)_{t\ge 0}$, we define its stochastic integral with respect to (w.r.t.) the singular control $D\in {\cal A}$ by 
% \[
%     \int_{0}^{\tau} g_{t} dD_t:= g_0D_0+\int_{(0,\tau]}g_{t}dD_t,
% \]
% which is well-defined in the Stieltjes sense; cf. \cite[Theorem~IV.15]{protter2003}. Here we note that the jump of $D$ at time zero is accounted, so that $\int_0^{t}dD_s = D_t$ for $t\ge0$.} 
    In the setting \eqref{eq:EZ_V}, the degree of risk aversion of intertemporal preference ordering is increased (i.e., $R>0$) without affecting the ``certainty preferences''.

    The following definition is inspired by the concept of stochastic differential utility (see Duffie and Epstein~\cite{DuffieEpstein1992}, Epstein and Zin~\cite{EpsteinZin1989}, Schroder and Skiadas~\cite{Schroder1999}), while tailored for our singular dividend flow setting.
    \begin{dfn}\label{dfn:recurs_V} 
		%Given $(\tau,D)\in\mathcal{T}\times\mathcal{A}$ and $\xi_\tau \in L^1({\cal F}_\tau)$, denote by 
        For any $x\in \mathbb{R}_+$ and $D\in {\cal A}^x$, %Let $\tau^{x,D}$ be defined as in~\eqref{def:bankrupcytime}. 
        we denote by $\mathbb{I}{(g_{\EZ},D,{\tau^{x,D}})}$ the set of all $V\in {\cal P}$ for which the stochastic integral $\int g_{\EZ}(t,V_t) dD_t$ is well-defined and the integrability condition $\mathbb{E}[\int_0^{\tau^{x,D}}|g_{\EZ}(s,V_s)|dD_s]<\infty$ holds. 
        \begin{enumerate}[label=\roman*.]
            \item A process $V\in \mathbb{I}{(g_{\EZ},D,{\tau^{x,D}})}$ is called a utility with $(g_{\EZ},D,{\tau^{x,D}},\xi^{1-R}_{\tau^{x,D}})$ if it satisfies~\eqref{eq:EZ_V};
            \item We denote by $\mathbb{UI}(g_{\EZ},D,{\tau^{x,D}},\xi^{1-R}_{\tau^{x,D}})$ the set of all uniformly integrable utilities $V$ with $(g_{\EZ},D,{\tau^{x,D}},\xi^{1-R}_{\tau^{x,D}})$.
        \end{enumerate}
	\end{dfn}
    %implying the dividend-value utility at time $t$ may depend in a nonlinear way on its value at future~times. 

    \begin{rem}
        The utility specification in \eqref{eq:EZ_V} is motivated by its connection to the conventional EZ aggregator in the infinite elasticity of intertemporal substitution limit.
    %in the infinite EIS limit.\mycomment{What is the infinite EIS limit..?} 
    While \(g_{\EZ}\) in \eqref{dfn:g_EZ} differs from the standard form by a scaling factor \((1-R)\) (rather than \((1-R)^{-\frac{R}{1-R}}\)), this does not affect preferences for any \(R \in (0,1)\)
    (see Herdegen et al.~\cite[Remark 2.4\;(a)]{herdegen2023infinite}). 
    Our scaling choice in \eqref{eq:EZ_V} clarifies the equivalence between EZ utility and robust Maenhout utility (see Theorem~\ref{thm:rob_rec}).
    For \( R \in (1,\infty) \), the sign and domain of the generator become more delicate; see, e.g., Aurand and Huang~\cite{Aurand2023}, Herdegen et al.~\cite{herdegen2023infinite,herdegen2023infinite2}, Xing~\cite{xing2017consumption}, for further discussion.
    \end{rem}

    We now establish the well-posedness (i.e., existence and uniqueness) of the EZ utility.  In particular, for each $x\in \mathbb{R}_+$ and $D\in {\cal A}^x$, $V^{x,D}$ in~\eqref{eq:EZ_V} will be characterized as the unique solution to the following backward stochastic differential equation (BSDE): for $t \geq 0$
    \begin{align}\label{dfn:BSDE1}
        V^{x,D}_t =\xi_{\tau^{x,D}}^{1-R} + \int_{t \wedge \tau^{x,D}}^{\tau^{x,D}} g_{\text{EZ}}(s, V^{x,D}_s) \, dD_s - \int_{t \wedge \tau^{x,D}}^{\tau^{x,D}} Z^{x,D}_s \, dW_s.
    \end{align}

    \begin{pro}%[Existence and uniqueness of EZ SSDU] 
    \label{pro:recurs_BSDE}
    Suppose that Assumption \ref{as:recurrent} holds. 
    %Let $(\tau, D, \xi_{\tau})$ satisfy Condition~\ref{cond:triplet}
    For any $x\in \mathbb{R}_+$ and $D\in {\cal A}^x$, %Let $g_{\mathrm{EZ}}$ be defined by~\eqref{dfn:g_EZ}. 
    there exists a unique utility $V^{x,D} \in \mathbb{UI}(g_{\EZ},D,{\tau^{x,D}},\xi^{1-R}_{\tau^{x,D}})$, which is strictly positive, has continuous paths, and satisfies $\mathbb{E}[\sup_{t \geq 0} (V^{x,D}_t)^2] < \infty$. %(see Definition \ref{dfn:recurs_V}). 
    Moreover, there exists $Z^{x,D} \in \mathcal{P}$ such that $\int_0^{\tau^{x,D}} (Z_t^{x,D})^2 \, dt < \infty$ $\mathbb{P}$-a.s., and $(V^{x,D}, Z^{x,D})$ is a solution of \eqref{dfn:BSDE1} in the sense of Definition \ref{def:BSDEdef}.
    \end{pro}

    \color{black}
    The proof of Proposition \ref{pro:recurs_BSDE} can be found in Section \ref{proof:pro:recurs_BSDE}. %In the following, we briefly describe the key points of the proof. The randomness on $\tau^{x,D}$, the nonlinearity of $g_{\EZ}$, and unboundedness on $D$ make the analysis on the BSDE \eqref{eq:} 
    We outline the main steps as follows. First, we establish the existence of a solution to \eqref{dfn:BSDE1} under a bounded dividend flow and a Lipschitz aggregator. In the second step, we prove a comparison theorem under aggregators $g(\cdot,\cdot)$ that are nonincreasing in their second argument, which yields uniqueness, and we derive a priori bounds for solutions to \eqref{dfn:BSDE1} under $g_\EZ(\cdot,\cdot)$ which is non-Lipschitz. Using these bounds, in the third step we approximate $g_\EZ(\cdot,\cdot)$ by a sequence of Lipschitz aggregators and show convergence of the corresponding BSDE solutions. Finally, for unbounded dividend flows \(D\in{\cal A}^x\), we follow some arguments in Aurand and Huang~\cite{Aurand2023} to complete the proof. \color{black}
    
    By Proposition \ref{pro:recurs_BSDE}, we are able to define the investor's dividend optimization problem under EZ preferences by
    \begin{align}
    \label{eq:opt_EZ}
    V_t^{x} := \esssup\{V^{x,D}_t: D \in {\mathcal{A}}^x\},\quad (t,x)\in \mathbb{R}_+\times \mathbb{R}_+.
    \end{align}

    Solving \eqref{eq:opt_EZ} is demanding due to the recursive nature of EZ preference, the random time horizon and singular controls. In the next section, we introduce a robust dividend formulation under ambiguity, which will later be shown to be equivalent to \eqref{eq:opt_EZ}. The robust formulation not only provides a convenient basis for characterizing the optimal dividend policy under suitable assumptions, but extends the previous findings in Maenout~\cite{maenhout2004robust} and Skiadas~\cite{skiadas2003robust}.

    \subsection{Robust singular dividend criterion under Maenhout preferences}\label{sec:3}
    %In this section, we stand at the perspective of the firm's executives. We consider that the executives (the agent) do not work on a particular utility function but aim to maximize the expected discounted dividend payout subject to the bankruptcy cost. However, executives have insider information on the possibility to realize the earnings implied by the firm's financial statements. Actually, executives encounter uncertainty on earnings or on the drift estimate of the diffusion model for the firm's surplus. We aim to link an investor's EZ singular control utility to the robust dividend policy of the firm's executives, establishing an equilibrium condition that both parties are satisfied with the dividend policy. 
    In this section, we take the perspective of the firm's executive, who aims to maximize expected discounted dividends while considering bankruptcy risk (see Definition \ref{dfn:bankruptcy}). At the same time, she suffers from ambiguity regarding the firm's earnings or the drift of the firm's surplus given in \eqref{eq:uncontrolX}.

    Throughout this section, we fix a parameter \(\cR \in (0,1)\), which corresponds to a {\it degree of ambiguity-aversion}. In the presence of ambiguity, the executive prefers a family of unspecified alternative models that are in close proximity to the reference model given in \eqref{eq:uncontrolX}. In line with this, we define the executive’s admissible sets of Girsanov kernels as follow.    
    %To formulate a well-defined robust singular dividend criterion for the executive, we introduce the following admissible set of Girsanov kernels. 
   % We recall that the set $\Theta$ defined in Section \ref{sec:notation}, and for each $x\in \mathbb{R}_+$ the set ${\cal A}^x$ is defined in Definition~\ref{dfn:admissible}.
    \begin{dfn}\label{dfn:admissD}
        For any initial wealth $x\in \mathbb{R}_+$ and consumption plan $D\in {\cal A}^x$,
we define the set of admissible Girsanov kernels $\Theta^{D} \subseteq \Theta$ by
\begin{equation} \label{eq:ThetaD}
\Theta^{D} := \left\{ \theta\in \Theta \;\middle|\;
\begin{array}{l}
  \mathbb{E}\bigg[\bigg(\displaystyle\int_{0}^{\tau^{x,D}} \eta_{t}^{\theta}e^{\int_{0}^{t} \frac{\theta^2_s}{2\cR}ds -\rho t} dD_t\bigg)^2\bigg]<\infty, \quad \text{and} \\[1.5em]
  \mathbb{E}\left[(\eta_{\tau^{x,D}}^{\theta})^2e^{\int_{0}^{\tau^{x,D}} \frac{\theta^2_u}{\cR}du } \max\{\xi_{\tau^{x,D}}^{2},1\} \right]<\infty
\end{array}
\right\}.
\end{equation}
 where $\eta^{\theta}$ is the stochastic exponential from \eqref{eq:eta}, 
$(\tau^{x,D},\xi_{\tau^{x,D}})$ is the bankruptcy process from Definition~\ref{dfn:bankruptcy}, 
and $\rho>0$ is the discount rate.
        % For any $x\in \mathbb{R}_+$ and $D\in {\cal A}^x$, let $\Theta^{D}$ denote the set of all $\theta\in \Theta$~satisfying  
        %  \begin{align}
        %  &\mathbb{E}\bigg[\bigg(\int_{0}^{\tau^{x,D}} \eta_{t}^{\theta}e^{\int_{0}^{t} \frac{\theta^2_s}{2\cR}ds -\rho t} dD_t\bigg)^2\bigg]<\infty\quad \mbox{and}\nonumber\\
        %  &\mathbb{E}\bigg[(\eta_{\tau^{x,D}}^{\theta})^2e^{\int_{0}^{\tau^{x,D}} \frac{\theta^2_u}{\cR}du } \max\{\xi_{\tau^{x,D}}^{2},1\} \bigg]<\infty,\label{eq:ThetaD}
        %  \end{align}
        %  where $\eta^{\theta}$ is defined as in \eqref{eq:eta}, $\rho>0$ is the discount rate appearing in Definition \ref{dfn:admissible}, and $(\tau^{x,D},\xi_{\tau^{x,D}})$ are defined as in Definition \ref{dfn:bankruptcy}. %and \color{red}$\xi_{\tau^{x,D}}\in L^{2}({\cal F}_{\tau^{x,D}})$ is a bankruptcy payout such that $\xi_{\tau^{x,D}}>0$ $\mathbb{P}$-a.s.. \color{black}
    \end{dfn}

    % For any $x\in \mathbb{R}_+$, $D\in {\cal A}^x$ and , we define a dividend-value recursive utility process $V^D:=(V^D_t)_{t\ge 0}$ of the investor associated with $D$ by
    % \begin{align}\label{eq:recur_singular_V}
    % 		V^D_t := \mathbb{E}_t \bigg[ \int_{t\wedge \tau}^{\tau} g(s, V^D_s)dD_s + \xi_{\tau}\bigg], \quad t\ge 0.
    % 	\end{align}
    %  \begin{dfn}%\label{dfn:recurs_V} 
    % 		%Given $(\tau,D)\in\mathcal{T}\times\mathcal{A}$ and $\xi_\tau \in L^1({\cal F}_\tau)$, denote by 
    %         We define $\mathbb{I}(g,D)$ the set of all processes $V\in {\cal P}$ such that the stochastic integral $\int g(t,V_t) dD_t$ exists and $V$ satisfies $\mathbb{E}[\int_0^{\tau}|g(s,V_s)|dD_s]<\infty$. 
    %         \begin{enumerate}[label=\roman*.]
    %             \item $V\in\mathbb{I}(g,D)$ is a utility process associated with $(g,D)$ if it 
    %             satisfies \eqref{eq:recur_singular_V};
    %             \item Denote by $\mathbb{UI}(g,D)$ the subset of all utility processes $V$ that are uniformly integrable.
    %         \end{enumerate}
    % \end{dfn}
    % This formulation is inspired by the concept of SDU %from regular control problems 
    % (see Duffie and Epstein~\cite{DuffieEpstein1992}, Epstein and Zin~\cite{EpsteinZin1989}, Schroder and Skiadas~\cite{Schroder1999}), implying %creates a feedback effect such 
    % the dividend-value utility at time $t$ may depend in a nonlinear way on its value at future~times. 

    %For any $x\in \mathbb{R}_+$ and $D\in {\cal A}^x$, we define 
    \noindent The executive's robust singular dividend criterion ${\cal V}^{x, D}:=({\cal V}_t^{x, D})_{t \geq 0}$ is
	\begin{align}\label{dfn:rbs_V}
		{\cal V}^{x,D}_t := \essinf\{{\cal V}_{t}^{x,D,\theta}:\theta\in \Theta^D \}, \quad t\ge 0,
	\end{align}
	where for each $\theta\in \Theta^D$, the process ${\cal V}^{x,D,\theta}:=({\cal V}_{t}^{x,D,\theta})_{t\geq 0}$ is defined by
	\begin{align}\label{dfn:rbs_V_theta}
		{\cal V}_{t}^{x,D,\theta} &:=  \mathbb{E}^{\theta}_t\bigg[ \int_{t\wedge\tau^{x,D}}^{\tau^{x,D}} \bigg(e^{-\rho s} dD_s+\frac{1}{2\cR} 	{\cal V}^{x,D,\theta}_s \theta_s^2ds\bigg) +\xi_{\tau^{x,D}}\bigg].
    \end{align}
    %so that ${\cal V}^{x,D,\theta}$ is the discounted version of \eqref{eq:object_J}, i.e., ${\cal V}_t^{x,D,\theta}=e^{-\rho t}J^{\theta}(X_t^{x,D};D)$, hence $\sup_{D\in {\cal A}^x}{\cal V}_0^{x,D}=J^*(x)$; see~\eqref{eq:object_J}.
    In the following proposition, we establish the well-posedness of the robust singular dividend criterion.
    The proof can be found in Section~\ref{proof:pro:Maenhoutdef}. 
    \begin{pro}\label{pro:Maenhoutdef}
    Suppose that Assumption \ref{as:recurrent} holds. For any $(x,D,\theta)\in \mathbb{R}_+\times {\cal A}^x \times \Theta^D$, there exists ${\cal V}^{x,D,\theta}\in {\cal P}$ satisfying \eqref{dfn:rbs_V_theta}, which is unique in $\{Y\in\mathcal{P}:\mathbb{E}[\sup_{t\ge 0} Y_{t\wedge\tau^{x,D}}^2]<\infty\}$. 
    % and admits the representation: for $t\ge0$
    % \begin{align}\label{eq:thm:rob_rec:0}
    %         {\cal V}_{t}^{x,D,\theta} = \mathbb{E}^{\theta}_t\bigg[ \int_{t\wedge\tau^{x,D}}^{ \tau^{x,D}} e^{\int_{t}^{s} \frac{\theta^2_u}{2\cR}du -\rho s} dD_s + e^{\int_{t\wedge\tau^{x,D}}^{\tau^{x,D}} \frac{\theta^2_u}{2\cR}du}\xi_{\tau^{x,D}}\bigg].
    % \end{align}
    Therefore, %the robust singular dividend criterion 
    ${\cal V}^{x,D}$ in \eqref{dfn:rbs_V} is well-defined. %Furthermore, $V^{\theta,D}$ 
    \end{pro}
    %\color{red} 
    
    We define the executive’s robust dividend optimization problem by
    \begin{align}
    \label{eq:opt_Maenhout}
    {\cal V}_t^{x} := \esssup \{{\cal V}^{x,D}_t:{D \in {\mathcal{A}}^x} \},\quad (t,x)\in \mathbb{R}_+\times \mathbb{R}_+.
    \end{align}

    In the next section, we establish the equivalence between the two optimization problems given in \eqref{eq:opt_EZ} and \eqref{eq:opt_Maenhout}. We then characterize the optimal dividend strategy in a Markovian setting under suitable conditions.
    %\color{black}

    %%%%%%%%%%%%%%%%%%%%%%%%%%%%%%%%%%%%%%%%%%%%%%%%%%%%%%%%%%%%%%%%%%%%%%%%%%%%%%%%%%%%%%%%%%%%%%%%%%%%%%%%%%%%%%%%%%%%%%%%%%%%%%%%%%%%%%%%%%
    \subsection{Main theorems}\label{sec:main}
    Throughout this section, we set $R\equiv {\cal R}\in (0,1)$, that is, the degrees of ambiguity aversion and risk aversion coincide. Our first main result establishes the equivalence between the EZ utility and the Maenhout's robust singular dividend criterion in the setting with random time horizons and singular controls. 
    \begin{thm}\label{thm:rob_rec}
        Suppose that Assumption \ref{as:recurrent} holds. 
        For any $x\in \mathbb{R}_+$ and $D\in {\cal A}^x$, the unique utility process $V^{x,D} \in \mathbb{UI}(g_{\EZ},D,{\tau^{x,D}},\xi^{1-R}_{\tau^{x,D}})$ and the robust singular control criterion ${\cal V}^{x,D}$ given in \eqref{dfn:rbs_V} (see Propositions~\ref{pro:recurs_BSDE}, \ref{pro:Maenhoutdef}) are equivalent in the sense~that for all $t\geq 0$, $\mathbb{P}$-a.s.,%it holds that 
        \[
            V_t^{x,D}  = ({\cal V}_t^{x,D})^{1-R}.
        \]
        Therefore, it holds that ${V}_t^x=({\cal V}_t^{x})^{{1-R}}$, $\mathbb{P}$-a.s.; see \eqref{eq:opt_EZ} and \eqref{eq:opt_Maenhout}. 
	\end{thm}

    We note that the general roadmap for Theorem \ref{thm:rob_rec} follows Skiadas~\cite[Theorem~5, Section~6]{skiadas2003robust}; see also Maenhout et al.~\cite[Section~3.2]{balter2025model} for a recent alternative perspective. Theorem \ref{thm:rob_rec} extends these works by showing that the two singular control problems under EZ and Maenhout preferences are equivalent. The proof of Theorem \ref{thm:rob_rec} is deferred to Section \ref{proof:thm:rob_rec}.

    %\subsection{Characterization of robust singular dividend flow}\label{sec:4}
    We now aim to characterize the optimal dividend flow by establishing a dynamic programming principle. To that end, we impose the following Markovian setting on the bankruptcy payout in Definition \ref{dfn:bankruptcy} as well as the regularity and growth conditions on the drift and volatility coefficients $\mu,\sigma$ in the surplus processes \eqref{eq:uncontrolX} and \eqref{eq:controlX}. We first recall the discount rate $\rho>0$.
    %under the following assumption. To that end, recall the discount rate $\rho>0$.
    \begin{as}\label{as:default}
        %\begin{itemize}
            %\item [(i)] 
            For any $x\in \mathbb{R}_+$ and $D\in {\cal A}^x$, the bankruptcy payout is given by $\xi_{\tau^{x,D}}:=e^{-\rho \tau^{x,D}}\xi_0$ for some constant $\xi_0>0$ (so that \eqref{def:bankrupcypayment} holds).
            %with $\rho$ is the discount rate.
            %\item [(ii)] The drift term $\mu:{\cal X}\to \mathbb{R}$ of the SDE \eqref{eq:uncontrolX} satisfies $\sup_{x\ge 0}( \mu(x) -\rho x ) \le \overline{\mu}$ with some constant $\overline{\mu}\in \mathbb{R}$.
        %\end{itemize}
    \end{as}
    %Moreover, we impose further conditions on 
    \begin{as}\label{as:bdy_cond} 
    $\mu$ is in $C^1({\mathbb{R}_+})$ with $\mu(0)>0$ and satisfies 
    \begin{align}\label{as:bound}
        \sup_{x\ge 0}( \mu(x) -\rho x ) \le \overline{\mu},
    \end{align}
    for some $\overline{\mu}\in \mathbb{R}$. $\sigma$ is in $C^1({\mathbb{R}_+})$ and is non-decreasing. Moreover, there exist $\underline{b},\hat{b}, \overline{b}\in[0,\infty]$ with $\underline{b} < \hat{b} < \overline{b}$ and the following conditions hold:
    \begin{enumerate}[label=\roman*.]
        \item $\mu(x)^2 \ge  3R\rho\sigma(x)^2  $ on $x\in[0,\underline{b}]$ and there exists $$\overline{b} = \inf\{x\in(\underline{b},\infty): \mu(x)^2 = 2 R\rho\sigma(x)^2\},$$ with %the convention that 
        $\inf\{\varnothing\} : =\infty$. In addition, $\mu(x)^2 < 2 R\rho\sigma(x)^2$ for all $x>\overline{b}$.
        \item Define $\psi^\pm:[0,\overline{b}]\ni x\rightarrow \psi^\pm(x) \in \mathbb{R}$ by $
            \psi^\pm(x):= \frac{\mu(x)\pm \sqrt{\mu^2(x)-2R\rho \sigma^2(x)}}{2\rho}.$
            
        In particular, $\psi^\pm$ are roots of $Q(\psi;x):=\rho \psi^2 - \mu(x) \psi + \frac{R}{2}\sigma^2(x)$.
        %In general, they are complex valued, but they are well-defined and real valued on $[0,\overline{b}]$ by~(ii).
        The map $\psi^+(x)-x$ increases on $[0,\underline{b}]$ and  decreases~on $(\underline{b},\overline{b})$. 
        \item  $ \sup_{x\in[0,\overline{b}]} \{\psi^-(x) -x \} < \xi_0 < \psi^+(\underline{b}) - \underline{b}$.
        \item $\hat{b}$  is the smallest value in $(\underline{b},\overline{b})$ that satisfies $\psi^+(\hat{b})-\hat{b} = \xi_0$.
    \end{enumerate}
    \end{as}

    \begin{rem}
    Assumption~\ref{as:bdy_cond} extends the sufficient conditions for solvability and regularity of the optimal dividend problem in Alvarez and Virtanen~\cite[Lemma 3.1]{alvarez2006class} to account for ambiguity in the underlying surplus model. Below, we outline implications and technical reasons therein (see also Remark~\ref{rem:as:bound}).
    \begin{enumerate}[label=\roman*.]
        \item The condition \eqref{as:bound} implies that the net appreciation rate of the surplus is uniformly bounded. This ensures the continuity of the map ${\cal V}:\mathbb{R}_+\to \mathbb{R}$ (defined in~\eqref{eq:obj_maxmin}, particularly at $x = 0$, which is crucial for the analysis of the variational inequality given in \eqref{eq:HJB_VI} and  the uniqueness of its solution.
        \item Assumption~\ref{as:bdy_cond}\;i. ensures that $\mu(x)$ is sufficiently strong relative to $\sigma(x)$ and %the ambiguity parameter 
        $\mathcal{R}$ in the low-reserve region. On the other hand, it also restricts $\mathcal{R}$ in the sense that if $\mathcal{R}$ is too large, immediate liquidation becomes optimal, reflecting extreme ambiguity aversion as discussed at the end of Section~\ref{sec:5}.
        \item The term $\rho(\psi^+(x)-x)$ in Assumption~\ref{as:bdy_cond}\;ii. represents the ambiguity-adjusted net appreciation rate, reducing to $\mu(x)-\rho x$ when $\mathcal{R}=0$ as in \cite[Lemma 3.1~(ii)]{alvarez2006class}. This condition requires the rate to be bounded, with dividends paid when it is decreasing, mirroring the classical case.
        \item The condition on the bankruptcy payment $\xi_0$ (in Assumption~\ref{as:bdy_cond}\;iii.) ensures that it is neither too low nor too high. This is necessary for a well-posed and meaningful optimal policy and prevents degenerate cases including immediate bankruptcy or indefinite reserve accumulation. 
    \end{enumerate}
    \end{rem}

    %Let us provide some examples satisfying Assumption \ref{as:bdy_cond}.
    %Before discussing the assumption, we validate it within a standard mean-reverting surplus model, an extension of Brownian motion with constant drift, which aligns with substantial empirical support (see \cite{Cadenillas2007MF}). 
    \begin{rem}\label{rem:OU} %Define the following two processes. %In mean-reverting surplus mode, the dynamic \eqref{eq:uncontrolX} of the uncontrolled surplus (or uncontrolled cash reservoir) is given by 
    The Ornstein–Uhlenbeck process satisfies Assumptions \ref{as:recurrent} and \ref{as:bdy_cond} under the following conditions on its parameters: 
    $X^x$ in \eqref{eq:uncontrolX} %of (or uncontrolled cash reservoir) 
        is given by 
        \begin{align*}
            d X^x_t  = \overline{\kappa}(\overline{m}-X^x_t) dt + \overline{\sigma} dW_t,\quad x\in {\cal X}:=\mathbb{R},
        \end{align*}
        where $\overline{\kappa}>0$, $\overline{m}>0$ and $\overline{\sigma}>0$ satisfy
        \begin{align*}
            (\overline{\kappa} \overline{m})^2 > 3R\rho\overline{\sigma}^2, \qquad  \sqrt{\frac{R}{2\rho}} <  \xi_0 < \frac{\overline{\kappa}\overline{m} + \sqrt{(\overline{\kappa}\overline{m})^2 -2R\rho\overline{\sigma}^2}}{2\rho},
        \end{align*}
        with $\xi_0>0$ given in Assumption \ref{as:default}. We let $\underline{b} = 0$ and $\overline{b} = \overline{m} - \overline{\sigma}\sqrt{2R\rho}/\overline{\kappa}$. We note that this mean-reverting process is considered in Cadenillas et al.~\cite{Cadenillas2007MF}. 
    \end{rem}

    Under Assumptions~\ref{as:recurrent}, \ref{as:default} and \ref{as:bdy_cond}, we define a map ${\cal V}:\mathbb{R}_+\to \mathbb{R}$ by
    \begin{align}\label{eq:obj_maxmin}
        {\cal V}(x):={\cal V}_0^x=\sup_{D\in{\cal A}^x}{\cal V}_0^{x,D}=\sup_{D\in{\cal A}^x}\inf_{\theta\in\Theta^D}{\cal V}_0^{x,D,\theta},\quad x\in \mathbb{R}_+,%\mbox{where}\quad \overline {\cal V}_t^{x,D}:=e^{\rho t} {\cal V}_t^{x,D}\quad t\geq 0,
    \end{align}
    where $({\cal V}_t^{x})_{t\geq 0}$ is the robust dividend optimization problem defined in \eqref{eq:opt_Maenhout}. %The regularity of \(J^*\) as well as its growth property is established in the following lemma. The proof is provided in Section~\ref{app:b}. 
    \color{black}
    
    %We will show that the map $J^*$ in \eqref{eq:obj_maxmin} and the threshold level $b$ can be characterized by a dynamic programming principle. 
    From the underlying SDE in \eqref{eq:controlX} and the robust objective in \eqref{eq:obj_maxmin}, we introduce a nonlinear operator ${\cal L}$ acting on $v\in C^2({\mathbb{R}_+})$, defined by a Hamiltonian over $\theta\in \mathbb{R}$: %which acts on $v\in C^2({\mathbb{R}_+})$ as
    \begin{align}
    %\begin{aligned}
    {\cal L}v(x):=&\inf_{\theta \in \mathbb{R}}\Big\{\frac{\sigma^2(x)}{2}v''(x)+\left(\mu(x)+\sigma(x) \theta\right)v'(x)+ \frac{\theta^2}{2R} v(x) -\rho v(x)  \Big\}\nonumber\\
    =&\frac{\sigma^2(x)}{2}v''(x) + \mu(x)v'(x)-R\frac{\sigma^2(x)}{2} \frac{(v'(x))^2}{v(x)}-\rho v(x),\label{eq:nonlinear_hamilton}
    %\end{aligned}
    \end{align}
    where the last equality is valid  provided that $v> 0$ on ${\mathbb{R}_+}$. 
    
    %Next, we present the relevant HJB equation. 
    Then the map ${\cal V}$ in \eqref{eq:obj_maxmin} is expected to be a solution of the following Hamilton-Jacobi-Bellman variational inequality (HJB-VI):
    \begin{align}\label{eq:HJB_VI}
    \left\{
    \begin{aligned}
        &\max\{{\cal L}v(x),1-v'(x)\}=0\qquad \mbox{on}\;\;{\mathbb{R}_+}; \\
        &v(0)=\xi_0.
    \end{aligned}
    \right.
    \end{align}
    %We note that the boundary condition in \eqref{eq:HJB_VI} should be $v(0+) = \xi_0$ as \eqref{eq:boundary_con_B} shows, rather than $v(0)= \xi_0$. Here and in what follows, we adopt a convention that the solution of the HJB-VI equation is extended to 0 by continuity and the value of $v$ at 0 is its limit from the right.
    
    %{Following our educated guess that the optimal control is a threshold control, we choose to work with the following more explicit free-boundary ODE. Namely, we are looking for a constant $b>0$ for which there exists $v\in C^2({\mathbb{R}_+})$ such that }
    Our educated guess from \eqref{eq:HJB_VI} is that the optimal dividend flow to \eqref{eq:obj_maxmin} is a {\it threshold~control}, which is defined as follows (see Cohen et al.~\cite{cohen2022optimal} and Chakraborty et al.~\cite{chakraborty2023optimal}): 
    %We focus on candidate controls known as {\it threshold dividend strategies}. %In many classical dividend-paying problems, the optimal strategy takes the form of a threshold or barrier strategy (see, e.g., \cite{alvarez2006class}).  Our main result demonstrates that these strategies is optimal even in a broader context,  which includes Maenhout's type of ambiguity and extends to the case of EZ singular control utility.
    %Following , we first provide a notion of a {\it Skorokhod map} on an
    %interval. 
    Denote by $\mathcal{D}(\mathbb{R}_+)$ the the space of right-continuous functions with left limits mapping $\mathbb{R}_+$ into $\mathbb{R}$.  Fix $\beta>0$. Given $\psi \in \mathcal{D}(\mathbb{R}_+)$ there exists a unique pair $(\phi, \eta) \in \mathcal{D}(\mathbb{R}_+)^2$ that satisfy the following:
    \begin{enumerate}[label=\roman*.]
        \item For every $t \in\mathbb{R}_+,  \phi(t)=\psi(t)+\eta(t) $;
        \item $\eta(0-)=0$, $\eta$ is nondecreasing, and $\int_0^{\infty} \mathds{1}_{\{\phi(s)<\beta\}} d \eta(s)=0.$
    \end{enumerate}
    We define the map $\Gamma_{\beta}[\psi]= (\Gamma_{\beta}^1,\Gamma_{\beta}^2)[\psi]:=(\phi, \eta)$ and call it {\it Skorokhod map}. We refer to Kruk et al.~\cite{kruk2007explicit} for its existence and uniqueness as well as continuity. 
    \begin{dfn}\label{def:thres_divid} Fix $x,b\in\mathbb{R}_+$. We call $D = D^{(b)}$ a $b$-threshold dividend strategy if for every continuous $\psi\in\mathcal{D}(\mathbb{R}_+)$, one has $(X^{x,D},D)(\psi) =\Gamma_{b}[\psi]$.
    \end{dfn}
    We note that $D^{(b)}\in \mathcal{A}^x$ for any $x,b\in\mathbb{R}_+$ (see Definition \ref{dfn:admissible}). 
    
    From the HJB-VI in \eqref{eq:HJB_VI} and Definition \ref{def:thres_divid}, we introduce a free-boundary problem. Specifically, we aim to find $b>0$ and $v\in C^2({\mathbb{R}_+})$ such that 
    \begin{align}\label{eq:fbdy}
    \left\{
        \begin{aligned}
            &{\cal L}v(x)=0,\;\; v'(x)\geq 1\qquad \mbox{on}\;\;(0,b];\\
            &{\cal L}v(x)\leq 0,\;\; v'(x)= 1\qquad \mbox{on}\;\;(b,\infty);\\
            &v(0)=\xi_0,
        \end{aligned}
    \right.
    \end{align}
    as well as $v''(b)=0$ and $v'(b)=1$.
    
    The intuition behind \eqref{eq:fbdy} is as follows: If $X_0^{x,D} = x\in (b,\infty)$, then in order to keep the process $X^{x,D}$ between 0 and $b$, there is an instantaneous dividend payment of size $x-b$. If $x\in(0,b)$, no action is being taken. If $X^{x,D}$ hits $b$, the threshold policy is taking action, leading to the Neumann boundary condition at $b$. The surplus level will be kept in $(0,b)$, with an initial payment $0\wedge(X_0^{x,D}-b)$ and then pay out dividend only when $X^{x,D}$ is at level $b$.

    We now prove that an optimal dividend strategy exists and that it is a threshold control. Moreover, we show that the threshold level and the value function ${\cal V}$ are characterized by the free-boundary problem \eqref{eq:fbdy}. The proofs of the following proposition and theorem are presented in Sections \ref{proof:pro:shooting} and \ref{sec:thm:verification}.

    \begin{pro}\label{pro:shooting} Suppose that Assumptions~\ref{as:recurrent}, \ref{as:default} and \ref{as:bdy_cond} hold. There exists a unique solution \( v^* : \mathbb{R}_+ \to \mathbb{R}_+ \) to \eqref{eq:fbdy} satisfying \( v^* \in C^2(\mathbb{R}_+) \). Moreover, the corresponding free boundary \( b^* \) is unique and lies in \( (\underline{b}, \hat{b}) \), where $\underline b$ and $\hat b$ are defined in Assumption~\ref{as:bdy_cond}. %Hence, $v^*$ is the solution of the HJB-VI \eqref{eq:HJB_VI}. \comment{Can we say that $C^2({\mathbb{R}_+})\cap C^1({\mathbb{R}_+})$? let us discuss later...}
    \end{pro}

    \begin{thm}\label{thm:verification} Suppose that Assumptions~\ref{as:recurrent}, \ref{as:default} and \ref{as:bdy_cond} hold. Let $ v^*$ be the unique solution to \eqref{eq:fbdy} with the free boundary $ b^*$ (see Proposition~\ref{pro:shooting}), and let ${\cal V}:\mathbb{R}_+\to\mathbb{R}$ be defined in \eqref{eq:obj_maxmin}. 
        %\item [(i)] $v^*$ is in $\in C^2({\mathbb{R}_+})\;\cap \;C({\mathbb{R}_+})$ and it satisfies the free boundary problem given in \eqref{eq:fbdy} with corresponding free boundary $x^*\in {\mathbb{R}_+}$. Hence, $v^*$ is a classical solution of the HJB-VI given in \eqref{eq:HJB_VI}.
        Furthermore, let $D^*:=D^{({{b}^*})}$ be the $b^*$-threshold dividend strategy (see Definition~\ref{def:thres_divid}). Then, for every $ x \in \mathbb{R}_+ $,
    $$
    v^*(x) = {\cal V}(x)={\cal V}_0^{x,D^*} = \inf_{\theta \in \Theta^{D^*}} {\cal V}_0^{x,D^{*},\theta}.
    $$
    In particular, it holds that for all $x\in \mathbb{R}_+$
    \begin{align*}
        v^*(x) = \sup_{D \in \mathcal{A}^x} \inf_{\theta \in \Theta^D} \mathbb{E}^{\theta} \bigg[ \int_{0}^{\tau^{x,D}} e^{-\rho t} \left( dD_t + \frac{\theta_t^2}{2 R} v^*(X_t^{x,D}) dt \right) + \xi_{\tau^{x,D}} \bigg].
    \end{align*}
    \end{thm}

    %%%%%%%%%%%%%%%%%%%%%%%%%%%%%%%%%%%%%%%%%%%%%%%%%%%%%%%%%%%%%%%%%%%%%%%%%%%%%%%%%%%%%%%%%%%%%%%%%%%%%%%%%%%%%%%%%%%%%%%%%%%%%%%%%%
    \subsection{Sensitivity analysis with respect to ambiguity aversion}\label{sec:5}
    %As in the previous section, we set \(R\equiv \cR\in (0,1)\) and
    In this section, we analyze the sensitivity of the robust singular criterion and its optimization w.r.t.~the degree of ambiguity aversion ${\cal R}$. Throughout this section, Assumption \ref{as:recurrent} is imposed without further mention (so that Propositions \ref{pro:recurs_BSDE} and \ref{pro:Maenhoutdef} and Theorem \ref{thm:rob_rec} hold). %All the proof of results are presented in Section \ref{sec:10}. 
    Furthermore, we rewrite all ${\cal R}$-dependent components in  Definition \ref{dfn:admissD}, \eqref{dfn:rbs_V}--\eqref{eq:opt_Maenhout} and \eqref{eq:obj_maxmin} as follows: for any ${\cal R}\in(0,1)$, 
    $\Theta^D({\cal R}),$ $({\cal V}_t^{x,D}(\cR))_{t\geq 0},$ $({\cal V}^{x,D,\theta}_t(\cR))_{t\geq 0}$, and $({\cal V}_t^x(\cR))_{t\geq 0}$.  
    
    In analogy, we rewrite all $R$-dependent components in \eqref{eq:EZ_V} and \eqref{eq:opt_EZ} as follows: for each $R\in(0,1)$, $(V_t^{x,D}(R))_{t\geq 0}$ and $(V_t^{x}(R))_{t\geq 0}$.% for $x\in \mathbb{R}_+$, $D\in{\cal A}^x$.
    
    The following theorem shows the sensitivity of the robust singular dividend criterion ${\cal V}^{x,D}(\cdot)$. First, we recall the classical dividend criterion $K^{x,D}$ in~\eqref{eq:std_singular_V}. %the Maenhout's robust singular control criterion. %w.r.t. ${\cal R}\in(0,1)$
%The proof can be found in Section~\ref{sec:10}.
\begin{thm}\label{thm:cts_VEZ}
For every $(x,D)\in \mathbb{R}_+\times {\cal A}^x$ and \(t \ge 0\), the mapping \((0,1) \ni {\cal R} \mapsto {\cal V}_t^{x,D}( {\cal R})\in \mathbb{R}_+\) is decreasing and continuous. Moreover, for every $t\geq0$, 
\begin{align}\label{eq:Vrob_lim}
    \lim_{\cR \downarrow 0} {\cal V}_t^{x,D}(\cR) = K_t^{x,D}\quad\text{and}\quad {\cal V}_t^{x,D}(\cR)\le K_t^{x,D},\quad \mbox{$\mathbb{P}$-a.s..}
\end{align}
%where 
%given in Setting \ref{set:sensitivity}\;i,\;ii..
\end{thm}
%The proof of can be found in Section \ref{sec:10}.
\begin{rem} \label{rem:Vrob(0)}
    \({\cal V}^{x,D}(\cdot) \) %Setting~\ref{set:sensitivity}~iii. 
    is not a priori defined at \({\cal R} = 0\). %due to the form of the utility process 
    However, by \eqref{eq:Vrob_lim}, we can extend its definition by setting \({\cal V}_t^{x,D}(0) := K^{x,D}_t\) for $t\in[0,T]$. 
    Moreover, Theorem~\ref{thm:rob_rec} implies \(V^{x,D}(\cR) = ({\cal V}^{x,D}(\cR))^{1-\cR}\) for all \(\cR\equiv R \in [0,1)\).  
   % Moreover, if we denote for every ${\cal R}\in[0,1)$ by $V^{\EZ,D}({\cal R}):=(V_t^{\EZ,D}({\cal R}))_{t\in[0,T]}$ the EZ singular control utility defined as in \eqref{eq:EZ_V} under setting $(\tau,D,\xi_\tau({\cal R}))$ given in Setting~\ref{set:sensitivity}, then by Theorem \ref{thm:rob_rec}, we have $V^{\EZ,D}({\cal R})=(V^{\rob,D}({\cal R}))^{1-{\cal R}}$ for all ${\cal R}\in[0,1)$.
\end{rem}
% In order to derive the sensitivity of \(J^*\) given in Section \ref{sec:4}, we introduce the following notations: for every ${\cal R}\in(0,1)$, $x\in\mathbb{R}_+$, and $D\in {\cal A}^x$, let $\Theta^D({\cal R})$ be defined as in \eqref{eq:ThetaD} under setting $(\tau^{x,D},D,(e^{-\rho \tau^{x,D}}\xi_0)^{1-{\cal R}})$; see Definition~\ref{dfn:primitive2}. Then for every $\theta \in \Theta^D({\cal R})$, let $J^\theta(x;D,{\cal R})$ be defined as in~\eqref{eq:object_J}. 
% Finally, for every $x\in \mathbb{R}_+$ and ${\cal R}\in(0,1)$, let $J^*(x;{\cal R})$ be defined as 
% \begin{align}\label{dfn:robust_J_R}
% J^*(x;\cR):= \sup_{D \in \mathcal{A}^x} \inf_{\theta \in \Theta^D(\cR)} J^{\theta}(x; D, \cR),%\quad \cR \in (0,1),
% \end{align}
% so that \(J^*(x;\cR) = V^{\rob,*}_0(\cR):=\sup_{D\in {\cal A}^x}V_0^{\rob,D}({\cal R})\) (see \eqref{eq:obj_maxmin}). Moreover, by Remark, we can extend the definition of $J^*(x;\cdot)$ at ${\cal R}=0$ by setting 
% \begin{align}\label{dfn:robust_J_0}
% J^*(x;0) := \sup_{D \in \mathcal{A}^x} \mathbb{E} \bigg[ e^{-\rho \tau^{x,D}}\xi_0 +  \int_0^{\tau^{x,D}} e^{-\rho s} \, dD_s \bigg].
% \end{align}

The following corollary shows the sensitivity of the robust singular dividend optimization ${\cal V}^{x}(\cdot)$. Before doing so, by Remark~\ref{rem:Vrob(0)}, we extend the definition of ${\cal V}^{x}(\cdot)$ at $\cR=0$ by setting 
\begin{align*}%\label{dfn:robust_J_0}
{\cal V}^{x}_t(0) := \esssup\{K_t^{x,D}:D\in {\cal A}^x\},\quad (t,x)\in \mathbb{R}_+\times \mathbb{R}_+.
%\sup_{D \in \mathcal{A}^x} \mathbb{E} \bigg[\int_{t\wedge\tau^{x,D} }^{\tau^{x,D}} e^{-\rho s}dD_s+ \xi_{\tau^{x,D}}  \bigg].
\end{align*}

\begin{cor}\label{coro:sec5} %Under Definition~\ref{dfn:primitive2}, for any \(x \in \mathbb{R}_+\), 
For every $(t,x)\in \mathbb{R}_+\times \mathbb{R}_+$, the mapping \([0,1) \ni \cR \mapsto {\cal V}^x_t(\cR)\) is decreasing, lower semi-continuous and right continuous.% In particular, 
\end{cor}

The following theorem shows that under additional conditions so that Theorem \ref{thm:verification} holds, the value function 
%With stronger conditions so that Theorem \ref{thm:verification} holds, %${\cal V}(x)$ defined in \eqref{eq:obj_maxmin} can be well defined for each ${\cal R}$
%we obtain (strong) continuity of the value function 
${\cal V}(x;{\cal R}):={\cal V}_0^{x}(\cR)$ has (strong) continuity w.r.t $\cR\equiv R$ (see \eqref{eq:obj_maxmin}). Consequently, the corresponding optimal threshold, denoted by $b^*_{\cR}$ for each~$\cR$, has continuity.
\begin{thm}\label{thm:cts_b*}
Suppose Assumption \ref{as:default} holds. Moreover, if there exists some closed interval \({\cal I}\subset[0,1)\) such that %Let \(x \in \mathbb{R}_+\) and let \({\cal I}\subset[0,1)\) be a . 
%Suppose that 
Assumption \ref{as:bdy_cond} hold for all $\cR\equiv R\in {\cal I}$, then for any $x\in\mathbb{R}_+$, the mappings \({\cal I}\ni\cR \mapsto {\cal V}(x;\cR)\in \mathbb{R}_+\) 
%(defined in  \eqref{dfn:robust_J_R},~\eqref{dfn:robust_J_0}) 
and \({\cal I}\ni\cR \mapsto b^*_{\cR}\in \mathbb{R}_+\) are continuous, i.e.,
\[ \lim_{\substack{|\cR_2 - \cR_1| \to 0 \\ \cR_1,\,\cR_2 \in {\cal I}}} \left| {\cal V}(x;{\cR_1}) - {\cal V}(x;{\cR_2}) \right| = 0, \quad
\lim_{\substack{|\cR_2 - \cR_1| \to 0 \\ \cR_1,\,\cR_2 \in {\cal I}}} \left| b^*_{\cR_1} - b^*_{\cR_2} \right| = 0.
\]
\end{thm}
The proofs of Theorem~\ref{thm:cts_VEZ}, Corollary~\ref{coro:sec5} and Theorem~\ref{thm:cts_b*} are presented in Section \ref{sec:10}.

    \begin{rem}
        %The monotonicity and right-continuity results given in 
        Theorem~\ref{thm:cts_VEZ} and Corollary \ref{coro:sec5} align with the existing sensitivity results for robust optimization problems under small ambiguity aversion (i.e., $\cR\ll 1$) in continuous time settings; see Arthur et al.~\cite{arthur2025sensitivity}, Bartl et al.~\cite{bartl2023sensitivity,bartl2024numerical}, Herrmann and Muhle-Karbe~\cite{herrmann2017model}, Herrmann et al.~\cite{herrmann2017hedging}. 
       % \color{black} While this paper establishes the well-posedness and sensitivity of the robust singular dividend optimization problem under small ambiguity aversion, the sensitivity of the robust dividend policy, as well as extensions of the analysis to the case ${\cal R}\in(1,\infty)$, remain open problems and constitute interesting directions for future research. 
        While this paper addressed the well-posedness and sensitivity of the robust singular dividend optimization for small ambiguity aversion, the sensitivity of the robust dividend policy and extensions of those analysis to the case for $\cR\in (1,\infty)$ and  suggest interesting avenues for future research.  
    \end{rem}
    \color{black}

    %%%%%%%%%%%%%%%%%%%%%%%%%%%%%%%%%%%%%%%%%%%%%%%%%%%%%%%%%%%%%%%%%%%%%%%%%%%%%%%%%%%%%%%%%%%%%%%%%%%%%%%%%%%%%%%%%%%%%%%%%%%%%%%%%%%%%%%%%%%%
    \section{Proofs}\label{sec:proof}

    \subsection{
    Well-posedness of EZ singular control utility: Proof of Proposition \ref{pro:recurs_BSDE}
    %Proof of Proposition \ref{pro:recurs_BSDE}
    }\label{proof:pro:recurs_BSDE}
    Throughout Section \ref{proof:pro:recurs_BSDE}, Assumption \ref{as:recurrent} is imposed without further mention. Moreover, we fix $(x,D)\in \mathbb{R}_+\times {\cal A}^x$ (see Definition \ref{dfn:admissible}) and set $\tau:=\tau^{x,D}$ and $\xi_{\tau}:=\xi_{\tau^{x,D}}$ (see Definition~\ref{dfn:bankruptcy}) for notational simplicity.% Moreover,  throughout.

    We introduce the following sets: %that are used throughout the appendices: 
    for any $p\geq1$ and $t\geq0$,
    \begin{enumerate}%[leftmargin=2.em]
		%\item[$\cdot$] $L^p(\mathcal{F}_t)$ is the set of all random variables $X\in L({\cal F}_t)$ such that $\|X\|_{L^p}^p:=\mathbb{E}[|X|^p]<\infty$; %Let $L^2(\mathcal{F}_t):=\{x: \text{ all }\mathcal{F}_t\text{-measurable random variable s.t. }\mathbb{E}[x^2]<\infty\}$.
	%	\item[$\cdot$] $L^p_{+}(\mathcal{F}_t)$ is the  subset of all random variables $X\in L^p(\mathcal{F}_t)$ which is bounded from below $\mathbb{P}$-a.s.;
            %Let $L^2_{+}$ denote the subspace of $x \in L^2$ such that $x$ is bounded from below a.s.
		\item[i.] $\mathcal{S}_p$ is the set of the set of all real-valued, $\mathbb{F}$-progressively measurable processes $Y$ such that $\|Y\|_{{\cal S}_p}^p:=\mathbb{E}[\sup _{t \geq 0}|Y_{t \wedge \tau}|^p]<\infty$;
		\item[ii.] ${\cal S}_{\infty}$ is the set of all real-valued, $\mathbb{F}$-progressively measurable processes $Y$ such that $\|Y\|_{{\cal S}_\infty}:=\inf\{C\geq 0: |Y_t| \leq C\;\mbox{for all}\;t\geq0\;\mbox{$\mathbb{P}$-a.s.}\}$; %$\|\sup _{t \geq 0}|Y_{t \wedge \tau}|\|_{\infty}<\infty$, that is $\sup _{t \geq 0}|Y_{t \wedge \tau}|$ is bounded a.s.
		\item[iii.] $\mathcal{M}_p$ is the set of all real-valued, $\mathbb{F}$-progressively measurable processes $Z$ such that $\|Z\|_{\mathcal{M}_p}^p:=\mathbb{E}[(\int_0^\tau |Z_t|^2 dt)^{p/2}]<\infty$;
		\item[iv.] Denote by $\mathcal{B}_p:=\mathcal{S}_{p} \times \mathcal{M}_p$ the product space equipped with the norm $\|(Y,Z) \|_{\mathcal{B}_p}^p: =\|Y\|_{{\cal S}_p}^p+ \|Z\|_{\mathcal{M}_p}^p$ for every $(Y,Z)\in\mathcal{B}_p$.
	\end{enumerate}

    According to Briand et al.~\cite{briand2003lp} and Popier~\cite{popier2007backward}, we introduce the following definition of a solution of BSDEs with random terminal times. First we denote for every $p\geq 1$ by $L_{\mathrm{loc}}^p(\mathbb{R}_+)$ the set of all locally $p$-integrable functions on ${\mathbb{R}_+}$.
	\begin{dfn}\label{def:BSDEdef}
            %Let $(\tau,D)\in {\cal T}\times\mathcal{A}$ and $\xi_\tau \in L^1({\cal F}_\tau)$. 
            A solution to \eqref{dfn:BSDE1} is a pair $(Y,Z) = (Y_t,Z_t)_{t\ge0}\in {\cal P}_+\times {\cal P}$ satisfying the following conditions $\mathbb{P}$-a.s.:
		\begin{enumerate}[label=\roman*.]
			\item $Y_t = \xi_{\tau}^{1-R}$ and $Z_t = 0$ on $\{t \ge \tau\}$;
			\item $t\mapsto \mathds{1}_{t\le\tau} g_{\text{EZ}}(t,Y_t)$ belongs to $L^1_{\mathrm{loc}}(\mathbb{R}_+)$, and $t\mapsto Z_t$ belongs to $L_{\mathrm{loc}}^2(\mathbb{R}_+)$;%\mycomment{any or fixed $p\geq 1$?} %{for $p\geq 1$, $L_{\mathrm{loc}}^p(\mathbb{R}_+)$	denotes the class of locally $p$-integrable functions on ${\mathbb{R}_+}$.} 
			\item For every $T\ge 0$, it holds that for $t\in[0,T]$,
			\begin{align*}
				Y_{t\wedge\tau} = Y_{T\wedge\tau} + \int_{t\wedge\tau}^{T\wedge\tau} g_{\text{EZ}}(s,Y_s) dD_s - \int_{t\wedge\tau}^{T\wedge\tau}Z_s dW_s.
			\end{align*}
		\end{enumerate}
		Moreover, if the solution $(Y,Z)$ is in ${\cal B}_2$, we then call it an $L^2$-solution to~\eqref{dfn:BSDE1}.
        %\begin{align*}
            %\mathbb{E}\bigg[ \sup_{t\ge0} |Y_{t\wedge\tau}|^2 + \int_{0}^{\tau} |Z_t|^2 dt  \bigg] <\infty
        %\end{align*}
		 %holds. 
         In particular, the $L^2$-solution implies that $Y\in \mathbb{UI}(g_{\EZ},D,{\tau},\xi^{1-R}_{\tau})$. %(see Definition~\ref{dfn:recurs_V}). %\comment{should be $\tau$? instead of $t\wedge \tau$.}
	\end{dfn}
    We establish the well-posedness of the EZ utility (Proposition \ref{pro:recurs_BSDE}) through four preliminary results. The complete proof, which integrates these steps, is detailed in Section~\ref{proof:final:pro:recurs_BSDE}.
    
    %The proof of Proposition \ref{pro:recurs_BSDE} proceeds in four steps after which Section~\ref{proof:final:pro:recurs_BSDE} incorporates all the steps and finalizes the proof. 
    
    \subsubsection{Step 1: Lipschitz aggregator: fixed time horizon and bounded control.} 
    For any $0<T<\infty$, we define the spaces ${\cal S}_{2}([0,T])$, $\mathcal{M}_2([0,T])$, and $\mathcal{B}_2([0,T])$ similarly, with $\tau$ replaced by $T$ in the definition of ${\cal S}_2$, ${\cal M}_2$, and ${\cal B}_2$. %We first present a useful estimation and a fundamental result. 
        % \begin{lem}\label{lem:A1}
        %  Let $T<\infty$ and $(Y_t,Z_t)\in \mathcal{B}_2([0,T])$. Then the continuous local martingale $(\int_{0}^t Y_s Z_s dW_s)_{t\in[0,T]}$ is a uniformly integrable martingale.
        % \end{lem}
        % \begin{proof}
        %     This is exactly \cite[Lemma A.1]{Aurand2023}, we omit the proof here.
        % \end{proof}
        \begin{lem}\label{lem:A2}
        Let $T<\infty$. Let $A\in {\cal S}_{2}([0,T])$ and $\hat D\in\mathcal{A}$ be such that $\hat D_T \le K$ $\mathbb{P}$-a.s.~for some constant $K$, and both $A$ and $\hat D$ have continuous paths. If the following holds for every $\zeta \in {\cal T}$ such that $\zeta \leq T$ $\mathbb{P}$-a.s., 
        \begin{align*}
            \mathbb{E}[A_{\zeta}] \le \alpha + \mathbb{E}\bigg[\int_{\zeta}^{T} A_s d\hat D_s\bigg], 
        \end{align*}
        with some constant $\alpha$, then it holds that $\mathbb{E}[A_t] \le \alpha e^K$ for every $t\in[0,T]$.
        \end{lem}
        \begin{proof}
            For fixed $t\in[0,T)$ and $0\le k\le K$, set $\zeta_k = \inf\{s \ge t: \hat D_s \ge k\}\wedge T$. Then each $\zeta_k$ is a stopping time bounded by $T$. Moreover, since $\{\zeta_k\le s\} = \{\hat D_s\ge k\}$ for $s\ge t$ and $\zeta_{K} = T$ holds $\mathbb{P}$-a.s., it follows that
            \begin{align*}
                \mathbb{E}[A_{\zeta_k}] \le& \alpha + \mathbb{E}\bigg[ \int_{\zeta_k}^{T} A_s d \hat D_s\bigg] =  \alpha + 
                \mathbb{E}\bigg[ \int_{t}^{T} A_s \mathds{1}_{s\ge \zeta_k}d \hat D_s\bigg]\\
                \le &  \alpha +  \mathbb{E}\bigg[ \int_{t}^{T} A_s \mathds{1}_{ \zeta_k \le s \le \zeta_{K}}d \hat D_s\bigg] %= \alpha +  \mathbb{E}\bigg[ \int_{k}^{K} A_{\zeta_s} d s\bigg] 
                =  \alpha +  \int_{k}^K \mathbb{E}[A_{\zeta_s}] ds.
            \end{align*}
            Set $u_s:= \mathbb{E}[A_{\zeta_s}]$ for $k\in[0,K]$. Then since $u_k \le \alpha +  \int_{k}^K u_s ds$, for $k\in[0,K]$, the standard (backward) Gr\"onwall's inequality ensures that $u_k \le \alpha e^{K-k}$ for $k\in[0,K]$. We complete the proof by letting $k=0$. 
             
        \end{proof}

    \begin{lem}[Lipschitz aggregator]\label{lem:BSDE_exist_Lip}
		%\color{magenta} 
        Fix $0<T < \infty$. Let $\hat D\in {\cal A}$ be such that  $ \int_0^{T}e^{-\rho t}d \hat D_t \le K$, $\mathbb{P}$-a.s., with some constant $K>0$. Let $\hat \xi_T \in L^2(\mathcal{F}_T)$. Assume that $g$ satisfies the following conditions: %\mycomment{here it is hard to understand the statement...}
        %\color{black}
		\begin{itemize}
			\labitem{(A1)}{item:A1} $\exists$ $C_g>0$ s.t. $|g(t,y_1)-g(t,y_2)|\le C_g e^{-\rho t}|y_1-y_2|$ $\forall y_1,y_2$.
			\labitem{(A2)}{item:A2} $\int_0^{T} |g(t,0)|d\hat D_t \in L^2({\cal F}_T)$.
		\end{itemize}
		Then the BSDE 
        \begin{align*}%\label{eq:BSDE2}
        Y_t = \hat \xi_{T}^{1-R} + \int_{t}^{T} g(s, Y_s) \, d\hat D_s - \int_{t}^{T} Z_s dW_s, \quad t \in[0,T],
        \end{align*}
        %with parameter $(T,\xi_{T},g,D)$ 
        has a unique solution in $\mathcal{B}_2$.
	\end{lem}
	\begin{proof}
		%The solution is constructed by modifying the proof of \cite[Proposition 2.2]{pardoux1990adapted}. 
		The idea is to extend the Picard iteration method given by Pardoux and Peng~\cite[Proposition~2.2]{pardoux1990adapted} and then to show that the corresponding sequence is Cauchy in~$\mathcal{B}_2$.
        
        %The idea is constructing an approximating sequence by a kind of Picard iteration, and we show that it is a Cauchy sequence in $\mathcal{B}_2$. 
        \noindent Let $y^0 = 0$, and $\{ (y^n_t,z^n_t)_{t\in[0,T]}\}_{n\ge1}$ be a sequence in $\mathcal{B}_2$ defined recursively as
		\begin{align*}
			y_t^{n+1} = \hat \xi_T^{1-R}+ \int_{t}^T g(s,y_s^n)d \hat D_s - \int_{t}^T  z_s^{n} dW_s ,\quad t \in[0,T],
		\end{align*}
     %{\color{redk} where the integration w.r.t $D$ is defined as the Stieltjes integral. Fix an $\omega$ such that $t\mapsto D_t(\omega)$ is right continuous and non-decrasing. This function induces a measure $\mu_{D}(\omega,ds)$ on $\mathbb{R}_+$. If $y^n_s$ is progressively measurable and so is $g(\omega,s,y^n_s)$, we can define, $\omega$-by-$\omega$, the integral $I(t,\omega) = \int_{0}^t g(\omega,s,y^n_s) d D_s(\omega)$. In addition, $I(t,\omega)$ is right-continuous in $t$ and progressively measurable. When the integrand process has continuous paths, i.e., $s \mapsto g(s,\omega)$ is continuous a.s., the Stieltjest integral is also known as the Riemann-Stieltjes integral (for fixed $\omega$).} 
     For any $y^n\in \mathcal{S}_2$, we construct $y^{n+1}$ and $z^n$ as follows. From \ref{item:A1} and \ref{item:A2}, 
    		\begin{align*}
    			&\mathbb{E}\bigg[\bigg(\hat \xi_T^{1-R}+\int_0^{T} g(s,y^n_s)d \hat D_s \bigg)^2 \bigg]\\
                &\quad \le 2 \mathbb{E}\bigg[\hat \xi_T^{2(1-R)}+ \bigg(\int_0^{T} \big(|g(s,0)| +  C_g e^{-\rho s}|y^n_s|\big)d \hat D_s \bigg)^2 \bigg]\\
    			&\quad \le2\mathbb{E}[\hat \xi_T^{2(1-R)}] + 4 \mathbb{E}\bigg[\int_0^{T} |g(s,0)|d\hat D_s\bigg]^2+  C_g \mathbb{E}\bigg[\sup_{0\le t\le T}|y^n_t|^2 \bigg]K^2<\infty.
    		\end{align*}
    		This implies that $\{ \mathbb{E}_t[\hat \xi_T^{1-R}+\int_0^{T} g(s,y^n_s)d \hat D_s  ]\}_{t\ge0}$ is a square integrable martingale. By the martingale representation theorem, we have a unique pair $(z^{n},y^{n+1})\in\mathcal{M}_2\times\mathcal{S}_2$ such that for $0\le t\le T$
    		\begin{align*}
    			&\int_0^t z^{n}_s dW_s = \mathbb{E}_t\bigg[\hat \xi^{1-R}_T+\int_0^{T} g(s,y^n_s)d \hat D_s \bigg] - \mathbb{E}\bigg[\hat \xi_T^{1-R}+\int_0^{T} g(s,y^n_s)d \hat D_s  \bigg]\nonumber\\ %\label{eq:mapphi}
    			&\quad \mbox{and}\quad y^{n+1}_t: = \mathbb{E}_t\bigg[\hat \xi_T^{1-R}+\int_t^{T} g(s,y^n_s)d \hat D_s \bigg].
    		\end{align*}
      For any $\hat \tau \in \mathcal{T}$, applying It\^{o}'s lemma to $(u_{s\wedge\hat \tau}^{n+1})^2: = |y_{s\wedge\hat \tau}^{n+1} - y_{s\wedge\hat \tau}^{n}|^2$ gives
    		\begin{align*}
    			&(u_{t\wedge\hat \tau}^{n+1})^2 + \int_{{s\wedge\hat \tau}}^{T\wedge\hat \tau}| z_s^{n}-  z_s^{n-1}|^2 ds \\
                    &\quad = 2\int_{{s\wedge\hat \tau}}^{T\wedge\hat \tau} \big(g(s,y_s^n)- g(s,y_s^{n-1})\big) u_{s}^{n+1}d\hat D_s - 2 \int_{{s\wedge\hat \tau}}^{T\wedge\hat \tau} u_s^{n+1}(z_s^{n}-  z_s^{n-1}) dW_s.
    		\end{align*}
    	It is clear by Aurand and Huang~\cite[Lemma A.1]{Aurand2023} that the above stochastic integral is a uniformly integrable martingale and has zero expectation. By taking expectations on both sides, together with condition~\ref{item:A1} it follows that 
    		\begin{align*}
    			&\mathbb{E} \big[(u_{t\wedge\hat \tau}^{n+1})^2\big]  + \mathbb{E}\bigg[\int_{t\wedge\hat \tau}^{T\wedge\hat \tau}| z_s^{n}-  z_s^{n-1}|^2 ds\bigg]\\
                &\quad \le  2\mathbb{E}\bigg[\int_{t\wedge\hat \tau}^{T\wedge\hat \tau} C_g u_{s}^n u_{s}^{n+1}e^{-\rho s}d\hat {D}_s\bigg] \\
    			&\quad \le %C_g \mathbb{E}\bigg[\int_{t\wedge\tau}^{T\wedge\tau} \big((u_{s}^n)^2 +  (u_{s}^{n+1})^2\big) e^{-\rho s}d{D}_s\bigg]
       C_g \mathbb{E}\bigg[\int_{t\wedge\hat \tau}^{T\wedge\hat \tau} (u_{s}^n)^2 e^{-\rho s}d\hat {D}_s\bigg] + C_g \mathbb{E}\bigg[\int_{t\wedge\hat \tau}^{T\wedge\hat \tau} (u_{s}^{n+1})^2 e^{-\rho s}d\hat {D}_s\bigg].
    		\end{align*}
       Lemma~\ref{lem:A2} implies 
    		\begin{align}\label{eq:thmA.3pf_1}
    			\mathbb{E}[(u_t^{n+1})^2] \le  C_g e^{C_g K} \mathbb{E}\bigg[\int_{t}^T  (u_{s}^n)^2 e^{-\rho s} d\hat {D}_s\bigg],\quad t\in[0,T].
    		\end{align}
    Consider $\zeta_{s}:=\inf\{t\ge0: \int_0^t e^{-\rho u}\hat D_{u}\ge s \}\wedge T$ for $s\in[0,K]$, which is bounded.
    From the proof of Lemma~\ref{lem:A2},  it is obvious that
    \begin{align}\label{eq:thmA.3pf_2}
    			\mathbb{E}[(u_{\zeta_t}^{n+1})^2] \le  C_g e^{C_g K} \mathbb{E}\bigg[\int_{\zeta_t}^T  (u_{s}^n)^2 e^{-\rho s} d\hat {D}_s\bigg]
    		\end{align}
    holds for $t\in[0,K]$. Iterating \eqref{eq:thmA.3pf_1} and \eqref{eq:thmA.3pf_2} gives for $t\in[0,T]$, 
    \begingroup
    \begin{align}
        \mathbb{E}[(u_t^{n+1})^2] &\le  C_g e^{C_g K} \mathbb{E}\bigg[\int_{t}^T  (u_{s}^n)^2 e^{-\rho s} d\hat {D}_s\bigg]\nonumber\\
        &\le C_g e^{C_g K} \mathbb{E}\bigg[\int_{0}^T  (u_{s}^n)^2 e^{-\rho s} \mathds{1}_{ \zeta_0 \le s \le \zeta_K}  d\hat {D}_s\bigg]\nonumber\\
        &= C_g e^{C_g K} \mathbb{E}\bigg[\int_{0}^{K}  (u_{\zeta_s}^n)^2 ds \bigg] =  C_g e^{C_g K} \int_{0}^{K} \mathbb{E}[ (u_{\zeta_s}^n)^2] ds\nonumber\\
        &\le (C_g e^{C_g K})^2 \int_{0}^{K}  \mathbb{E}\bigg[\int_{\zeta_s}^T  (u_{t}^{n-1})^2 e^{-\rho t} d\hat {D}_t\bigg] ds\nonumber\\
        &= (C_g e^{C_g K})^2 \int_{0}^{K} \int_{t_1}^{K} \mathbb{E}[ (u_{\zeta_{t_2}}^{n-1})^2] d t_2 d t_1 \nonumber\\
        &\cdots\nonumber\\
        &\le (C_g e^{C_g K})^n \int_{0}^{K} \int_{t_1}^{K}\cdots \int_{t_{n-1}}^K \mathbb{E}[ (u_{\zeta_{t_n}}^{1})^2] d t_{n} dt_{n-1}\cdots d t_1\nonumber\\
        &\le (C_g e^{C_g K})^n \mathbb{E}\bigg[\sup_{t\in[0,T]}(u_t^1)^2\bigg] \frac{K^n}{n!}.\nonumber
    \end{align}
    \endgroup
    
    The above estimate implies that $\{(y^n,z^n)\}_{n\in \mathbb{N}}$  is Cauchy in $\mathcal{B}_2$, due to the fact that $y^1\in\mathcal{S}_2$. It follows that $\lim_{n\rightarrow\infty}(y^n,z^n) = (y,z) \in \mathcal{B}_2$. 
    		By the standard arguments, we show that $(y,z)$ solves \eqref{dfn:BSDE1}. 
    		The uniqueness of $(y,z)$ follows standard arguments with comparison theorem (see Lemma~\ref{lem:comparision}).  
	\end{proof}

    \subsubsection{Step 2: Comparison theorem and priori bounds.}
    %We introduce the notion of subsolutions and supersolutions under the general condition.%, i.e., random horizon and general unbounded $D$.
	\begin{dfn}\label{def:BSDE_sup_sub_sol}
		%Let $(\tau, D,\xi_{\tau})$ satisfy Condition~\ref{cond:triplet}. 
        %\color{red} Recall that we have set $\tau:=\tau^{x,D}$ and $\xi_{\tau}:=\xi_{\tau^{x,D}}$. \color{black}
        Adapted processes $(Y_t,Z_t)_{t\ge 0}$ are called a supersolution (resp. subsolution) to \eqref{dfn:BSDE1} %with $(\tau,\xi_{\tau},g,D)$~
        if 
		\begin{align*}
			Y_t = \xi_{\tau}^{1-R} + \int_t^{\tau} g_{\EZ}(s,Y_s)dD_s - \int_t^{\tau}Z_s dW_s + \int_t^{\tau} dA_s\quad \bigg(\text{resp.}-\int_{t}^{\tau} dA_s\bigg),
		\end{align*}
		where %$(A_t)_{t\ge0} \in \Lambda_{+}$ 
            $(A_t)_{t\ge0}$ is a right continuous left limit process in ${\cal P}$, and it holds that $(Y_t+ \int_0^{t \wedge\tau}g_{\EZ}(s,Y_s)dD_s)_{t\ge 0}$ is a local-supermartingale (resp. -{submartingale}). A solution of \eqref{dfn:BSDE1} is both a supersolution and a subsolution.
	\end{dfn}

The following result shows a comparison result for \eqref{dfn:BSDE1} without restricting ourselves to the conditions given by Lemma~\ref{lem:BSDE_exist_Lip}. As it is a direct consequence of Kobylanski~\cite[Theorem 2.6]{kobylanski2000backward}, we omit the proofs here.
	
\begin{lem}\label{lem:Bcomparision} Let $(Y,Z)$ (resp. $(\widetilde{Y},\widetilde{Z})$) be a supersolution (resp. subsolution) to \eqref{dfn:BSDE1} (resp.~to the BSDE obtained from \eqref{dfn:BSDE1} by replacing $(g_{\EZ},\xi_{\tau}^{1-R})$ with $(\widetilde g,\widetilde\xi_{\tau}^{1-R})$). Assume that both $(Y,Z)$ and $(\widetilde{Y},\widetilde{Z})$ are of class $\mathcal{B}_2$ and one of the following conditions hold:
		\begin{enumerate}[label=\roman*.]
			\item $\tilde{g}(t,Y_t) \le g(t,Y_t)$  and $\tilde{g}(t,y)$ is nonincreasing in $y$, $d\mathbb{P}\times dt$-a.e., or
			\item $\tilde{g}(t,\widetilde{Y}_t) \le g(t,\widetilde{Y}_t)$  and $g(t,y)$ is  nonincreasing in $y$, $d\mathbb{P}\times dt$-a.e..
		\end{enumerate}
		If $\xi_{\tau} \ge \tilde{\xi}_{\tau}$, then $\widetilde{Y}_t \le Y_t$ for any $t \ge 0$ $\mathbb{P}$-a.s.
\end{lem}

    We introduce two auxiliary processes that act as lower and upper bounds for the solution to the BSDE \eqref{dfn:BSDE1}. 

    %We introduce two auxiliary utility processes that act as lower and upper bounds for the solutions of our BSDE \eqref{dfn:BSDE1}. %The well-defined nature of these processe is established in the following lemmas.
	\begin{lem}[Upper Bound]\label{lem:overlineY} There exists a process $\overline Y\in {\cal P}$ given by 
		\begin{align*}
			\overline{Y}_t: = \bigg(\mathbb{E}_t\bigg[ \int_t^{\tau} e^{-\rho s} dD_s + \xi_{\tau}\bigg]\bigg)^{1-R}, \quad t \ge 0,
		\end{align*}
        such that $\mathbb{E}[\sup_{t\ge 0}\overline{Y}_t^{\frac{2}{1-R}}]<\infty$ (in particular $(\overline{Y}_t)_{t\ge 0}\in\mathcal{S}_2$.
		%where $R\in(0,1)$ is the risk-aversion coefficient in $g_{\text{EZ}} $ defined in \eqref{dfn:g_EZ}.
	\end{lem}
    {The proof is a standard application of Doob's maximal inequality.}
	%\begin{proof}
		% We observe that $(\overline{Y})^{\frac{1}{1-R}}$ is the classical dividend-paying utility process; that is $(\overline{Y})^{\frac{1}{1-R}}$ corresponds to $K^{x,D}$ given in \eqref{eq:std_singular_V} (noting that $(\tau,\xi_{\tau})=(\tau^{x,D},\xi_{\tau^{x,D}})\in{\cal T}\times L^2({\cal F}_\tau)$; see Definition~\ref{dfn:bankruptcy}).
  %       %BSDE with parameter $(\tau, (\xi_{\tau})^{\frac{1}{1-R}}, \overline{g},D)$, with $\overline{g}(t,v)=e^{-\rho t}$. 
  %       Since $D\in {\cal A}^x$ (see Definition~\ref{dfn:admissible}), by Doob's maximal inequality it holds that $\mathbb{E}[\sup_{t\ge 0}\overline{Y}_t^{\frac{2}{1-R}}]<\infty$. As a consequence of Jensen's inequality, it follows that $\overline{Y}\in\mathcal{S}_2$. 
	%\end{proof}
%The triplet $(T,\hat D,\hat \xi_T)$ satisfying conditions in Lemma~\ref{lem:BSDE_exist_Lip} is a special case of the above Lemma.
	\begin{lem}[Lower Bound]\label{lem:underlineY}  For the triplet $(T,\hat D,\hat \xi_T)$ satisfying the conditions in Lemma~\ref{lem:BSDE_exist_Lip}, there exists $(\underline{Y}_t)_{t\in[0,T]}\in\mathcal{S}_2$ satisfying
		\begin{align*}%\label{eq:overlineY}
			\underline{Y}_t: = \mathbb{E}_t\bigg[ \int_t^{T} \underline{g}(s,\underline{Y}_s) d\hat D_s + \hat \xi_T^{1-R} \bigg],\quad  \underline{g}(t,v) =e^{-\rho t}(1-R v),\quad t\in[0,T].
		\end{align*}
		In addition, $\underline{Y}_t \ge e^{-R \|\int_0^T e^{-\rho s} d\hat D_s \|_{\infty}} \mathbb{E}_t[\hat \xi^{1-R}_T]$ $\mathbb{P}$-a.s., for all $t\in[0,T]$.
	\end{lem}
     {The existence of $\underline{Y}$
  and the stated bound follow from a direct application of the comparison theorem of BSDEs in Lemma~\ref{lem:Bcomparision}.}
	% \begin{proof}
	% 	Since the aggregator $\underline{g}(\cdot,y)$ satisfies the condition in Lemma~\ref{lem:BSDE_exist_Lip}, we obtain the existence and uniqueness of the solution $(\underline Y,\underline Z)$ to the BSDE obtained from \eqref{eq:BSDE2} by replacing $g$ with $\underline g$. It follows from the comparison theorem in Lemma~\ref{lem:Bcomparision} that $\underline{Y}_t \ge \tilde{Y}_t$, where $(\tilde{Y},\tilde{Z})$ is the unique solution to %BSDE
	% 	\begin{align*}%\label{eq:lem:underlineY_1}
	% 		\tilde{Y}_t = \hat \xi_T^{1-R} - \int_t^{T} e^{-\rho s} R \tilde{Y}_s d\hat D_s - \int_{t}^{T} \tilde{Z}_t dW_s.
	% 	\end{align*}
	% 	Moreover, since $\tilde{Y}_t = \mathbb{E}_t[e^{-R \int_{t}^T e^{-\rho s}d\hat D_s}\hat \xi_T^{1-R} ] \ge e^{-R \|\int_0^T e^{-\rho s} d\hat D_s \|_{\infty}} \mathbb{E}_t[\hat \xi^{1-R}_T]$ holds $\mathbb{P}$-a.s., this completes the proof.  
	% \end{proof}

    \subsubsection{Step 3: EZ aggregator with fixed time horizon and bounded control.}

        %\vspace{0.3em}
	We start by approximating the EZ aggregator $g_{\text{EZ}}$ in \eqref{dfn:g_EZ} using Lipschitz aggregators that satisfy conditions~\ref{item:A1} and \ref{item:A2}. The proof of the construction of approximation is straightforward, so we omit it in the following lemma.
	\begin{lem}\label{lem:appro_agg_gn}
		The sequence of aggregators
		\begin{align*}%\label{eq:gn}
			g_n(\omega,t,y) = \inf_{x > 0 }\{g_{\text{EZ}}(\omega,t,x) + n R e^{-\rho t}(x-y)\}
		\end{align*}
		is well-defined for each $n\ge 1$, satisfying both \ref{item:A2} and the following:
		\begin{itemize}
			\labitem{i.}{item:prop_agg_i} Monotonicity in $n$: $\forall~y \ge 0$, $g_n(\omega,t,y)$ increase in $n$, $d\mathbb{P}\times dt$-a.s.;
			\labitem{ii.}{item:prop_agg_ii}  Lipshcitz condition \ref{item:A1}: $|g_n(\omega,t,y_1)-g_n(\omega,t,y_2)|\le nR e^{-\rho t}|y_1-y_2|$ for all $y_1, y_2\ge 0$, $d\mathbb{P}\times dt$-a.s.;
            % $$\forall~y_1, y_2\ge 0, ~|g_n(\omega,t,y_1)-g_n(\omega,t,y_2)|\le nR e^{-\rho t}|y_1-y_2|.$$
			\labitem{iii.}{item:prop_agg_iii}  Monotonicity in $y$: $g_n$ is nonincreasing in $y$, $d\mathbb{P}\times dt$-a.s..
		\end{itemize}
	\end{lem}
	\begin{lem}\label{lem:BSDE_exist_g}
		For the triplet $(T,\hat D,\hat \xi_T)$ satisfying the conditions in Lemma~\ref{lem:BSDE_exist_Lip}, we further assume $\hat \xi_{T}>C$ $\mathbb{P}$-a.s. for some constant $C$. The BSDE %\eqref{dfn:BSDE1} with parameter $(T,\xi_T,g_{\text{EZ}},D)$ 
        \begin{align}\label{eq:BSDE3}
        Y_t = \hat \xi_{T}^{1-R} + \int_{t}^{T} g_{\EZ}(s, Y_s) \, d\hat D_s - \int_{t}^{T} Z_s dW_s, \quad t \in[0,T],
        \end{align}
        has a unique solution $(y_t,z_t)_{t\in[0,T]}\in\mathcal{B}_2$.
	\end{lem}
	\begin{proof}
        %Throughout the proof, we consider arbitrary but fixed triplet $(T,\xi_T, D)$. 
        For given $(\omega,t)$, let $(g_n(\omega,t,y))_{n\in\mathbb{N}}$ be the sequence constructed in Lemma~\ref{lem:appro_agg_gn}.  It follows from Lemma~\ref{lem:BSDE_exist_Lip} that for each $n\ge1$, %the BSDE 
        \begin{align}\label{eq:BSDE3_1}
        y^n_t = \hat \xi_{T}^{1-R} + \int_{t}^{T} g_{n}(s, y^n_s) \, d\hat D_s - \int_{t}^{T} z^n_s dW_s, \quad t \in[0,T],
        \end{align}
        has a unique solution $(y_t^n,z_t^n)_{t\in[0,\tau]} \in \mathcal{B}_2$.  
		Lemma~\ref{lem:appro_agg_gn} \ref{item:prop_agg_i} and \ref{item:prop_agg_iii} and Lemma~\ref{lem:Bcomparision} yield that for $n\ge1$, $y_t^1(\omega)\le y_t^{n}(\omega)\le y_t^{n+1}$, $d\mathbb{P}\times dt$-a.s.
	Note that $y^1$ is identical to $\underline{Y}$ defined in Lemma~\ref{lem:underlineY}, %(with current triplet $(T,D,\xi_T)$), 
    since $g_1$ and $\underline{g}$ are identical. 
    Under the condition that $\hat \xi_{T}\in L^2(\mathcal{F}_T)$ and $\int_0^T e^{-\rho s} d\hat D_s\le K$ a.s., by Lemma~\ref{lem:underlineY} and the assumption that $\hat \xi_{T}>C$ $\mathbb{P}$-a.s., there is a constant $\hat C$ such that $y^1_t = \underline{Y}_t \ge e^{-R K} \mathbb{E}_t[\hat \xi_T^{1-R}] \ge \hat C $ for $t\in[0,T]$ a.s. {On the other hand, the upper bound $Y^n \le \overline{Y}$ is established by applying It\^o's formula to $\psi := (Y^n)^{1/(1-R)}$ and using a standard BSDE comparison argument.}

		% We show that $(y^n)_{n\in \mathbb{N}}$ is bounded from above by $\overline{Y}$ which is defined by %Lemma~\ref{lem:overlineY} (with current triplet $(T,D,\xi_T)$). 
  %       \begin{align*}
		% 	\overline{Y}_t: = \bigg(\mathbb{E}_t\bigg[ \int_t^{T} e^{-\rho s} d\hat D_s + \hat \xi_{T}\bigg]\bigg)^{1-R}, \quad t \in[0,T].
		% \end{align*}

  %       For $n\ge1$, a direct application of It\^{o}'s formula (cf. Jacod and Shiryaev~\cite[Theorem I.4.57]{jacod2013limit}) to $\psi:=(y^n)^{\frac{1}{1-R}}$ yields
		% \begin{align*}
		% 	d\psi_t =& -\frac{\psi_t^{R} g_n(t,\psi_t^{1-R}) }{1-R}d\hat D_t + \frac{(\psi_{t})^{R}z_t^n }{1-R}dW_t + \frac{R}{2(1-R)^2}(\psi_{t})^{2R-1}|z_t^n|^2 dt.
		% \end{align*}
		% Note that $\int_0^{\cdot} \frac{R}{2(1-R)^2}(\psi_{t})^{2R-1}|z_t^n|^2 dt$ is an increasing process, so $\psi$ is a subsolution of the BSDE obtained from \eqref{eq:BSDE3_1} by replacing $(\hat \xi_{T}^{1-R},g_n)$ with $(\hat \xi_T,h)$, where $h(t,y) = y^{R} g_n(t,y^{1-R})/(1-R )$ for $(t,y)\in \mathbb{R}_+\times\mathbb{R}$.  

		% Since $g_n\le g_{\text{EZ}}$, we have \(
  %           y^{R}\frac{g_n(t,y^{1-R})}{1-R} \le y^{R}\frac{g_{\text{EZ}}(t,y^{1-R})}{1-R} = e^{-\rho t}.\) Combining this with Lemma~\ref{lem:Bcomparision} yields that $\psi \le (\overline{Y})^{\frac{1}{1-R}}$; hence $y^n \le \overline{Y}$.
		
	{Now for each $n\ge1$, $\hat C \le y_t^1(\omega)\le y_t^{n}(\omega)\le y_t^{n+1}\le \overline{Y}_t$, $d\mathbb{P}\times dt$-a.s.. Hence there exists an $\mathcal{F}_t$-progressively measurable $y$ satisfying $\lim_{n\rightarrow\infty} y_t^n(\omega) = y_t(\omega)$, $d\mathbb{P}\times dt$-a.s.} %\mycomment{can you check this part?}
		Hence, $\mathbb{E}[\sup_{s\in[0,T]}|y_s|^2] \le \mathbb{E}[\sup_{s\in[0,T]}|\overline{Y}_s|^2]<\infty.$
        In order to take limit of $(z^n_t)_{t\in[0,T]}$, we apply It\^{o}'s lemma for $(y_t^n)^2$ to have
		\begin{align*}
			\mathbb{E}\bigg[ \int_{0}^{T} |z^n_s|^2 ds\bigg]&= \mathbb{E}[\hat \xi_T^{2(1-R)}] + 2 \mathbb{E}\bigg[\int_0^{T} y_s^ng_n(s,y_s^n)d  \hat D_s \bigg] \nonumber \\
			&\le \mathbb{E}[\hat \xi_T^{2(1-R)}] + 2 \mathbb{E}\bigg[\int_0^{T} y_s^ng_{\text{EZ}}(s,y_s^1)d \hat D_s \bigg] \\
            &\le \mathbb{E}[\hat \xi_T^{2(1-R)}]+ 2\hat C^{\frac{-R}{1-R}} \mathbb{E}\bigg[\sup_{s\in[0,T]}|\overline{Y}_s|^2\bigg]  \bigg\|\int_0^T e^{-\rho s} d \hat D_s \bigg\|_{\infty} <\infty,  \nonumber
		\end{align*}
  where we use the fact that $0\le g^{n}(t, y^n) \le g_{\text{EZ}}(t,y^n) \le g_{\text{EZ}}(t,y^1) \le \hat C^{\frac{-R}{1-R}}e^{-\rho t}$ for any $n$ and $t\in[0,T]$.
  Therefore, there exists $z\in\mathcal{M}_2$ and a sub sequence $(z^{n_j})_j$ of $(z^n)_n$ such that $z^{n_j} \rightharpoonup z$ weakly in $\mathcal{M}_2.$ 

Passing to a subsequence if necessary, we can show that the whole sequence $(z^n)_n$ converges strongly to $z$ in $\mathcal{M}_2$ by %Lebesgue's 
DCT. Moreover, $y$ is a continuous process because $y^n$ has continuous paths. %by construction. 
We can now pass the limit in \eqref{eq:BSDE3_1} 
         % \begin{align*}
         % y_t^n = \hat \xi_T^{1-R} + \int_{t}^T g_n(s,y^n_s)d D_s - \int_t^T z_s^n dW_s,
         % \end{align*}
so that $(y_t,z_t)_{t\in[0,T]}\in\mathcal{B}_2$ is a solution to \eqref{eq:BSDE3}, as claimed.  
	\end{proof}

    \subsubsection{Step 4: EZ aggregator with random time horizon \& unbounded control.}
	\begin{lem}\label{lem:BSDE_exist_g_tau}
		%\color{red} Recall that we have set $\tau:=\tau^{x,D}$ and $\xi_{\tau}:=\xi_{\tau^{x,D}}$. \color{black} 
        Assume that there is a constant $C>0$ such that $\xi_{\tau}>C$ $\mathbb{P}$-a.s.. Then the BSDE~\eqref{dfn:BSDE1} has a unique solution $(Y_t,Z_t)_{t\ge 0}\in\mathcal{B}_2$. In addition, it holds that $Y$ is continuous, $Y\ge C$ $\mathbb{P}$-a.s. and $\mathbb{E}[\sup_{t\ge 0} (Y_t)^{\frac{2}{1-R}}]<\infty$.
	\end{lem}
	\begin{proof}
		%{%The techniques are motivated by \cite{Aurand2023}, but here 
        %We aim to extend the result in Lemma~\ref{lem:BSDE_exist_g} to the case with random horizon and unbounded controls.}
  %       , but here we aim to extend the result in  to the case with both random horizon and unbounded controls. 
		% We utilize previous results on fixed horizon and a sequence of stopping times to construct a sequence of solutions that is Cauchy in $\mathcal{B}_2$,  and show that its limit solves \eqref{dfn:BSDE1} with parameter $(\tau,\xi_{\tau},g_{\text{EZ}},D)$.
		% The structure of the proof parallels that occurs in , but the details are more delicate since the previous estimates derived from bounded controls (such as \eqref{eq:thm:BSDE_exist_g:1}) are no longer valid for unbounded controls.
		
		For each $n\in\mathbb{N}$, we construct a solution $(Y_t^n,Z_t^n)_{t\ge0} \in \mathcal{B}_2$ to the BSDE
		\begin{align}\label{eq:Y_n_construct}
			Y_t^n = \xi_{\tau}^{1-R} + \int_{t \wedge {\tau}}^{{\tau}} \mathds{1}_{s\in[0,n\wedge\tau^n]} g_{\text{EZ}}(s,Y_s^n)dD_s - \int_{t \wedge {\tau}}^{{\tau}} Z_s^n dW_s, \quad t\ge0, 
		\end{align}
		where $\tau^n: = \inf\{t\ge 0: \int_0^{t}e^{-\rho s}d D_s \ge n \}\in\mathcal{T}$. 
		We denote by $\overline{\tau}^n := {\tau}\wedge\tau^n$. 
		Lemma~\ref{lem:BSDE_exist_g} implies that there exits a unique solution $(y_t,z_t)_{t\in[0,n]}\in\mathcal{B}_2$ to %for fixed horizon $[0,n]$ BSDE
		\begin{align*}
			y_t = \mathbb{E}_n[ \xi_{\tau}^{1-R}] + \int_{t}^{n} \mathds{1}_{s\in[0,\overline{\tau}^n]} g_{\text{EZ}}(s,y_s)dD_s - \int_t^n z_s  dW_s.
		\end{align*}
		The conditions in Lemma~\ref{lem:BSDE_exist_g} are satisfied as $ \int_{0}^n \mathds{1}_{s\in[0,\overline{\tau}^n]} e^{-\rho s}dD_s \le n$ $\mathbb{P}$-a.s.. Moreover, under the condition that $\xi_{\tau}^{1-R}\in L^2_+(\mathcal{F}_{\tau})$, the martingale representation theorem implies that there exists $(\eta_t)_{t\ge 0} \in \mathcal{M}_2$ such that 
		\begin{align*}
			\mathbb{E}_{t}[\xi_{\tau}^{1-R}] = \xi_{\tau}^{1-R}- \int_t^{{\tau}} \eta_s dW_s,\text{ on }\{t<{\tau}\};\quad \eta_t =  0   \text{ on }\{t >{\tau} \}.
		\end{align*}
		So we construct $(Y_t^n,Z_t^n)_{t\ge0}\in \mathcal{B}_2$ as
		\begin{align*}
			Y_t^n = y_t \mathds{1}_{t\in[0,n\wedge\overline{\tau}^n]} + \mathbb{E}_{t}[\xi_{\tau}]\mathds{1}_{t\in(n\wedge\overline{\tau}^n,\infty)},\quad Z_t^n =  z_t \mathds{1}_{t\in[0,n\wedge\overline{\tau}^n]} + \eta_t\mathds{1}_{t\in(n\wedge\overline{\tau}^n,\infty)}.
		\end{align*}
		Next, we aim to show that $(Y^n,Z^n)_{n\in\mathbb{N}}$ is Cauchy in $\mathcal{B}_2$. 
        
        For any $m>n$, let $\Delta Y_t: =Y_t^m - Y_t^n$, $\Delta Z_t : =Z_t^m-Z_t^n$. Then it suffices to show that $|(\Delta Y, \Delta Z)|_{\mathcal{B}_2}\rightarrow 0$ as $m,n\rightarrow \infty$.
	
		For $0\le t \le n$, by \eqref{eq:Y_n_construct} and the fact that $\tau^{n}(\omega)\le  \tau^{m}(\omega)$ we have 
		\begin{align*}
			&\Delta Y_{t}-\Delta Y_{n\wedge {\tau}}\\
            &\quad = \int_{t\wedge \overline{\tau}^n}^{n\wedge \overline{\tau}^n} \Delta g_{\text{EZ}}^{m,n}(s)dD_s +   \int_{n\wedge \overline{\tau}^n}^{n\wedge \overline{\tau}^m}  g_{\text{EZ}}(s,Y_s^m) dD_s- \int_{t\wedge {\tau}}^{n\wedge {\tau}} \Delta Z_s dW_s,
		\end{align*}
        where $\Delta g_{\text{EZ}}^{m,n}(s): =g_{\text{EZ}}(s,Y_s^m) - g_{\text{EZ}}(s,Y_s^n)$.
		It follows from It\^{o}'s lemma that 
		\begin{align*}
			&|\Delta Y_{t\wedge{\tau}}|^2 +  \int_{t\wedge{\tau}}^{n\wedge {\tau}} |\Delta Z_s|^2 ds =  2  \int_{t\wedge\overline{\tau}^n}^{n\wedge \overline{\tau}^n} \Delta Y_s \Delta g_{\text{EZ}}^{m,n}(s) dD_s \\
			&\quad + 2 \int_{n\wedge \overline{\tau}^n}^{n\wedge \overline{\tau}^m}  \Delta Y_s g_{\text{EZ}}(s,Y_s^m) dD_s- 2 \int_{t\wedge{\tau}}^{n\wedge{\tau}} \Delta Y_s \Delta Z_s dW_s +| \Delta Y_{n\wedge {\tau}}|^2 . 
		\end{align*}
		By the fact that $g_{\text{EZ}}$ is non-increasing in its second argument, so that $\Delta g_{\text{EZ}}^{m,n}(s)$ should have opposite sign of $\Delta Y_s$. This, together with the non-decreasing property of $D$, gives  $\int_{t\wedge \overline{\tau}^n}^{n\wedge \overline{\tau}^n} \Delta Y_s \Delta g_{\text{EZ}}^{m,n}(s)dD_s\le0$. 
		\begin{align*}
			&\int_{n\wedge \overline{\tau}^n}^{n\wedge \overline{\tau}^m}  \Delta Y_s g_{\text{EZ}}(s,Y_s^m) dD_s = 	\int_{n\wedge \overline{\tau}^n}^{n\wedge \overline{\tau}^m}  \Delta Y_s  \big(\Delta g_{\text{EZ}}^{m,n}(s)+ g_{\text{EZ}}(s,Y^n_s)\big)  dD_s \\
			&\quad \le 	\int_{n\wedge \overline{\tau}^n}^{n\wedge \overline{\tau}^m}  \Delta Y_s   g_{\text{EZ}}(s,Y^n_s)  dD_s = \int_{n\wedge \overline{\tau}^n}^{n\wedge \overline{\tau}^m}  \Delta Y_s g_{\text{EZ}}(s,\mathbb{E}_s[\xi])  dD_s\\
            &\quad \le 2 \sup_{t\in[0,{\tau}]}|\overline{Y}_t| C_4^{\frac{-R}{1-R}} \int_{n\wedge \overline{\tau}^n}^{n\wedge \overline{\tau}^m} e^{-\rho s} dD_s ,
		\end{align*}
		where the second line follows by construction that $Y^n_s = \mathbb{E}_{s}[\xi]$ for $s>n\wedge\tau^n$, the last line follows by the fact that $\xi_{\tau} \ge C_4$ $\mathbb{P}$-a.s. for some constant $C_4>0$, as well as the priory bound $Y^m$, $Y^n\le \overline{Y}$ given in Lemma~\ref{lem:overlineY}. 	%We then conclude 
        Hence,
		\begin{align*}
			&|\Delta Y_{t\wedge{\tau}}|^2 +  \int_{t\wedge {\tau}}^{n\wedge {\tau}} |\Delta Z_s|^2 ds\\
            &\quad \le | \Delta Y_{n\wedge {\tau}}|^2 +2  \sup_{t\in[0,{\tau}]}|\overline{Y}_t| C_4^{\frac{-R}{1-R}} \int_{n\wedge \overline{\tau}^n}^{n\wedge \overline{\tau}^m} e^{-\rho s} dD_s - 2\int_{t\wedge {\tau}}^{n\wedge {\tau}} \Delta Y_s \Delta Z_s dW_s,
		\end{align*}
		where the above local martingale is a uniformly integrable martingale. %\mycomment{perhaps... we can say that ``the above local maringale is a uniformly intebrable martingale''.}
        Hence,
		\begin{align}
            \mathbb{E}\bigg[\int_{0}^{n\wedge {\tau}} |\Delta Z_s|^2 ds \bigg] &\le \mathbb{E}[ | \Delta Y_{n\wedge  {\tau}}|^2] + 2 C_4^{\frac{-R}{1-R}}  \mathbb{E}\bigg[ \sup_{t\in[0,{\tau}]}|\overline{Y}_t| \int_{n\wedge \overline{\tau}^n}^{n\wedge \overline{\tau}^m} e^{-\rho s} dD_s\bigg] \nonumber \\
			&=: \mathbb{E}[ | \Delta Y_{n\wedge {\tau}}|^2] + 2C_4^{\frac{-R}{1-R}} \operatorname{I}^{n,m}. \label{eq:temp_est_YD}
            %&\le \mathbb{E}[ | \Delta Y_{n\wedge {\tau}}|^2] + 4C_4^{\frac{-R}{1-R}} \mathbb{E}\bigg[\sup_{t\in[0,{\tau}]}|\overline{Y}_t|^2\bigg]^{\frac{1}{2}}  \mathbb{E}\bigg[\bigg( \int_{n\wedge \overline{\tau}^n}^{n\wedge \overline{\tau}^m}e^{-\rho s}  dD_s \bigg)^2\bigg]^{1/2}\\
			%&\le  \mathbb{E}[ | \Delta Y_{n\wedge {\tau}}|^2] + 4 C_5(m-n);
            \end{align}
            Moreover, it holds that 
            \begin{align}
            \operatorname{I}^{n,m}&\leq \mathbb{E}\bigg[\sup_{t\in[0,{\tau}]}|\overline{Y}_t|^2\bigg]^{\frac{1}{2}}  \mathbb{E}\bigg[\bigg( \int_{n\wedge \overline{\tau}^n}^{n\wedge \overline{\tau}^m}e^{-\rho s}  dD_s \bigg)^2\bigg]^{1/2}\leq C_5(m-n).\label{eq:temp_est_YD_1}
            \end{align}
            where the constant $C_5>0$ (not depending on $m$,$n$,$\Delta Y$,$\Delta Z$) exists because $D\in{\cal A}$ has continuous path.
            In a similar manner, we have that 
            \begin{align}
			%&\mathbb{E}\Big[ \sup_{t\in[0,n]} |\Delta Y_{t\wedge {\tau}}|^2  \Big] \le \mathbb{E}\big[ | \Delta Y_{n\wedge {\tau} }|^2\big] + 2\mathbb{E}\bigg[\sup_{t\in[0,n]}\int_{0}^{t\wedge  {\tau}} \Delta Y_s \Delta Z_s dW_s \bigg]+ 4 C_5 (m-n)\\
            &\mathbb{E}\bigg[ \sup_{t\in[0,n]} |\Delta Y_{t\wedge {\tau}}|^2  \bigg]\nonumber\\
            &\quad \le \mathbb{E}[ | \Delta Y_{n\wedge {\tau} }|^2]+2C_4^{\frac{-R}{1-R}} \operatorname{I}^{n,m}+ 2\mathbb{E}\bigg[\sup_{t\in[0,n]}\int_{0}^{t\wedge  {\tau}} \Delta Y_s \Delta Z_s dW_s \bigg]\nonumber\\
            &\quad =: \mathbb{E}[ | \Delta Y_{n\wedge {\tau} }|^2]+2C_4^{\frac{-R}{1-R}} \operatorname{I}^{n,m}+2 \operatorname{II}^{n,m},\label{eq:temp_est_YD_2}
			 %&\le \mathbb{E}\big[ | \Delta Y_{n\wedge {\tau}}|^2\big] + 2C_6 \mathbb{E}\bigg[\bigg(\int_{0}^{n\wedge {\tau}} |\Delta Y_s|^2 |\Delta Z_s|^2 ds \bigg)^{1/2} \bigg]+ 4 C_5(m-n)\\
   %            &  \le \mathbb{E}\big[ | \Delta Y_{n\wedge {\tau}}|^2\big] +  \mathbb{E}\bigg[\Big(\sup_{t\in[0,n]}|\Delta Y_{t\wedge\tau}|^2 \Big)^{1/2} \bigg(4C_6^2\int_{0}^{n\wedge {\tau}} |\Delta Z_s|^2 ds \bigg)^{1/2} \bigg]+ 4 C_5(m-n)\\
			% &  \le \mathbb{E}\big[ | \Delta Y_{n\wedge  {\tau}}|^2\big] + \frac{1}{2}\mathbb{E}\Big[ \sup_{t\in[0,n]} |\Delta Y_{t\wedge {\tau}}|^2  \Big] + 2C_6^2 \mathbb{E}\bigg[\int_{0}^{n\wedge  {\tau}} |\Delta Z_s|^2 ds \bigg]+ 4 C_5(m-n)
		\end{align}
        where $\operatorname{I}^{n,m}$ is given in \eqref{eq:temp_est_YD}, and that
        \begin{align}
            \operatorname{II}^{n,m}
            &\leq C_6 \mathbb{E}\bigg[\bigg(\int_{0}^{n\wedge {\tau}} |\Delta Y_s|^2 |\Delta Z_s|^2 ds \bigg)^{1/2} \bigg]\nonumber \\
            &\leq  \frac{1}{4}\mathbb{E}\Big[ \sup_{t\in[0,n]} |\Delta Y_{t\wedge {\tau}}|^2  \Big] + C_6^2 \mathbb{E}\bigg[\int_{0}^{n\wedge  {\tau}} |\Delta Z_s|^2 ds \bigg],\label{eq:temp_est_YD_3}
        \end{align}
        where the first inequality follows from the inequality of arithmetic and geometric means, and the constant $C_6>0$ (not depending on $m$,$n$,$\Delta Y$,$\Delta Z$) exists by the Burkholder-Davis-Gundy's inequality. %\mycomment{pls check.}
        %for some positive constants $C_5$ and $C_6$ not depending on $m$, $n$, $\Delta Y$ or $\Delta Z$. The existence of $C_5$ follows the fact that $D\in{\cal A}$ has continuous path. The existence of $C_6$ follows by Burkholder-Davis-Gundy's inequality. %\mycomment{pls check the estimate. seems the constant is not correct.}
        
        Combining \eqref{eq:temp_est_YD}, \eqref{eq:temp_est_YD_1}, \eqref{eq:temp_est_YD_2}, and \eqref{eq:temp_est_YD_3}, we have some constant $C_7>0$ (that does not depend on $m$,$n$) such that 
		\begin{align*}
			\mathbb{E}\bigg[ \sup_{0\le t\le n}| \Delta Y_{t\wedge {\tau}}|^2+\int_{0}^{n\wedge {\tau}} |\Delta Z_s|^2 ds \bigg]  \le C_7 \Big(\mathbb{E}\big[ | \Delta Y_{n\wedge {\tau}}|^2\big]+(m-n)\Big).
		\end{align*}
  Next, for $n<t \le m$, we have $\Delta Y_{t}%& = \int_{t \wedge \overline{\tau}^m}^{m\wedge \overline{\tau}^m} g_{\text{EZ}}(s,Y_s^m)dD_s - \int_{t\wedge {\tau}}^{{\tau}} {\Delta}Z_s dW_s\\
            =  \int_{t \wedge \overline{\tau}^n}^{m\wedge  \overline{\tau}^m} g_{\text{EZ}}(s,Y_s^m)dD_s - \int_{t\wedge {\tau}}^{m\wedge {\tau}} {\Delta}Z_s dW_s,$
		% \begin{align*}
		% 	\Delta Y_{t}%& = \int_{t \wedge \overline{\tau}^m}^{m\wedge \overline{\tau}^m} g_{\text{EZ}}(s,Y_s^m)dD_s - \int_{t\wedge {\tau}}^{{\tau}} {\Delta}Z_s dW_s\\
  %           =  \int_{t \wedge \overline{\tau}^n}^{m\wedge  \overline{\tau}^m} g_{\text{EZ}}(s,Y_s^m)dD_s - \int_{t\wedge {\tau}}^{m\wedge {\tau}} {\Delta}Z_s dW_s, 
		% \end{align*}
		where the equality holds by the fact that  $\Delta Z_s = 0$ on $s > m$ since $Z^m_s = Z^n_s =\eta_s$ for $s>m\wedge{\tau} \ge m \wedge \overline{\tau}^m$.  Similar argument as for the case $0\le t \le n$ gives the existence of $C_8>0$ independent of $m$, $n$ such that
		\begin{align*}
			\mathbb{E}\bigg[ \sup_{n\le t\le m}| \Delta Y_{t\wedge {\tau}}|^2+\int_{n \wedge {\tau}}^{m \wedge {\tau}} |\Delta Z_s|^2 ds \bigg]  \le C_8 (m-n).
		\end{align*}
    
  %       It then follows from It\^{o}'s lemma and similar arguments as above that 
		% \begin{align*}
		% 	&|\Delta Y_{t\wedge{\tau}}|^2 + \int_{t\wedge {\tau}}^{m\wedge{\tau}} |{\Delta}Z_s|^2 ds\\
  %           &\quad = 2 \int_{t\wedge \overline{\tau}^m}^{m\wedge\overline{\tau}^m} \Delta Y_s g_{\text{EZ}}(s,Y^m_s) dD_s - 2\int_{t\wedge {\tau}}^{m\wedge {\tau}} \Delta Y_s {\Delta}Z_s dW_s\\
		% 	&\quad \le 2 \int_{t\wedge \overline{\tau}^m}^{m\wedge\overline{\tau}^m} |\Delta Y_s| g_{\text{EZ}}(s,\mathbb{E}_s[\xi_{\tau}]) dD_s - 2\int_{t\wedge {\tau}}^{m\wedge {\tau}} \Delta Y_s {\Delta}Z_s dW_s.
		% \end{align*}
		% Note that $n\wedge \overline{\tau}^m \le t\wedge \overline{\tau}^m  \le m\wedge \overline{\tau}^m $, therefore similar argument as for the case $0\le t \le n$ gives the existence of $C_8>0$ independent of $m$, $n$ such that
		% \begin{align*}
		% 	\mathbb{E}\bigg[ \sup_{n\le t\le m}| \Delta Y_{t\wedge {\tau}}|^2+\int_{n \wedge {\tau}}^{m \wedge {\tau}} |\Delta Z_s|^2 ds \bigg]  \le C_8 (m-n).
		% \end{align*}
		
    \noindent Finally, $t>m\wedge{\tau}\ge m\wedge\overline{\tau}^m$, $Y_t^m = Y_t^n = \mathbb{E}_{t\wedge{\tau}}[\xi_{\tau}^{1-R}]$ and $Z_t^m = Z_t^n = \eta_t$, and
		\begin{align}\label{eq:mn_cauchy_1}
			\mathbb{E}\bigg[ \sup_{t\ge 0}| \Delta Y_{t\wedge {\tau}}|^2+\int_{0}^{{\tau}} |\Delta Z_s|^2 ds \bigg] \le C_7 \mathbb{E}\big[ | \Delta Y_{n\wedge {\tau}}|^2\big] + C_9(m-n),
		\end{align}
            with $C_9:=(C_7+C_8)$.
		% It follows from the prior estimate in Lemma~\ref{lem:overlineY} that 
		% \begin{align*}
		% 	\frac{1}{2}|\Delta Y_{n\wedge{\tau}}|^2 \le (Y^m_{n\wedge{\tau}})^2 + (Y^n_{n\wedge{\tau}})^2 \le 2\sup_{t\in[0,{\tau}]}|\overline{Y}_t|^2.
		% \end{align*}
		By the DCT, we have
		\begin{align*}
			\lim\limits_{n,m\rightarrow\infty} \mathbb{E}\big[ | \Delta Y_{n\wedge {\tau}}|^2\big] &= \mathbb{E}\Big[ 	\lim\limits_{n,m\rightarrow\infty}|Y^m_{n\wedge {\tau}} - Y^n_{n\wedge {\tau}}|^2 \Big]  \\
            &=\mathbb{E}\Big[ 	\lim\limits_{n\rightarrow\infty}|Y_{n\wedge\tau}-Y^n_{n\wedge {\tau}}|^2\Big] = \mathbb{E}\big[ | \xi_{\tau} - \xi_{\tau}|^2\big] = 0. 
		\end{align*}
		Since the second term in \eqref{eq:mn_cauchy_1} vanishes as $m,n\rightarrow\infty$, we conclude that $\|(\Delta Y,\Delta Z)\|_{\mathcal{B}_2} \rightarrow 0$ as $m$, $n\rightarrow\infty$, i.e., $(Y^n,Z^n)_n$ is Cauchy in $\mathcal{B}_2$. Since $\mathcal{B}_2$ is complete, $\lim_{n\rightarrow\infty}(Y^n,Z^n) = (Y,Z)\in\mathcal{B}_2$ exists.
        
        Finally, we show the limit $(Y,Z)$ solves \eqref{dfn:BSDE1}. For each $n\in\mathbb{N}$, $(Y^n,Z^n) \in \mathcal{B}_2$ solves \eqref{eq:Y_n_construct}. 
		% \begin{align*}
		% 	Y_t^n = \xi_{\tau} + \int_{t \wedge {\tau}}^{{\tau}} \mathds{1}_{s\in[0,n\wedge\tau^n]} g_{\text{EZ}}(s,Y_s^n)dD_s - \int_{t \wedge {\tau}}^{{\tau}} Z_s^n dW_s, \quad t\ge0.
		% \end{align*}
    It is not difficult to pass the limit $n\rightarrow\infty$ in the above equation and show that each term of it converges to a corresponding term in \eqref{dfn:BSDE1} for almost all $\omega\in\Omega$ uniformly in $t\ge0$, thus we omit the details here.  
    \end{proof}

    %%%%%%%%%%%%%%%%%%%%%%%%%%%%%%%%%%%%%%%%%%%%%%%%%%%%%%%%%%%%%%%%%%%%%%%%%%%%%%%%%%%%%%%%%%%%%%%%%%%%%%%%%%%%%%%%%%%%%%%%%%%%%%
    \subsubsection{Proof of Proposition \ref{pro:recurs_BSDE}}\label{proof:final:pro:recurs_BSDE}
    \begin{proof}%\emph{of Theorem \ref{thm:recurs_BSDE}\;i.} %\vspace{0.5em}%\noindent{\it Proof of Theorem \ref{thm:recurs_BSDE}\;i.} 
    Fix $(x,D)\in \mathbb{R}_+\times {\cal A}^x$. Set $V:=V^{x,D}$ %and simply %write $V$ %to simplify notation. 
     for notational convenience. \color{black} 
   For each $n\in\mathbb{N}$, set $\xi_{\tau}^{(n)} := \frac{1}{n}\vee \xi_{\tau}$. By Lemma~\ref{lem:BSDE_exist_g_tau}, there exists a unique $L^2$-solution $(Y^n,Z^n)\in {\cal B}_2$ of \eqref{dfn:BSDE1} (with $\xi_\tau$ replaced by~$\xi_{\tau}^n$) such that
		\begin{align*}
			Y^n_t = \mathbb{E}_t \bigg[ \int_{t\wedge \tau}^{\tau} e^{-\rho s}(1-R) 	(Y^n_s)^{\frac{-R}{1-R}}dD_s + (\xi^{(n)}_{\tau})^{1-R}\bigg],\quad  t\ge0.
		\end{align*} 
        Note that $Y_t^n$ is non-increasing in $n \in \mathbb{N}$ for every $t\geq0$ (see Lemma~\ref{lem:Bcomparision}) and $Y^n_t > 0$ $\mathbb{P}$-a.s.~for every $t\geq0$ and $n\in \mathbb{N}$. Hence $V_t:= \lim_{n\rightarrow \infty} Y_t^n$ for $t\geq0$ is well-defined. Moreover, 
        the monotone convergence theorem (MCT) implies %that for $t\geq 0$,
		\begin{align*}
			V_t =& \lim_{n\rightarrow\infty}\mathbb{E}_t \bigg[ \int_{t\wedge \tau}^{\tau} e^{-\rho s}(1-R) 	(Y^n_s)^{\frac{-R}{1-R}}dD_s + (\xi^{(n)}_{\tau})^{1-R}\bigg] \\
			=& \mathbb{E}_t \bigg[ \int_{t\wedge \tau}^{\tau} e^{-\rho s}(1-R) 	(V_s)^{\frac{-R}{1-R}}dD_s + \xi_{\tau}^{1-R}\bigg],\quad t\geq 0.
		\end{align*} 
        Since $\xi_{\tau}>0$ $\mathbb{P}$-a.s. (see \eqref{def:bankrupcypayment}), %it holds that 
        $V_t \ge \mathbb{E}_t[\xi_{\tau}^{1-R}]>0$ $\mathbb{P}$-a.s. for all $t\geq0$. 
        
        We claim that $V\in\mathbb{UI}(g_{\text{EZ}},D,\tau,\xi
        _\tau^{1-R})$. Since $Y_0^n$ is non-increasing in $n \in \mathbb{N}$ and $V_0=\lim_{n\rightarrow \infty}Y^n_0$, it holds that
        \[
        \mathbb{E}\bigg[ \int_{0}^{\tau} e^{-\rho s}(1-R) 	(V_s)^{\frac{-R}{1-R}}dD_s\bigg] <  V_0 \le Y_0^n<\infty\quad \mbox{for every $n\in\mathbb{N}$.}
        \]
        
        Similarly, we have $\mathbb{E}[\sup_{t\ge 0} (V_t)^{\frac{2}{1-R}}] \le \mathbb{E}[\sup_{t\ge 0} (Y^n_t)^{\frac{2}{1-R}}]<\infty$. 
        This ensures that $\mathbb{E}[\sup_{t\ge 0} (V_t)^2]<\infty$ by Jensen's inequality with exponent $\frac{1}{1-R}>1$. Hence, $V$ is uniformly integrable and therefore is of class $\mathbb{UI}(g_{\text{EZ}},D,\tau,\xi
        _\tau^{1-R})$.

        It remains to show that $V$ is the unique utility process.
        By the martingale representation theorem, there exists $Z\in \mathcal{P}$ such that $\int_0^{\tau} Z_t^2 dt <\infty$ $\mathbb{P}$-a.s. and $(V,Z)$ solves %the BSDE given in 
        \eqref{dfn:BSDE1}. Let $V'$ be another utility process of class $\mathbb{UI}(g_{\text{EZ}},D)$ and %there exists 
        $Z'\in \mathcal{P}$ be such that $\int_0^{\tau} (Z'_t)^2 dt <\infty$ $\mathbb{P}$-a.s., and $(V',Z')$ solves \eqref{dfn:BSDE1}.
        Here we note that $V_t = V'_t = \xi_{\tau}$ on $\{t\ge \tau\}$. 
        
        Set $(\Delta V, \Delta Z):= (V-V', Z-Z')$. Then ($\Delta V, \Delta Z)$ solves for all $t\geq0$
        \begin{align*}
            \Delta V_t =  \int_{t\wedge\tau}^{\tau} \frac{g_{\text{EZ}}(s,V_s) - g_{\text{EZ}}(s,V'_s)}{\Delta V_s} \Delta V_s \mathds{1}_{\Delta V_s\neq0} dD_s -  \int_{t\wedge\tau}^{\tau} \Delta Z_s dW_s.
        \end{align*}
        Let $\alpha=(\alpha_s)_{s\ge0}$ be defined by $\alpha_s: =  \mathds{1}_{\Delta V_s\neq0}(g_{\text{EZ}}(s,V_s) - g_{\text{EZ}}(s,V'_s))/\Delta V_s$. 
        Since $y\mapsto g_{\text{EZ}}(\cdot,y)$ is decreasing, we have $\alpha\le0$. It then follows that 
        \begin{align*}
            e^{\int_{0}^t \alpha_u du}\Delta V_t =  -  \int_{t\wedge\tau}^{\tau} e^{\int_{0}^s \alpha_u du} \Delta Z_s dW_s\quad \mbox{for all $t\geq0$}.
        \end{align*}
        From the fact that $\alpha\le 0$ and $V,V'\in\mathbb{UI}(g_{\text{EZ}},D)$, the local martingale $\int_0^{t}e^{\int_{0}^s \alpha_u du} \Delta Z_s dW_s$ is a martingale. 
        Hence, $\Delta V_t= 0$ $\mathbb{P}$-a.s.~for all $t\geq0$. %This completes the proof.
         
    \end{proof}

    %%%%%%%%%%%%%%%%%%%%%%%%%%%%%%%%%%%%%%%%%%%%%%%%%%%%%%%%%%%%%%%%%%%%%%%%%%%%%%%%%%%%%%%%%%%%%%%%%%%%%%%%%%%%%%%%%%%%%%%%%%%%%%%%%%%%%%%%%%%%%%%%%%%%%%%%%%
    \subsection{Proof of Proposition \ref{pro:Maenhoutdef}}\label{proof:pro:Maenhoutdef}
    \begin{proof}
    Fix $(x,D,\theta)\in \mathbb{R}_+\times {\cal A}^x\times \Theta^D$ (see Definitions \ref{dfn:admissible} and \ref{dfn:admissD}). Set $\tau:=\tau^{x,D}$, $\xi_{\tau}:=\xi_{\tau^{x,D}}$ (see Definition \ref{dfn:bankruptcy}), and ${\cal V}^{\theta}:={\cal V}^{x,D,\theta}$ for notational simplicity.
    
    By H\"older's inequality, %with exponent~2,
    \(
    \mathbb{E}[\eta_{\tau}^{\theta}e^{\int_{0}^{\tau} \frac{\theta_u^2}{2\cR}du }\xi_{\tau}]\le  \mathbb{E}[(\eta_{\tau}^{\theta})^2e^{\int_{0}^{\tau} \frac{\theta^2_u}{\cR}du } ]^{\frac{1}{2}} \mathbb{E}[\xi_{\tau}^{{2}}]^{\frac{1}{2}}<\infty. 
    \)
    Furthermore, Jensen's inequality (with exponent 2) ensures that 
    $$\mathbb{E}\bigg[ \int_{0}^{ \tau} \eta_s^{\theta} e^{\int_{0}^{s} \frac{\theta^2_u}{2\cR}du -\rho s} dD_s \bigg]\leq \mathbb{E}\bigg[\bigg( \int_{0}^{ \tau} \eta_s^{\theta} e^{\int_{0}^{s} \frac{\theta^2_u}{2\cR}du -\rho s} dD_s\bigg)^2 \bigg] <\infty.$$
    Hence, the representation 
    \begin{align}\label{eq:thm:rob_rec:0}
			{\cal Y}_t^{\theta}:= \mathbb{E}^{\theta}_t\bigg[ \int_{t\wedge\tau}^{ \tau} e^{\int_{t\wedge \tau}^{s} \frac{\theta^2_u}{2\cR}du -\rho s} dD_s + e^{\int_{t\wedge\tau}^{\tau} \frac{\theta^2_u}{2\cR}du }\xi_{\tau}\bigg],
    \end{align}
    for $t\ge0$, is well-defined. 
    
    By Jensen's inequality, we have $\mathbb{E}[\sup_{t\ge 0} ({\cal Y}^{\theta}_{t\wedge\tau})^2]<\infty$.
    It implies that there exists $\iota_t^{\theta}\in\mathcal{P}$ such that $\int_0^{\infty} (\iota_t^{\theta})^2 dt<\infty$, $\mathbb{Q}^\theta$-a.s., and for $t\geq 0$ %$({Y}^{x,\theta,D},{\iota}^{x,\theta,D})$ solves
 		\begin{align}\label{eq:thm:rob_rec:1}
			{\cal Y}^{\theta}_t = \xi_{\tau} + \int_{t\wedge\tau}^{\tau}\bigg( e^{-\rho s} dD_s + \frac{1}{2\cR}{\cal Y}^{\theta}_s\theta^2_s  ds\bigg) - \int_{t\wedge\tau}^{\tau} {\iota}_s^{\theta} dW_s^{\theta}.
		\end{align}
    Since $(M_t)_{t\geq 0}:=(\int_{0}^{t}  {\iota}_s^{\theta}dW_s^{\theta})_{t\ge0}$ is a local martingale under $\mathbb{Q}^{\theta}$, we shall find an increasing sequence of stopping times $\{T_n\}_{n\in\mathbb{N}}$ such that $T_n\ge t$, $\mathbb{Q}^{\theta}(\lim_{n\rightarrow\infty} T_n = \infty) = 1$ and the stopped process $(M_{t\wedge T_n})_{t\ge 0}$ is a martingale under $\mathbb{Q}^{\theta}$. Applying $\mathbb{E}_t^{\theta}$ to both sides of \eqref{eq:thm:rob_rec:1} gives%-\mycomment{should be $V^{\theta,D}..?$}
    \begin{align*}
        {\cal Y}_{t}^{\theta} =  \mathbb{E}^{\theta}_t\bigg[ \int_{t\wedge\tau}^{T_n\wedge\tau} \bigg(e^{-\rho s} dD_s%\bigg] 
        +	\frac{1}{2\cR} %\mathbb{E}^{\theta}_t\bigg[
        %\int_{t\wedge\tau^{x,D}}^{T_n\wedge\tau^{x,D}}  
        {\cal Y}^{\theta}_s \theta_s^2 ds\bigg)\bigg]+ \mathbb{E}^{\theta}_t[{\cal Y}_{T_n\wedge\tau}^{\theta}]. 
    \end{align*}
    Then %By sending $n\rightarrow\infty$, 
    the first two terms on the right-hand side converge as $n\to \infty$ by MCT, while the last term converges by the dominated convergence theorem (DCT). We conclude that ${\cal Y}^{\theta}$ defined via \eqref{eq:thm:rob_rec:0} satisfies \eqref{dfn:rbs_V_theta}. The uniqueness follows from standard localization arguments, and we omit it here.
     
    \end{proof}

    %%%%%%%%%%%%%%%%%%%%%%%%%%%%%%%%%%%%%%%%%%%%%%%%%%%%%%%%%%%%%%%%%%%%%%%%%%%%%%%%%%%%%%%%%%%%%%%%%%%%%%%%%%%%%%%%%%%%%%%%%%%%%%%%%%%%%%%%%%%%%%%%%%%%%%%%%%
    \subsection{Proof of Theorem \ref{thm:rob_rec}}\label{proof:thm:rob_rec}
    
    \begin{lem}\label{lem:rob_rec_strict} 
            Suppose that Assumption \ref{as:recurrent} holds. %, in addition, 
             For any $(x,D)\in \mathbb{R}_+\times {\cal A}^x$, we assume that $\xi_{\tau^{x,D}}>C$ $\mathbb{P}$-a.s. with $C>0$, denote by $(V^{x,D},Z^{x,D})$ the unique $L^2$-solution to \eqref{dfn:BSDE1} (see Lemma \ref{lem:BSDE_exist_g_tau} and Definition \ref{def:BSDEdef}), and set $\mathbb{Y}^{x,D} := (V^{x,D})^{\frac{1}{1-R}}$ and ${\mathbb{Z}}^{x,D} := (\mathbb{Y}^{x,D})^{R}{Z}^{x,D}/(1-R)$. Then for any $\theta\in \Theta^D$ %\comment{kh: it holds for every $\theta\in \Theta^D$?}   
            \begin{align}\label{eq:thm:rob_rec_link:2}
                {\cal V}^{x,D,\theta}_t = \mathbb{Y}_t^{x,D} + \mathbb{E}_t^{\theta}\bigg[\int_{t\wedge \tau^{x,D}}^{\tau^{x,D}} \frac{R}{2{\mathbb{Y}}^{x,D}_s} \Big({\mathbb{Z}}^{x,D}_s + \frac{{\mathbb{Y}}^{x,D}_s}{R}\theta_s\Big)^2 e^{\frac{1}{2R}\int_t^{s}\theta^2_u du } d s  \bigg].%\quad \theta\in\theta^D.
            \end{align}
           Consequently, ${\cal V}^{x,D} \ge \mathbb{Y}^{x,D}$. In addition, $\theta^{*;x,D}:= (-R\mathbb{Z}^{x,D}_t/\mathbb{Y}^{x,D}_t)_{t\ge0} \in \Theta^D$, which implies ${\cal V}^{x,D} = {\cal V}^{x,D,\theta^{*;x,D}} = \mathbb{Y}^{x,D}$, %^ = (V^{x,D})^{\frac{1}{1-R}}$. 
           and 
           $({\cal V}^{x,D},{\mathbb Z}^{x,D})$ solves
\begin{align}\label{eq:Vrob_BSDE}
    {\cal V}^{x,D}_t = \xi_{\tau^{x,D}} + \int_{t\wedge\tau^{x,D}}^{\tau^{x,D}} e^{-\rho s} dD_s - \int_{t\wedge\tau^{x,D}}^{\tau^{x,D}} \frac{R}{2} \frac{({\mathbb Z}^{x,D}_s)^2}{{\cal V}^{x,D}
    _s} ds -  \int_{t\wedge\tau^{x,D}}^{\tau^{x,D}} \mathbb{Z}^{x,D}_sdW_s.
\end{align}
    \end{lem}
    \begin{proof} 
    %We start by proving ii., which is stated under the~condition that $\xi_{\tau}>C$ $\mathbb{P}$-a.s..%\mycomment{I would rather keep the notation $V^{\theta,D}$ as before.} 
    Fix $(x,D)\in \mathbb{R}_+\times {\cal A}^x$. Set $(\tau,\xi_{\tau}):=(\tau^{x,D},\xi_{\tau^{x,D}})$,
    \[
    \mbox{$(V,Z):=(V^{x,D},Z^{x,D})$, $\quad (\mathbb Y,\mathbb Z):=(\mathbb{Y}^{x,D},\mathbb{Z}^{x,D}),\quad $ and $\quad {\cal V}:={\cal V}^{x,D}$,}
    \]
    for notational simplicity. For any $\theta\in \Theta^{D}$, we let ${\cal V}^{\theta}:={\cal V}^{x,D,\theta}$.
    
    From Lemma \ref{lem:BSDE_exist_g_tau} and It\^{o}'s lemma, $(\mathbb{Y},\mathbb{Z})(=(V^{\frac{1}{1-R}},\mathbb{Y}^RZ/(1-R)))$ solves
		\begin{align*}
			\mathbb{Y}_t %=& \xi_{\tau} + \int_{t\wedge\tau}^{\tau} e^{-\rho s} dD_s - \frac{R}{2}  \int_{t\wedge\tau}^{\tau} \frac{(\mathbb{Z}_s)^2}{\mathbb{Y}_s}ds - \int_{t\wedge\tau}^{\tau} \mathbb{Z}_s dW_s \\
			=& \xi_{\tau}+ \int_{t\wedge\tau}^{\tau} e^{-\rho s} dD_s- \frac{R}{2}  \int_{t\wedge\tau}^{\tau} \frac{(\mathbb{Z}_s)^2}{\mathbb{Y}_s}ds-  \int_{t\wedge\tau}^{\tau}\mathbb{Z}_s\theta_s ds  -  \int_{t\wedge\tau}^{\tau} \mathbb{Z}_s dW_s^{\theta},
		\end{align*}
        under $\mathbb{Q}^\theta$. Recall that $({\cal V}^{\theta},{\iota}^{\theta})$ is a solution of \eqref{eq:thm:rob_rec:1} (see Proposition \ref{pro:Maenhoutdef}).~Then%We obtain
		\begin{align*}
			{\cal V}^{\theta}_t -\mathbb{Y}_t = \int_{t\wedge\tau}^{\tau}\bigg( \frac{R}{2\mathbb{Y}_s} \bigg(\mathbb{Z}_s + \frac{\mathbb{Y}_s}{R}\theta_s\bigg)^2 + \frac{\theta_s^2({\cal V}^{\theta}_s -\mathbb{Y}_s )}{2R}\bigg) ds -  \int_{t\wedge\tau}^{\tau}({\iota}^{\theta}_s - \mathbb{Z}_s) dW_s^{\theta}.
		\end{align*}
		Set $\beta_{t}: = \exp(\frac{1}{{2R}}\int_0^{t}{\theta^2_s} ds)$ for $t\geq 0$. Then using It\^{o}'s lemma, 
		\begin{align}\label{eq:pf:thm:rob_rec_1}
			\beta_{t}(	{\cal V}^{\theta}_t -\mathbb{Y}_t)  = \int_{t\wedge\tau}^{\tau}\beta_{s} \frac{R}{2\mathbb{Y}_s} \Big(\mathbb{Z}_s + \frac{\mathbb{Y}_s}{R}\theta_s\Big)^2ds - (M_{\tau} - M_{t\wedge \tau}), 
		\end{align}
		where $M_t := \int_0^t \beta_{s}({\iota}^{\theta}_s - \mathbb{Z}_s)dW^{\theta}_s$.
        %\mycomment{$\iota^\theta \to \iota$ }?
		Since $(M_s)_{s\geq0}$ is a $\mathbb{Q}^{\theta}$-local martingale, we shall find an increasing sequence of stopping times $\{T_n\}_{n\in\mathbb{N}}$ such that $T_n\ge t$, $\mathbb{Q}^{\theta}(\lim_{n\rightarrow\infty} T_n = \infty) = 1$ and the stopped process $(M_{s\wedge T_n})_{s\geq t}$ is a $\mathbb{Q}^{\theta}$-martingale. Applying $\mathbb{E}_t^{\theta}$ to both sides of \eqref{eq:pf:thm:rob_rec_1}, we have for $t\ge 0$
		\begin{align*}%\label{eq:thm:rob_rec_link:3}
			{\cal V}^{\theta}_t -\mathbb{Y}_t =& \mathbb{E}_t^{\theta}\bigg[ \int_{t\wedge \tau}^{T_n\wedge \tau}\frac{R}{2\mathbb{Y}_s} \Big(\mathbb{Z}_s + \frac{\mathbb{Y}_s}{R}\theta_s\Big)^2 e^{\int_{t\wedge \tau}^{s}\frac{\theta^2_u}{2R} du} d s\bigg]\nonumber\\ 
            &+ \mathbb{E}_t^{\theta} [e^{\int_{t\wedge \tau}^{T_n\wedge \tau}\frac{\theta^2_u}{2R} du} {\cal V}^{\theta}_{T_n\wedge \tau}] -\mathbb{E}_t^{\theta} [e^{\int_{t\wedge \tau}^{T_n\wedge \tau}\frac{\theta^2_u}{2R} du}\mathbb{Y}_{T_n\wedge \tau}]=: \operatorname{I}_t^n+\operatorname{II}_t^n - \operatorname{III}_t^n.
		\end{align*}
		%In order to pass the limit of $n\rightarrow\infty$ under the expectation, we check that each terms of the right-hand side of above equation converges. 
        {The equality \eqref{eq:thm:rob_rec_link:2} is obtained by taking the limit as $n \to \infty$. The convergence of the terms follows from standard arguments based on the MCT and DCT, for which we omit the technical verification.}
    Moreover, \eqref{eq:thm:rob_rec_link:2} implies that ${\cal V}^{\theta}_t \ge \mathbb{Y}_t = (V_t)^{\frac{1}{1-R}}$
    %$V^{\theta}_t \ge \mathbb{Y}_t = (V^D_t)^{\frac{1}{1-R}}$
    and equality holds for $\theta^* = -R\mathbb{Z}/\mathbb{Y}$ if $\theta^*\in \Theta^D$. %\mycomment{pls check.}
    It remains to prove that $\theta^* = -R\mathbb{Z}/\mathbb{Y}\in \Theta^D$. We first show that $(\eta^{\theta^*}_t)_{t\geq 0}$  is a martingale for $t\ge0$.
    %\mycomment{It should be like $(\eta^{\theta^*}_t)_{t\geq 0}$ is a martingale..?}
    By It\^{o}'s lemma, 
    \begin{align}\label{eq:pf_thm3.2_logY}
         \ln\Big(\frac{\mathbb{Y}_{t\wedge\tau}}{\mathbb{Y}_0}\Big) =  \int_{0}^{t\wedge\tau} \frac{R-1}{2} \Big(\frac{\mathbb{Z}_s}{\mathbb{Y}_s}\Big)^2 ds +  \int_{0}^{t\wedge\tau} \frac{\mathbb{Z}_s}{\mathbb{Y}_s} dW_s  -\int_{0}^{t\wedge\tau} \frac{e^{-\rho s}}{\mathbb{Y}_s} dD_s.
    \end{align}
    Since $(\mathbb{Y}_t)_{t\ge0}$ is continuous and never $0$, $1/\mathbb{Y}$ is locally bounded.
    Moreover, by Definition~\ref{def:BSDEdef}~i., it holds that $\theta^* = 0$ on $\{t\ge\tau\}$, Hence we have for\;$t\ge0$,
    \begin{align*}
        &\int_{0}^{t\wedge \tau} \theta^*_s dW_s  -\int_0^{t\wedge \tau}\frac{(\theta^*_s)^2}{2}ds\\
        &\quad = -R\ln\Big(\frac{\mathbb{Y}_{t\wedge\tau}}{\mathbb{Y}_0}\Big) - R\int_{0}^{t\wedge\tau}\bigg( \frac{e^{-\rho s}}{\mathbb{Y}_s} d D_s + \frac{1}{2}\bigg( \frac{\mathbb{Z}_s}{\mathbb{Y}_s}\bigg)^2ds\bigg).
    \end{align*}
    By definition of $\eta^{\theta^*}$ in \eqref{eq:eta}, we have that for $t\ge0$, $\mathbb{P}\text{-a.s.}$,
    \begin{align}\label{eq:pf_thm3.2_eta}
        \eta^{\theta^*}_t =\Big(\frac{\mathbb{Y}_0}{\mathbb{Y}_{t\wedge\tau}}\Big)^R e^{ - \int_{0}^{t\wedge\tau} \frac{R}{\mathbb{Y}_s}e^{-\rho s} d D_s -\frac{R}{2} \int_{0}^{t\wedge\tau}( \frac{\mathbb{Z}_s}{\mathbb{Y}_s})^2ds}\le \mathbb{Y}_0^R C^{-{R}},
    \end{align}
    where we have used that $\xi_{\tau}> C$ $\mathbb{P}$-a.s.. Hence ${\cal V}_t\ge C$ $\mathbb{P}$-a.s., for $t\geq 0$ (see Lemma \ref{lem:BSDE_exist_g_tau}). We know that $(\eta^{\theta^*}_t)_{t\geq 0}$ is a $\mathbb{P}$-local martingale. The uniform bound \eqref{eq:pf_thm3.2_eta} further implies that it is a $\mathbb{P}$-martingale.
    
    Next we verify $\theta^*\in \Theta^D$ (see \eqref{eq:ThetaD}). For the first condition, by \eqref{eq:pf_thm3.2_logY}%, we have
    \begin{align*}
        \mathbb{E}[(\eta^{\theta^*}_{\tau})^2e^{\int_0^{\tau} \frac{(\theta^*_s)^2}{R} ds} \xi_{\tau}^{2}] = \mathbb{E}[e^{ - 2R\int_{0}^{\tau} \frac{1}{\mathbb{Y}_s}e^{-\rho s} dD_s} \mathbb{Y}_{\tau}^{2(1-R)}Y_0^{2R}]<\infty,
        %\le (V^D_0)^{\frac{2R}{1-R}} \mathbb{E}[\xi_{\tau}^2]<\infty,
    \end{align*}
    where we have used the fact that ${\cal F}_0$ is trivial (see Section \ref{sec:2}) in the inequality. 

    For the last condition, it holds that
    \begin{align*}
        &\mathbb{E}\bigg[\bigg(\int_{0}^{\tau} e^{\int_{0}^{t} \frac{(\theta^*_s)^2}{2R}ds -\rho t}\eta_t^{\theta^*} dD_t\bigg)^2\bigg]
        = \mathbb{E}\bigg[\bigg(\int_{0}^{\tau} e^{ R(\ln(\frac{\mathbb{Y}_0}{\mathbb{Y}_{t}})  - \int_{0}^{t} \frac{1}{\mathbb{Y}_s}e^{-\rho s} dD_s) -\rho t} dD_t\bigg)^2\bigg]\\
        &\quad \le\mathbb{Y}_0^{2R} \mathbb{E}\bigg[\bigg( \sup_{t\ge 0} \mathbb{Y}_{t\wedge\tau}^{1-R} \int_{0}^{\tau} e^{ - R \int_{0}^{t} \frac{1}{\mathbb{Y}_s}e^{-\rho s} dD_s}\frac{1}{\mathbb{Y}_t} e^{-\rho t}  dD_t\bigg)^2\bigg]\\
        &\quad\le \frac{\mathbb{Y}_0^{2R}}{R} \mathbb{E}\bigg[\sup_{t\ge 0} V_{t\wedge\tau}^2\bigg] <\infty,
       % \le& \mathbb{Y}_0^{2R} \mathbb{E}\left[\sup_{t\ge 0} (V^D_{t\wedge\tau})^2 \frac{1}{R}(1 - e^{-R\int_{0}^{\tau} \frac{1}{\mathbb{Y}_s}e^{-\rho s} dD_s})^2\right] <\infty.
    \end{align*}
    as claimed. Hence, this completes the proof.  
    \end{proof}

    We now provide the proof of Theorem \ref{thm:rob_rec}.

    \begin{proof}
        Consider a non-increasing sequence $\xi_{\tau^{x,D}}^{(n)}:= \frac{1}{n}\vee\xi_{\tau^{x,D}}$, $n\in \mathbb{N}$. Then using the same arguments as in the proof of Proposition \ref{pro:recurs_BSDE}, we are able to obtain the corresponding limit from a sequence constructed via Lemma \ref{lem:rob_rec_strict} for each $\xi_{\tau^{x,D}}^{(n)}$ (which is greater than equal to $\frac{1}{n}$). This completes the proof.
     
    \end{proof}

    %%%%%%%%%
    \subsection{Analysis of the free-boundary problem \eqref{eq:fbdy}: Proof of Proposition~\ref{pro:shooting} }\label{proof:pro:shooting}
    We start with some preliminary results on the free-boundary problem \eqref{eq:fbdy}.  Recalling $\psi^+$ defined in Assumption~\ref{as:bdy_cond}, we claim that for every $b>0$ and $\gamma\in \mathbb{R}$, there is $v_{b,\gamma}\in C^2((0,b)\cup(b,\infty))\cap C^1(\mathbb{R}_+)$ ($v_{b,0}\in C^2({\mathbb{R}_+})$ when $\gamma=0$) solving the following nonlinear ODE: for $x\in\mathbb{R}_+$,
\begin{align}
\left\{
\begin{aligned}
    &\dfrac{\sigma^2(x)}{2}v''_{b,\gamma}(x)+ \mu(x)v_{b,\gamma}'(x) - R\dfrac{\sigma^2(x)}{2}\dfrac{\big(v_{b,\gamma}'(x)\big)^2}{v_{b,\gamma}(x)} = (\rho+\gamma) v_{b,\gamma}(x) ,\\
    &v_{b,\gamma}'(b)=1, \quad v_{b,\gamma}(b)= \psi^+(b). \label{eq:HJB_ODE1}
\end{aligned}
\right.
\end{align}
Indeed, for any given $b>0$ and $\gamma\in \mathbb{R}$, one-to-one correspondence 
\begin{align}\label{eq:1to1_cor}
    h_{b,\gamma}:=(v_{b,\gamma})^{1-R}
\end{align}
enables to consider the following {linear} ODE %(which is equivalent to \eqref{eq:HJB_ODE1}) 
given by
\begin{align}\label{eq:ODE_h}
\left\{
\begin{aligned}    
    &\dfrac{1}{2}\sigma^2(x)h''_{b,\gamma}(x)+ \mu(x)h_{b,\gamma}'(x) = (1-R)(\rho+\gamma) h_{b,\gamma}(x)\quad \mbox{on}\;\;[a_1,a_2],\\
    &h_{b,\gamma}'(b)=(1-R)\big(\psi^+(b)\big)^{-R}, \quad h_{b,\gamma}(b)= \big(\psi^+(b)\big)^{1-R},
\end{aligned}
\right.
\end{align}
with an interval $[a_1,a_2]\subset \mathbb{R}$ with $a_1<b<a_2$. 

As \eqref{eq:ODE_h} admits a unique solution $h_{b,\gamma}\in C^2([a_1,b)\cup(b,a_2])\cap C^1([a_1,a_2])$ ($h_{b,0}\in C^2([a_1,a_2])$ when $\gamma=0$) for any $[a_1,a_2]\subset \mathbb{R}$ with $a_1<b<a_2$ (The Picard-Lindel\"{o}f's existence/uniqueness theorem for linear ODEs; see e.g., Nagle et al.~\cite[Chapter 13, Section 2, Theorem 3]{nagle1996fundamentals}), our claim holds true. 

%In a similar manner, we have for every $b>0$ that when $\gamma=0$, 

Define $g_{b,\gamma}:\mathbb{R}_+\ni x\rightarrow g_{b,\gamma}(x)\in\mathbb{R}$ for every $b>0$ and $\gamma\in \mathbb{R}$ by %for $x\in{{\mathbb{R}_+}}$,
\begin{align*}%\label{eq:equiv_deriv}
    g_{b,\gamma}(x) := {v'_{b,\gamma}(x)}/{v_{b,\gamma}(x)}.    
\end{align*}
Since $v_{b,\gamma}\in C^2((0,b)\cup(b,\infty))\cap C^1(\mathbb{R}_+)$ (when $\gamma=0$, $v_{b,0}\in C^2(\mathbb{R}_+)$) is the solution of \eqref{eq:HJB_ODE1} and satisfies 
\begin{align}\label{eq:equiv_deriv_2}
    v_{b,\gamma}(x) = \psi^+(b)\exp\left(\int_{b}^x g_{b,\gamma}(u) du \right),
\end{align}
it follows that $g_{b,\gamma}$ is of $C^1((0,b)\cup(b,\infty))\cap C(\mathbb{R}_+)$ (when $\gamma=0$, $g_{b,0}\in C^1({\mathbb{R}_+})$) solves the following ODE: for $x\in \mathbb{R}_+$
\begin{align}
\left\{
\begin{aligned}   
    &\dfrac{\sigma^2(x)}{2}g_{b,\gamma}'(x)+\mu(x)g_{b,\gamma}(x) + \dfrac{(1-R)}{2}\sigma^2(x)g_{b,\gamma}^2(x) = \rho+\gamma,\\
    &g_{b,\gamma}(b) = {1}/{\psi^+(b)}.\label{eq:ODE_g}
\end{aligned}
\right.
\end{align}
For notational simplicity, set 
$v_b:=v_{b,0}$, $h_{b}:=h_{b,0}$, and $ g_{b}:=g_{b,0}$ for all $b>0$. %We state the following lemma

%Targeting at proving Proposition~\ref{pro:shooting}, 
Our main task is to show the following proposition. The corresponding proof builds on a sequence of preliminary lemmas given in the next subsection.
\begin{pro}\label{pro:shooting_sub}
    Suppose that Assumption \ref{as:bdy_cond} holds. For every $b>0$, let $v_{b}\in C^2({\mathbb{R}_+})$ be the solution of \eqref{eq:HJB_ODE1} (when $\gamma=0$). Then there exists a unique threshold $b^{\ast}\in (\underline{b},\hat{b})$ such that
    \(
    v_{b^{*}}(0) =\xi_0\) and \(v'_{b^{*}}(x)\ge 1\)  {on} \((0,b^{*}],
    \)
    with $\underline{b},\hat{b}$ appearing in Assumption \ref{as:bdy_cond}. We adopt the convention that $ v_{b^{*}}$ is extended to 0 by continuity and the value of $ v_{b^{*}}$ at 0 is its limit from the right.
    %Setting up $f'_{\beta^{\epsilon}}(x)=1$ on $(\beta^{\epsilon},\infty)$ and the second line of \eqref{eq:HJBVI_b_1} holds. So the main challenge is choosing a point $b=\beta^{\epsilon}$ such that the bound $f'_{b}(x)\ge 1$ holds. 
\end{pro}

\subsubsection{Preliminary lemmas for Proposition \ref{pro:shooting_sub}}
In what follows, we often make use of the following elementary properties that are based on Cohen et al.~\cite[Lemma 1]{cohen2022optimal}.% that will be used several times below. 
\begin{lem}\label{lem:ele} Let $(a_1,a_2)\subseteq \mathbb{R}$ and $a_3\in (a_1,a_2)$, and let $f:(a_1,a_2)\ni x \rightarrow f(x)\in \mathbb{R}$ be a function of class $C^1((a_1,a_2))$. Then the following hold:
\begin{enumerate}[label=\roman*.]
    \item If $f(a_3)>c$ (resp.~$<c$) with some $c\in\mathbb{R}$ and there exist $y_1: = \sup\{y\in(a_1,a_3):f(y)=c\}$ and $y_2 : =\inf\{y\in(a_3,a_2):f(y)=c \}$, then 
    \[
    f'(y_1)\ge0\;\;\mbox{(resp.~$\le 0$)},\quad f'(y_2)\le 0\;\;\mbox{(resp.~$\ge0$)}.
    \]
    \item If $f'(a_3)>0$ (resp., $<0$) and there exist $y_3:=\sup\{y\in(a_1,a_3):f(y)=f(a_3)\}$ and $y_4:=\inf\{y\in(a_3,a_2):f(y) = f(a_3)\}$, then 
    \[
    f'(y_3)\le 0\;\;\mbox{(resp.~$\ge0$)},\quad f'(y_4)\le 0\;\; \mbox{(resp.~$\ge0$)}.
    \]
\end{enumerate}
%Let $f$ be a $C^1$ function defined on $(a,b)$. Suppose at $x\in(a,b)$, $f(x)>c$ (resp., $<c$) for some $c\in\mathbb{R}$. Then upon existence of $y_1: = \sup\{y\in(a,x):f(y)=c\}$ and $y_2 : =\inf\{y\in(x,b):f(y)=c \}$, $f'(y_1)\ge0$ (resp., $\le 0$), $f'(y_2)\le 0$ (resp., $\ge0$). \\
%\indent As a corollary, let $f$ be a $C^1$ function defined on $(a,b)$. Fix $x\in(a,b)$, and let $y_1:=\sup\{y\in(a,x):f(y)=f(x)\}$ and $y_2=\inf\{y\in(x,b):f(y) = f(x)\}$ if they exist. If $f'(x)>0$ (resp., $<0$) then, $f'(y_1)\le 0$ (resp., $\ge0$), and $f'(y_2)\le 0$ (resp., $\ge0$).
\end{lem}
\noindent We start with a perturbation result which helps to get estimations for $g_b$ via the estimates of $g_{b,\gamma}$.
\begin{lem}\label{lem:pertub_cts} Suppose that Assumption \ref{as:bdy_cond} holds. For each $b>0$, let $v_{b,\gamma}$ be the solution of \eqref{eq:HJB_ODE1} and $g_{b,\gamma}$ be the solution of \eqref{eq:ODE_g}. The following hold: 
%let $v_{b,\gamma}\in C^2((0,b)\cup(b,\infty))\cap C^1(\mathbb{R}_+)$ ($v_{b}=v_{b,0}\in C^2(\mathbb{R}_+)$ when $\gamma=0$) be the solution of \eqref{eq:HJB_ODE1} and $g_{b,\gamma}\in C^1((0,b)\cup(b,\infty))\cap C(\mathbb{R}_+)$ ($g_{b}=g_{b,0}\in C^1(\mathbb{R}_+)$ when $\gamma=0$) be the solution of \eqref{eq:ODE_g}.  Then, the followings hold for every $b>0$:
\begin{enumerate}[label=\roman*.]
    \item Let $b\in (\underline{b},\overline{b})$ and $\gamma<0$ (with $\underline{b},\overline{b}$ appearing in Assumption \ref{as:bdy_cond}). Define $$y_{b,\gamma}:=\sup\{x\in(0,b):g_{b,\gamma}(x) \le 1/\psi^+(x)\},$$ with the usual convention $\sup\{\varnothing\} : =0$.
    % \begin{align*}\qquad
    % y_{b,\gamma}:=\left\{
    % \begin{aligned}
    %     &0\quad \mbox{if $~g_{b,\gamma}(x)>{1}/{\psi^+(x)}$ for every $x\in (0,b)$};\\
    %     &\sup\{x\in(0,b):g_{b,\gamma}(x) \le 1/\psi^+(x)\}\quad \mbox{else}.
    % \end{aligned}
    % \right.
    % \end{align*}
    Then $y_{b,\gamma}\in[0,\underline{b}]$, and for every $x\in(y_{b,\gamma},b)$, $g_{b,\gamma}(x)> 1/\psi^+(x)$.

    \item  Let $b_1,b_2\in (\underline{b},\overline{b})$ with $b_1\leq b_2$ and $\gamma_1,\gamma_2 \in (-\infty,0]$ with $\gamma_1>\gamma_2$. Then the following holds: for every $x\in(0,b_1)$, 
    \(
    g_{b_1,\gamma_1}(x)<g_{b_2,\gamma_2}(x).
    \)
    \item As $\gamma\rightarrow0$, we have 
    $
     \sup_{x\in[0,b]} |g_{b,\gamma}(x) - g_b(x)| = O(\gamma),
    $
    where $O(\cdot)$ denotes the Landau symbol. Moreover, for every $x\in[0,b]$, %we have
    \begin{align*}%\label{eq:lem_pertub_cts_1}
        \lim_{\gamma\rightarrow0}|v_{b,\gamma}(x) - v_{b}(x)|=0,\qquad \lim_{\gamma\rightarrow0}|v'_{b,\gamma}(x) - v'_{b}(x)|=0.
    \end{align*} 
\end{enumerate}
\end{lem}
\begin{proof} {\it of i.} We first note that %$g_{b,\gamma}'(b)<-{1}/{(\psi^+(b))^2}$. Indeed, 
   by \eqref{eq:ODE_g}, the following inequality holds:
    \begin{align}\label{eq:est_g1}
    g_{b,\gamma}'(b)= \frac{2Q(\psi^+(b))}{(\sigma(b)\psi^+(b))^2}+\frac{2\gamma}{(\sigma(b))^2}-\frac{1}{(\psi^+(b))^2}<-\frac{1}{(\psi^+(b))^2},
    \end{align}
    where we use that $Q(\psi^+(b))=0$ (defined in Assumption\;\ref{as:bdy_cond}\;ii.) and $\gamma<0$.

    Define $\phi_{b,\gamma}:\mathcal{X}\rightarrow \mathbb{R}$ as $\phi_{b,\gamma}(x) := g_{b,\gamma}(x) - 1/\psi^+(x)$. Since $g_{b,\gamma}(b)={1}/{\psi^+(b)}$ and $(\psi^+)'(b)\le 1$ (because $\psi^+(x)-x$ is decreasing on $[\underline{b},\overline{b})$ and $b$ is in $(\underline{b},\overline{b})$; see Assumption\;\ref{as:bdy_cond}\;ii.), the inequality in \eqref{eq:est_g1} gives
    \begin{align}\label{eq:est_g2}
    \phi_{b,\gamma}(b)=0,\quad \phi_{b,\gamma}'(b)=g'_{b,\gamma}(b) + \frac{(\psi^+)'(b)}{(\psi^+(b))^2} < \frac{(\psi^+)'(b)-1}{(\psi^+(b))^2} \le0 .
    \end{align}

    We claim that $y_{b,\gamma}\leq \underline{b}$. Since the case with $y_{b,\gamma}=0$ is trivial, we can and do consider only the case with $y_{b,\gamma}=\sup\{x\in(0,b):\phi_{b,\gamma}(x)=0\}$. Arguing by contradiction, assume that $y_{b,\gamma}>\underline{b}$. Then by the continuity of $\phi_{b,\gamma}$ (due to that of $g_{b,\gamma}$ and $\psi^+$), we have that $\phi_{b,\gamma}(y_{b,\gamma})= g_{b,\gamma}(y_{b,\gamma}) - {1}/{\psi^+(y_{b,\gamma})}=0.$ Furthermore, using the same arguments devoted for the second property given in \eqref{eq:est_g2} (along with $y_{b,\gamma}>\underline{b}$), we have 
    $$(\phi_{b,\gamma})'(y_{b,\gamma})=g'_{b,\gamma}(y_{b,\gamma}) + \frac{(\psi^+)'(y_{b,\gamma})}{(\psi^+(y_{b,\gamma}))^2} <0,$$ which contradicts Lemma \ref{lem:ele}\;ii.. Hence, the claim holds true.
    
    Furthermore, from the definition of $y_{b,\gamma}$ and the fact that $(\phi_{b,\gamma})'(b)=(g_{b,\gamma}- 1/\psi^+)'(b)<0$ (see~\eqref{eq:est_g2}), the last statement of i. is true. 
    
    \noindent {\it Proof of ii.} To that end, let $y_{b_2,\gamma_2}$ be defined as Lemma~\ref{lem:pertub_cts}\;i. (by assigning $b=b_2$ and $\gamma=\gamma_2$). Define $\Delta g: [0,b]\ni x\rightarrow \Delta g(x)\in \mathbb{R} $ by $\Delta g(x):= g_{b_1,\gamma_1}(x)-g_{b_2,\gamma_2}(x)$ for $x\in [0,b_1]$. Since $g_{b_i,\gamma_i}$, $i=1,2$, is the solution of \eqref{eq:ODE_g} (when $b=b_i$, $\gamma=\gamma_i$), the following holds
        \begin{align}
                &\dfrac{1}{2}\sigma^2(x)(\Delta g)'(x) + \Delta g(x) \Big( \mu(x)+ \frac{(1-R)}{2}\sigma^2(x)\left({\Delta g(x)} + 2g_{b_2,\gamma_2}(x)\right)\Big) \nonumber\\
                &= \gamma_1-\gamma_2,\quad\mbox{for}\;\;x\in [0,b_1]\label{eq:est_g4_2}
        \end{align}
        and $\Delta g(b_1) = 1/{\psi^+(b_1)}-g_{b_2,\gamma_2}(b_1)<0$ which follows from Lemma \ref{lem:pertub_cts}\;i. with fact that $b_1\in (y_{b_2,\gamma_2},b_2)$ and $g_{b_1,\gamma_1}(b_1)=1/{\psi^+(b_1)}$. 

    Here, we claim that $\Delta g(x)<0$ on $(0,b_1)$. Arguing by contradiction, assume it does not hold; then there exists $y_{3}:=\sup\{x\in(0,b_1): \Delta g(x) = 0\}$. Since $\Delta g(y_3)=0$ (by its continuity) and $y_3\in[0,b_1]$, assigning $x=y_3$ into \eqref{eq:est_g4_2} ensures that $\Delta g'(y_3)=2(\gamma_1-\gamma_2)/\sigma^2(y_3)>0$. This together with the fact that $\Delta g(b_1)<0$ contradicts Lemma \ref{lem:ele}\;i.

    \noindent {\it Proof of iii.} This proof is similar to the one given in Cohen et al.~\cite[Lemma~2]{cohen2022optimal}, so we omit the details here.
 
\end{proof}

\begin{lem}\label{lem:lowerbound_g} Suppose that Assumption \ref{as:bdy_cond} holds. For every $b>0$, let $v_b\in C^2(\mathbb{R}_+)$ be the solution of \eqref{eq:HJB_ODE1} (when $\gamma=0$) and $g_{b}\in C^1(\mathbb{R}_+)$ be the solution of \eqref{eq:ODE_g} (when $\gamma=0$). Then the following hold: 
\begin{enumerate}[label=\roman*.]
    \item Let $b\in(\underline{b},\overline{b})$.
    Then there exists $y_b^*\in [0,\underline{b}]$ satisfying that for all $x\in [y_b^*,b]$, $g_{b}(x)\ge 1/\psi^+(x),$ with $\underline{b}$ appearing in Assumption \ref{as:bdy_cond}.
    \item Let $b_1,b_2\in (\underline{b},\overline{b})$ with $b_1<b_2$. Then the following holds: for all $x\in (0,b_1]$, $g_{b_1}(x)\le g_{b_2}(x).$
    \item Let $b\in (\underline{b},\overline{b})$ and $\gamma\le 0$. As $\varepsilon \rightarrow 0$,   $\sup_{x\in[0,b]}|g_{b+\varepsilon,\gamma}(x) - g_{b,\gamma}(x)| =O(\varepsilon)$. Furthermore,  for every $x\in[0,b]$,
    \begin{align}\label{eq:lem_cts_b_g_2}
        \lim_{\varepsilon \rightarrow 0} |v_{b+\varepsilon,\gamma}(x) - v_{b}(x)| =0,\qquad \lim_{\varepsilon \rightarrow 0} |v_{b+\varepsilon, \gamma}'(x) - v_{b,\gamma}'(x)| =0.
    \end{align}
\end{enumerate}
\end{lem}
\begin{proof}
{\it of i.} Let $y_{b,\gamma}\in[0,\underline{b}]$ for every $\gamma<0$ be defined as in Lemma \ref{lem:pertub_cts}\;i. Then $y_{b,\gamma}$ increases in $\gamma<0$. Indeed, Lemma \ref{lem:pertub_cts}\;ii. ensures that for any $\gamma_1,\gamma_2\in(-\infty,0)$ such that $\gamma_1>\gamma_2$, the following holds that $
    g_{b,\gamma_1}(x) <g_{b,\gamma_2}(x)$ for every $x\in(0,b)$.
    Combining this with the definition of $y_{b,\gamma_i}$, $i=1,2$, (see Lemma \ref{lem:pertub_cts}\;i.), we have $y_{b,\gamma_1}>y_{b,\gamma_2}$. This ensures that the (increasing) monotonicity of $y_{b,\gamma}$ with resepect to $\gamma<0$.

    Consider an increasing sequence $(\gamma_n)_{n\in \mathbb{N}}\subseteq (-\infty,0)$ such that $\gamma_n\uparrow 0$ as {$n\rightarrow\infty$}. Then the monotonicity of $(y_{b,\gamma_n})_{n\in\mathbb{N}}$ together with uniformly boundedness within $[0,\underline{b}]$ ensures that there exists
    $y_{b}^*:=\lim_{n\rightarrow\infty}y_{b,\gamma_n}$ satisfying $y_{b}^*\in[0,\underline{b}]$ and $y_{b}^*\geq y_{b,\gamma_n}$ for every $n\in \mathbb{N}$. 
    Combined with Lemma~\ref{lem:pertub_cts}\;i., this ensures that for every $n\in\mathbb{N}$ and every $x\in[y_{b}^*,b)$, 
    $
    g_{b,\gamma_n}(x)\geq  1/\psi^+(x).$ 
    %Therefore 
    Using Lemma \ref{lem:pertub_cts}\;iii. and letting $n\rightarrow \infty$, we complete the proof. 

    %\vspace{0.5em}
    \noindent {\it Proof of ii.} Lemma \ref{lem:pertub_cts}\;ii. ensures that for every $\gamma<0$ and every $x\in(0,b_1]$, $g_{b_1,0}(x)=g_{b_1}(x)< g_{b_2,\gamma}(x).$ Using Lemma \ref{lem:pertub_cts}\;iii. and letting $\gamma\uparrow 0$, we complete the proof.

    %\vspace{0.5em}
    \noindent {\it Proof of iii.} We establish the proof for the case where $\gamma = 0$, as the analysis can be readily extended to accommodate the case where $\gamma\neq 0$.
    
    We start with proving the first convergence therein. Let~$\varepsilon\in [0,1\wedge(\overline{b}-b))$. Since $g_{b+\varepsilon}\in C^2(\mathbb{R}_+)$ (resp. $g_{b}\in C^2(\mathbb{R}_+)$) is the solution of \eqref{eq:ODE_g} (when $\gamma=0$ with $b+\varepsilon$ (resp. $b$)), the following holds: for every $x\in\mathbb{R}_+$,
    \begin{align}
      &g_{b+\varepsilon}(x) - g_{b}(x) -\big(g_{b+\varepsilon}(b) - g_{b}(b)\big)= - \int_{x}^b \big(g'_{b+\varepsilon}(y) - g'_{b}(y) \big) dy \nonumber\\
     &=  \int_{x}^b \Big( (1-R)\sigma^2(x)\big(g_{b+\varepsilon}(y) +g_b(y)\big) + \frac{2\mu(x)}{\sigma^2(x)}  \Big)\big(g_{b+\varepsilon}(y) - g_b(y)\big) dy. \label{eq:est_g5}
    \end{align}

    Lemma~\ref{lem:lowerbound_g}\;ii. ensures that $g_b(x)\le g_{b+\varepsilon}(x)\le g_{b+1}(x)\le C_0$ for every $x\in[0,b]$, where $C_0>0$ is a constant (that depends on $1\wedge (\overline{b}-b)$ but not on $\varepsilon$). %, otherwise (i.e., $\varepsilon \in (\underline{b}-b,0]$), $g_{b+\varepsilon}(x)\le g_{b}(x)\le C_0$ for every $x\in[0,b+\varepsilon]$, with some $C_0$.
    Furthermore, $\sigma$ and $\mu$ are uniformly bounded on $[0,b]$ (see Assumption~\ref{as:bdy_cond}). Hence, combined with \eqref{eq:est_g5} these properties ensure that there is a constant $C_1>0$ (that depends on $C_0$ but not on $\varepsilon$) such that for every $x\in[0,b)$, $
    0\leq g_{b+\varepsilon}(x) - g_{b}(x)  \leq g_{b+\varepsilon}(b) - g_{b}(b) + C_1  \int_b^x \big(g_{b+\varepsilon}(y) - g_b(y)\big) dy.$ An application of Gr\"onwall's inequality gives the existence of a constant $C_2>0$ satisfying 
    \begin{align}\label{eq:est_g6}
    0\leq \sup_{x\in[0,b]}(g_{b+\varepsilon}(x) - g_{b}(x)) \leq C_2 (g_{b+\varepsilon}(b) - g_{b}(b)).
    \end{align}

    Furthermore, we claim that there is a constant $C_3$ (that depends on $1\wedge (\overline{b}-b)$ but not on~$\varepsilon$) such that $
     0\leq g_{b+\varepsilon}(b) - g_{b}(b) \leq C_3 \varepsilon.$ Indeed, 
       \begin{align*}
    \left|g_{b+\varepsilon}(b) - g_{b}(b) \right| \le |g_{b+\varepsilon}(b+\varepsilon) -g_{b}(b) | + |g_{b+\varepsilon}(b+\varepsilon) - g_{b+\varepsilon}(b)|: = \operatorname{I} + \operatorname{II}.
    \end{align*}
    $\operatorname{I}$ is bounded by  $\varepsilon  \sup_{x\in[b,b+\varepsilon]}|{(\psi^+)'(x)}/{(\psi^+)^2(x)}|$, which is of $O(\varepsilon)$ by the fact that $g_{b+\varepsilon}(b+\varepsilon)=1/\psi^+(b+\varepsilon)$, $g_{b}(b)=1/\psi^+(b)$, and $\psi^+\in C^1([0,\overline{b}])$.
    $\operatorname{II}$ is of $O(\varepsilon)$ with similar arguments devoted for \eqref{eq:est_g5}.
    
    Combined with \eqref{eq:est_g6}, % (noting that $\varepsilon\in[0,1\wedge (\overline{b}-b))$), 
    this ensures that
    \begin{align*}%\label{eq:est_g7}
    \lim_{\varepsilon\downarrow 0}\sup_{x\in[0,b]}\Big(g_{b+\varepsilon}(x) - g_{b}(x)\Big) =0.
    \end{align*}
\eqref{eq:lem_cts_b_g_2} holds because $v_{b}(x) =  \psi^+(b)\exp(\int_{b}^x g_{b}(u) du )$ and $v_b'(x) = v_b(x)g_b(x)$. 

The proof in the case $\varepsilon\uparrow 0$ is similar therefore we omit it.  \end{proof}

\begin{lem}\label{lem:b_les_x}
    Suppose that Assumption \ref{as:bdy_cond} holds. Let $b\in (0,\underline{b}]$ and $v_b\in C^2(\mathbb{R}_+)$ be the solution of \eqref{eq:HJB_ODE1} (when $\gamma=0$), with $\underline{b}$ appearing in Assumption \ref{as:bdy_cond}. Then for every $x\in(0,b]$, $v_{b}'(x) \le 1.$
\end{lem}
\begin{proof}
    Let $\gamma>0$ and $v_{b,\gamma}$ be the solution of \eqref{eq:HJB_ODE1}. Since $v_{b,\gamma}'(b)=1$ and $v_{b,\gamma}(b)= \psi^+(b)$, \eqref{eq:HJB_ODE1} ensures that 
    $v''_{b,\gamma}(b) = 2\gamma\psi^+(b)/\sigma^2(b)>0$.
     
    We claim that for every $x\in(0,b)$, $v'_{b,\gamma}(x)<1$. Arguing by contradiction, we assume that it does not hold. Then together with $v_{b,\gamma}'(b)=1$, the following supremum is attained: $x_1:=\sup\{x\in(0,b): v'_{b,\gamma}(x) = 1 \}$. By definition of this, it is clear that for every $x\in[x_1,b]$, $v'_{b,\gamma}(x)\le1,$ which ensures that $v_{b,\gamma}(x) - v_{b,\gamma}(b)\geq x-b$ for every $x\in[x_1,b]$.

    Furthermore, from Assumption~\ref{as:bdy_cond}\;ii. (with $b<\underline{b}$) and $v_{b,\gamma}=\psi^+(b)$, it follows that for every $x\in[x_1,b]$, $
    v_{b,\gamma}(x_1) \ge \psi^+(b)-(b-x_1) \ge \psi^+(x_1)$. This implies that $Q(v_{b,\gamma}(x_1)) \ge Q(\psi^+(x_1))=0$ (see Assumption \ref{as:bdy_cond}\;ii.). From the nonlinear ODE \eqref{eq:HJB_ODE1}, it follows that
    \begin{align*}
        \frac{1}{2}\sigma^2(x_1) v''_{b,\gamma}(x_1)=Q\big(v_{b,\gamma}(x_1)\big)\frac{1}{v_{b,\gamma}(x_1)}+\gamma v_{b,\gamma}(x_1)
        \ge \gamma \psi^+(x_1)>0.
    \end{align*}
    Combining this with the fact that $v''_{b,\gamma}(b)>0$ contradicts Lemma~\ref{lem:ele}\;ii. Hence we have shown that the claim holds.
    
    Using Lemma~\ref{lem:pertub_cts}\;iii. and letting $\gamma\downarrow0$, we complete the proof.  % 
\end{proof}

\begin{lem}\label{lem:comparision}
    Suppose that Assumption \ref{as:bdy_cond} holds.  For every $b>0$ and $\gamma\in\mathbb{R}$, let $v_{b,\gamma}\in C^2((0,b)\cup(b,\infty))\cap C^1(\mathbb{R}_+)$ ($v_{b}=v_{b,0}\in C^2(\mathbb{R}_+)$ when $\gamma=0$) be the solution of \eqref{eq:HJB_ODE1}. %and $g_{b,\gamma}\in C^1((l,b)\cup(b,\infty))\cap C({{\mathbb{R}_+}})$ ($g_{b}=g_{b,0}\in C^1({{\mathbb{R}_+}})$ when $\gamma=0$) be the solution of \eqref{eq:ODE_g}.  
    Then the following hold: 
    \begin{enumerate}[label=\roman*.]
        \item  Let $b\le \hat{b}$ (with $\underline{b},\hat{b}$ appearing in Assumption \ref{as:bdy_cond}) and $\gamma\le 0$. If $v_{b,\gamma}(0) \ge \xi_0$, then for every $x\in[0,b]$, $v_{b,\gamma}(x)\ge x + \xi_0.$
        \item   Let $\gamma < 0$. Then there exists a unique threshold $b^*(\gamma)$ defined by
    \[
    b^*(\gamma) := \inf\{b \in (\underline{b}, \hat{b}) : v_{b,\gamma}(0) = \xi_0\}.
    \]
    Moreover, $b^*(\gamma)$ is monotonically increasing in $\gamma < 0$. Consequently, %the limit
    \begin{align}\label{def:b*gamma}
        b^* := \lim_{\gamma \uparrow 0} b^*(\gamma)
    \end{align}
    is well-defined and unique, with $b^* \in (\underline{b}, \hat{b}]$ and $v_{b^*}(0) = \xi_0$. \color{black}
    %\comment{Is that correct?? Cohen's paper seems not explicitly speak out...}
    % \[
    % b^*\in (\underline{b},\hat{b}], \qquad v_{b^*}(0) = \xi_0.
    % \] 
    \end{enumerate}

\end{lem}
\begin{proof}{\it of i.} Arguing by contradiction, we assume that it does not hold. Recalling the one-to-one correspondence $h_{b,\gamma}=(v_{b,\gamma})^{1-R}$ (see \eqref{eq:1to1_cor}), 
    %then using the transformation $h(x) = (f_{b,\gamma})^{1-\epsilon}(x)$, 
    we obtain~that 
    \[
    \max_{x\in[0,b]}\{ (x+\xi_0)^{1-R} - h_{b,\gamma}(x) \} > 0.
    \] 
    We note that from Assumption\;\ref{as:bdy_cond}\;iii. (with $b\leq \hat{b}$) and the condition $v_{b,\gamma}(0) \ge \xi_0$, it follows that $h_{b,\gamma}(x)\ge(x+\xi_0)^{1-R}$ when $x=0$ and $x=b$. Hence this implies that $(x+\xi_0)^{1-\epsilon} - h_{b,\gamma}(x)$ takes its maximum at $x_0\in(0,b)$. We have
    %Furthermore, %since $h_{b,\gamma}\in C^2((0,b)\cup (b,\infty))$ and the maximality of $x_0$, the following hold:
    \begin{align}\label{pf:lem_comp_1}
    -\frac{R(1-R)}{(x_0+\xi_0)^{1+R}} \le h_{b,\gamma}''(x_0),\quad (1-R)(x_0+\xi_0)^{-R} =h_{b,\gamma}'(x_0). 
    \end{align}
    From the linear ODE given in \eqref{eq:ODE_h}, we get to a contradiction
    \begin{align*}
    0 &= \dfrac{1}{2}\sigma^2(x_0)h_{b,\gamma}''(x_0)+ \mu(x_0)h_{b,\gamma}'(x_0) - (1-R)(\rho+\gamma) h_{b,\gamma}(x_0)\\
    %\ge & -\frac{R(1-R)}{2}\sigma^2(x_0)(x_0+\xi_0)^{-1-R} + \mu(x_0)(1-R)(x_0+\xi)^{-R} - (1-R)(\rho+\gamma)h_{b,\gamma}(x_0)\\
    &> \frac{(1-R)}{(x_0+\xi_0)^{1+R}}\Big(-\frac{R}{2}\sigma^2(x_0) + \mu(x_0)(x_0+\xi_0) -(\rho+\gamma)(x_0+\xi_0)^2\Big)\ge0,
    \end{align*}
    where the first inequality follows from \eqref{pf:lem_comp_1} and that $h_{b,\gamma}(x_0) < (x_0+\xi_0)^{1-R}$, and the second inequality follows from Assumption\;\ref{as:bdy_cond}\;ii. and $\gamma\le 0$. This completes the proof of i.
    
    %Hence, a contradiction is deduced. 
    %\vspace{0.5em}
    \noindent {\it Proof of ii.} This statement is proved in three steps.
    
    {\bf Step 1.}~We claim that $v_{\underline{b},\gamma}(0)> \xi_0$ and $v_{\hat{b},\gamma}(0)<\xi_0$.
    
    Lemma \ref{lem:b_les_x} (when $b=\underline{b}$) ensures that $v_{\underline{b},\gamma}(0)\geq v_{\underline{b},\gamma}(\underline{b})-(\underline{b}-x)$ for every $x\in[0,\underline{b}]$. In particular, $v_{\underline{b},\gamma}(0)>v_{\underline{b},\gamma}(\underline{b})-\underline{b}$ because it cannot be the case that $v_{\underline{b},\gamma}'(x)\equiv 0$ for $x\in[0,\underline{b}]$.   Furthermore, from the fact that $v_{\underline{b},\gamma}(\underline{b})=\psi^+(\underline{b})$ and Assumption \ref{as:bdy_cond}\;ii.\;iii. (with $\underline{b}<\hat{b}$), it follows that $v_{\underline{b},\gamma}(0) > \psi^+(\underline{b})-\underline{b}> \psi^+(\hat{b})-\hat{b} = \xi_0$. 
    
    It remains to show $v_{\hat{b},\gamma}(0)<\xi_0$. Arguing by contradiction, we assume that it does not hold. %Then we have $v_{\hat{b},\gamma}(0) \ge \xi_0$. 
    From Lemma~\ref{lem:Bcomparision}\;i., it follows that $v_{\hat{b},\gamma}(x) \ge x + \xi_0$ for every $x\in[0,\hat{b}]$. Since $v_{\hat{b},\gamma}$ is the solution of \eqref{eq:HJB_ODE1}, assigning $b=\hat{b}$ into \eqref{eq:HJB_ODE1} ensures that $v''_{\hat{b},\gamma}(\hat{b}) = 2\gamma\sigma^2(\hat{b})\psi^+(\hat{b})<0$. 
    
    By using a proof with contradiction, we have that $v_{\hat{b},\gamma}'(x)>1$, for $x\in(\underline{b},\hat{b})$.
     From the claim, the fact that $v_{\hat{b},\gamma}(\hat{b})=\psi^+(\hat{b})$, and Assumption~\ref{as:bdy_cond}\;ii. (with $\hat{b}\in[\underline{b},\overline{b}]$), it follows that $v_{\hat{b},\gamma}(\underline{b}) < \psi^+(\hat{b}) - (\hat{b}-\underline{b}) =  \xi_0+\underline{b}$. This contradicts with Lemma~\ref{lem:Bcomparision}\;i. Therefore, we conclude that $v_{\hat{b},\gamma}(0)< \xi_0$. 
    %By Lemma~\ref{lem:cts_b_g}, the mapping $b\in (x_{\epsilon}, \overline{x}_{\epsilon,\xi}) \mapsto f_{b,\gamma}(0)$ is continuous. 
    %Together with Lemma~\ref{lem:b_les_x}, we have $f_{x_{\epsilon},\gamma}(0)>\xi$ and $f_{\overline{x}_{\epsilon,\xi},\gamma}(0)<\xi$, therefore there exists $ b^*(\gamma)$ such that $f_{b^*(\gamma),\gamma}(0) = \xi$.
    
    {\bf Step 2.} We claim that
    $b^*(\gamma): = \inf\{b\in(\underline{b},\hat{b}): v_{b,\gamma}(0) = \xi_0\}$ is well-defined and that $v_{b,\gamma}(0) < \xi_0$ for every $b\in( b^*(\gamma), \hat{b})$.

 By Lemma~\ref{lem:lowerbound_g}\;iii., we have the continuity of $b\mapsto v_{b,\gamma}(0)$ with fixed $\gamma\le 0$. There $b^*(\gamma): = \inf\{b\in(\underline{b},\hat{b}): v_{b,\gamma}(0) = \xi_0\}$ is well-defined.

    Now we claim that $v_{b,\gamma}(0) < \xi_0$ for every $b\in( b^*(\gamma), \hat{b})$. 
    From Lemma~\ref{lem:lowerbound_g}\;i. together with $b^*(\gamma)>\underline{b}$, it follows that $g_{b,\gamma}(x) \ge 1/\psi^+(x)$ for every $x\in [b^*(\gamma),b]$. Combining this with the fact that $(\psi^+)'(x) < 1$ for every $x\in[b^*(\gamma),b]$ (see Assumption~\ref{as:bdy_cond}\;ii.) ensures that $g_{b,\gamma}(x) > (\psi^+)'(x)\psi^+(x)$ for every $x\in [b^*(\gamma),b]$. 
    
    Furthermore, as $b>b^*(\gamma)$, Lemma \ref{lem:lowerbound_g}\;ii. ensures that $g_{b,\gamma}(x) > g_{b^*(\gamma),\gamma}(x)$ for $x\in(0,b^*(\gamma)]$. Hence we conclude that
    \begin{align*}
        v_{b,\gamma}(0) &=\psi^+(b) \exp\bigg({\int_{b}^{b^*(\gamma)}g_{b,\gamma}(y)dy + \int_{b^*(\gamma)}^{0}g_{b,\gamma}(y)dy }\bigg)\\
        %& <\psi^+(b) \exp\left({\int_{b}^{b^*(\gamma)}\frac{(\psi^+)'(y)}{\psi^+(y)}dy + \int_{b^*(\gamma)}^{0}g_{b^*(\gamma),\gamma}(y)dy }\right) \\
        & <\psi^+(b^*(\gamma))  \exp \bigg(\int_{b^*(\gamma)}^{0}g_{b^*(\gamma),\gamma}(y)dy \bigg)  = v_{b^*(\gamma),\gamma}(0) = \xi_0.
    \end{align*}
    
    {\bf Step 3}. We claim that $b^*(\gamma)$ increases in $\gamma<0$, which ensures that 
    $b^*$ defined in \eqref{def:b*gamma} is well-defined by monotone convergence. Furthermore, we claim that $b^*\in (\underline{b},\hat{b})$ and $v_{b^*}(0) = \xi_0.$
    
    Let $b\in(\underline{b},\overline{b})$ and $\gamma_1,\gamma_2\in(-\infty,0]$ with $\gamma_1\ge\gamma_2$. Then Lemma \ref{lem:pertub_cts}\;ii. ensures that $g_{b,\gamma_1}(x) < g_{b,\gamma_2}(x)$ for every $[0,b]$. From this and the relationship given in \eqref{eq:equiv_deriv_2}, it follows that $v_{b,\gamma_1}(x) \ge v_{b,\gamma_2}(x)$ on $[0,b]$. In particular, since it cannot be the case that $g_{b,\gamma_1} \equiv v_{b,\gamma_2}$ for $x\in[0,b]$, hence $v_{b,\gamma_1}(0) > v_{b,\gamma_2}(0)$.   
   
    From this and the fact that $b^*(\gamma_1)=\inf\{b\in(\underline{b},\hat{b}): v_{b,\gamma_1}(0) = \xi_0\}$, it follows that $\xi_0 = v_{b^*(\gamma_1),\gamma_1}(0) >  v_{b^*(\gamma_1),\gamma_2}(0)$. Furthermore, Step~2 ensures that $v_{b,\gamma_2}(0) < \xi_0$ for every $b\in( b^*(\gamma_2), \hat{b})$. We hence have $b^*(\gamma_1) > b^*(\gamma_2)$. 

    By the monotonicity of $b^*(\gamma)$ and its uniform boundedness, we can apply the MCT to have the existence of $b^*$ in \eqref{def:b*gamma}. 
    
    It remains to show that $v_{b^*}(0) = \xi_0$ and $b^*\in (\underline{b},\hat{b})$. To that end, consider an increasing sequence $(\gamma_n)_{n\in \mathbb{N}}$ such that $\gamma_n\in (-\infty,0)$ for every $n\in\mathbb{N}$ and $\lim_{n\rightarrow \infty } \gamma_n = 0$.

    Since $\lim_{n\rightarrow \infty} b^*(\gamma_n) = b^*$, Lemma~\ref{lem:lowerbound_g}\;iii. ensures that for any $\varepsilon>0$, there exists $N_\varepsilon \in \mathbb{N}$ satisfying that for every $n\geq N_\varepsilon$, $
        |v_{b^*}(0) - v_{b^*(\gamma_n)}(0)|<\varepsilon.$
    Furthermore Lemma \ref{lem:pertub_cts}\;iii. ensures that there exists $\widetilde{N}_\varepsilon\in \mathbb{N}$ satisfying that  for every $n\geq \widetilde{N}_\varepsilon$, $|v_{b^*(\gamma_n),\gamma_n}(0) - v_{b^*(\gamma_n)}(0)|<\varepsilon.$
    Since $v_{b^*(\gamma_n),\gamma_n}(0) = \xi_0$ for every $n\in\mathbb{N}$ (see Step 2.), we hence obtain that for every $n\geq \widetilde{N}_\varepsilon \vee N_\varepsilon$,
    \begin{align*}
        |v_{b^*}(0) - \xi_0 |\leq |v_{b^*}(0) - v_{b^*(\gamma_n)}(0)|+|v_{b^*(\gamma_n),\gamma_n}(0) - v_{b^*(\gamma_n)}(0)|  < 2\varepsilon.
    \end{align*}
    By letting $\varepsilon\downarrow 0$, we show $v_{b^*}(0) = \xi_0$. Finally, we know that $\underline{b}<b^*(\gamma)< \hat{b}$ for every $\gamma<0$ by Step 1, therefore $\underline{b}\le \lim_{\gamma\uparrow 0}b^*(\gamma) = b^*\le \hat{b}$. The strict inequality $\underline{b}< b^*$ follows by the increasing monotonicity of $b^*(\gamma)$.  
    %This completes the proof. 
\end{proof}
%%%%%%%%%%%%%%%%%%%%%%%%%%%%%%%%%%%%%%%%%%%%%%%%%%%%%%%%%%%%%%%%%
\subsubsection{Proof of Proposition \ref{pro:shooting_sub}.}
%{\it Proof of Proposition \ref{pro:shooting_sub}.}
\begin{proof}
Given the existence of $b^*\in(\underline{b},\hat{b}]$ satisfying $v_{b^*}(0)=\xi_0$ by Lemma \ref{lem:comparision}\;ii., we first prove that $v_{b^*}'(x)\geq 1$ for every $x\in[0,b^*]$.
    %The main idea is based on a perturbation analysis with respect to $\gamma\in (-\infty,0]$. 

    From Lemma \ref{lem:comparision}\;ii., define $b^{\gamma}$ for every $\gamma<0$ by
    \begin{align}\label{eq:quad_perturb_0}
    b^{\gamma}:=b^*(\gamma)= \inf\{b\in(\underline{b},\hat{b}): v_{b,\gamma}(0) = \xi_0\}.
    \end{align}
    Furthermore, denote by $
    v_{b^{\gamma},\gamma}\in C^2((0,b^{\gamma})\cup(b^{\gamma},\infty))\cap C^1(\mathbb{R}_+)$ the solution of \eqref{eq:HJB_ODE1} (when $b=b^{\gamma}$). 
    
    The proof is achieved in two steps.

    {\bf Step 1.} %We start with perturbation analysis with a f
    Fix $\gamma\in(-\rho,0)$. We  claim that
    \begin{align}\label{eq:claim_1_shooting}
        v_{b^\gamma,\gamma}'(b^\gamma)(x)\ge 1,\quad \mbox{for all}\quad  x\in(0,b^{\gamma}].
    \end{align}
    From the boundary conditions $v_{b^{\gamma},\gamma}'(b^{\gamma})=1$ and $v_{b^{\gamma},\gamma}(b^{\gamma})= \psi^+(b^{\gamma})$, the $C^2$ property of $v_{b^{\gamma},\gamma}$ and the ODE \eqref{eq:HJB_ODE1} gives  $v''_{b^{\gamma},\gamma}(b^{\gamma}) = 2\gamma\sigma^2(b^{\gamma})\psi^+(b^{\gamma})<0.$ Argue by contradiction, assume \eqref{eq:claim_1_shooting} doesn't hold. Then combined with the fact that $v_{b^\gamma,\gamma}'(b^\gamma)=1$, we have the existence of
    \[
    x_1:=\sup\{ 
    x\in(0,b^\gamma): v'_{b^\gamma,\gamma}(x) = 1 \}.
    \]
    In the following, we prove that the existence of $x_1$ leads to a contradiction. %Hence \eqref{eq:claim_1_shooting} holds.
    
    i. {\it  Case $x_1 \in (\underline{b},{b}^\gamma)$}. Since $v_{b^\gamma,\gamma}'(x_1)= 1$ (noting that $v_{b^\gamma,\gamma}\in C^1(\mathbb{R}_+)$), it follows that $v_{b^\gamma,\gamma}'(x) > 1$ for every $x\in(x_1,b^\gamma)$. Furthermore, since $\psi^+(x)-x$ decreases on $[x_1,b^\gamma]\subset [\underline{b},\hat{b}]$ (see Assumption \ref{as:bdy_cond}\;ii. and $x_1 \in (\underline{b},{b}^\gamma)$) and $v_{b^\gamma,\gamma}(b^\gamma)=\psi^+(b^\gamma)$, hence
    \begin{align*}%\label{eq:x1_monotone}
        v_{b^\gamma,\gamma}(x_1) \le \psi^+(b^\gamma)-(b^\gamma-x_1)\le \psi^+(x_1).
    \end{align*}
    Moreover, since $x_1 \in (\underline{b},\hat{b})$ and $v_{b^\gamma,\gamma}(0)= \xi_0$ (by definition of $b^\gamma$), Lemma\;\ref{lem:comparision}\;i. ensures that $v_{b^\gamma,\gamma}(x_1) \ge x_1 +\xi_0 \ge \psi^-(x_1)$, so that $Q(v_{b^\gamma,\gamma}(x_1);x_1)\le 0$ by recalling that $Q(\cdot;x_1)$ defined in Assumption \ref{as:bdy_cond}\;ii. This implies that
    \begin{align*}
        v''_{b^\gamma,\gamma}(x_1) = \frac{2}{\sigma^2(x_1)}\Big(\gamma v_{b^\gamma,\gamma}(x_1) +\frac{Q(v_{b^\gamma,\gamma}(x_1);x_1)}{v_{b^\gamma,\gamma}(x_1)}\Big)\le
        %-\mu(x_1) + \dfrac{R}{2}\dfrac{\sigma^2(x_1)}{v_{b^\gamma,\gamma}(x_1)} \le 
        \frac{2}{\sigma^2(x_1)}\gamma (x_1+\xi_0)<0.
    \end{align*}
    This together with the fact that $v''_{b^{\gamma},\gamma}(b^{\gamma})<0$ contradicts Lemma~\ref{lem:ele}\;ii. 

    ii. {\it Case $x_1\in(0,\underline{b}]$}. 
    Recalling $x_1=\sup\{ 
    x\in(0,b^\gamma): v'_{b^\gamma,\gamma}(x) = 1 \}$ and the fact that $v_{b^\gamma,\gamma}'(b^\gamma)=1$ and $v_{b^\gamma,\gamma}''(b^\gamma)<0$, we employ Lemma~\ref{lem:ele}\;ii. to obtain that $v''_{b^\gamma,\gamma}(x_1)\geq 0$, which can be rewritten~by
    \begin{align}\label{eq:quad_perturb}
        v''_{b^\gamma,\gamma}(x_1)= \frac{2}{\sigma^2(x_1)}\frac{1}{v_{b^\gamma,\gamma}(x_1)}\widetilde{Q}_{\gamma}(v_{b^\gamma,\gamma}(x_1);x_1)\geq 0,
    \end{align}
    where $\widetilde{Q}_\gamma(\psi;x_1):=(\rho+\gamma) \psi^2-\mu(x_1)\psi + \frac{R}{2}\sigma^2(x_1)$ (which is well-defined by Assumption \ref{as:bdy_cond}\;ii. with the fact that $\rho>\rho+\gamma>0$) and we have used the nonlinear ODE of $v_{b^\gamma,\gamma}$ given in \eqref{eq:HJB_ODE1}.

    Lemma~\ref{lem:comparision}\;i. together with Assumption \ref{as:bdy_cond}\;iii. (and $x_1\in(0,\underline{b}]\subset [0,\hat{b}]$) ensures that $v_{b^\gamma,\gamma}(x_1) \geq x_1+\xi_0\geq  \psi^-(x_1)$. Furthermore, recalling that $\gamma$ satisfies $\rho>\rho+\gamma>0$, we have
    \begin{align}
    v_{b^\gamma,\gamma}(x_1) \geq& \frac{\mu(x_1) - \sqrt{\mu^2(x_1)-2R\rho \sigma^2(x_1)}}{2\rho}\label{eq:quad_perturb2.1} \\
    >&\frac{\mu(x_1) - \sqrt{\mu^2(x_1)-2R(\rho+\gamma) \sigma^2(x_1)}}{2(\rho+\gamma)} =:\widetilde{\psi^-}(x_1),\label{eq:quad_perturb2}
    \end{align}
    where we note that RHS of \eqref{eq:quad_perturb2.1} equals $\psi^-(x_1)$ and that the last term $\widetilde{\psi^-}(x_1)$ is the smaller one among two roots of the quadratic function $\widetilde{Q}_\gamma(\cdot;x_1) = 0$. 

    Since $\widetilde{Q}_{\gamma}(v_{b^\gamma,\gamma}(x_1);x_1)\geq0$ (see \eqref{eq:quad_perturb}), from \eqref{eq:quad_perturb2} it follows that 
    \begin{align}\label{eq:quad_perturb3}
    v_{b^\gamma,\gamma}(x_1)\geq \frac{\mu(x_1) + \sqrt{\mu^2(x_1)-2R(\rho+\gamma) \sigma^2(x_1)}}{2(\rho+\gamma)}=:\widetilde{\psi^+}(x_1),
    \end{align}
    where the right-hand side $\widetilde{\psi^+}(x_1)$ is the other (larger) root of $\widetilde{Q}_\gamma(\cdot;x_1)$.

    Now we claim that for $x$ fixed, if $(\psi^+)'(x) \ge 1$ then $(\widetilde{\psi^+})'(x) \ge 1$.  By taking derivative of $\psi^+(x)$ w.r.t. $x$, we have 
   \begin{align*}
       (\psi^{+})'(x) = \frac{1}{2\rho}\mu'(x) + \frac{\mu(x)\mu'(x)}{2\rho\sqrt{\mu(x)^2-2R\rho\sigma^2(x)}} - \frac{R\sigma(x)\sigma'(x)}{\sqrt{\mu(x)^2-2R\rho\sigma^2(x)}}.
   \end{align*}
When $(\psi^+)'(x) \ge 1$, it is clear that $\mu'(x)>0$ since $\sigma'(x)\ge 0$ by Assumption~\ref{as:bdy_cond}.  %Therefore both the terms $\frac{1}{2\rho}\mu'(x)$ and $- \frac{R\sigma(x)\sigma'(x)}{\sqrt{\mu(x)^2-2R\rho\sigma^2(x)}}$ are decreasing w.r.t $\rho$. 
Therefore the first term and the last term in above equation are decreasing w.r.t $\rho$.
It remains to show the same monotonicity for the second term in $(\psi^{+})'(x)$. 

By differentiating $h(\rho;x): = \rho\sqrt{\mu(x)^2-2R\rho\sigma^2(x)}$ w.r.t $\rho$, we have
\begin{align*}
    \partial_{\rho} h(\rho;x) = \sqrt{\mu(x)^2-2R\rho\sigma^2(x)} - \frac{R\rho \sigma^2(x)}{\sqrt{\mu(x)^2-2R\rho\sigma^2(x)}} \ge 0, 
\end{align*}
where the inequality is followed by Assumption~\ref{as:bdy_cond}\;i. with $(x_2,x_1)\subset [0,\underline{b}]$. As a result, we show that $1 \le (\psi^{+})'(x)\le (\widetilde{\psi^+})'(x)$ since $\rho+\gamma < \rho$ for $\gamma<0$. Combined with Assumption ~\ref{as:bdy_cond}\;ii., which says that $\psi^+(x)-x$ increases on $[0,\underline{b}]$, we have $\widetilde{\psi^+}(x)-x$ increases on $[0,\underline{b}]$.

     To proceed, recalling that $v'_{b^\gamma,\gamma}(x_1) = 1$ and $v''_{b^\gamma,\gamma}(x_1)\geq 0$ (see \eqref{eq:quad_perturb}), 
      standard arguments by using a proof by contradiction shows that there exists $x_2$ defined as $x_2:=\sup\{ x\in(0,x_1): v'_{b^\gamma,\gamma}(x) = 1 \}$.
    Since $v'_{b^\gamma,\gamma}(x)<1$ for every $x\in(x_2,x_1)$ (from the definition of $x_2$ and the fact that $v'_{b^\gamma,\gamma}(x_1) = 1$ and $v''_{b^\gamma,\gamma}(x_1)\geq 0$), the following hold: $ v'_{b^\gamma,\gamma}(x_2)=1$ and 
    \begin{align}
    %\begin{aligned}
    v_{b^\gamma,\gamma}(x_2)&>v_{b^\gamma,\gamma}(x_1)-(x_1-x_2)\nonumber\\
    &\geq \frac{\mu(x_1) + \sqrt{\mu^2(x_1)-2R(\rho+\gamma) \sigma^2(x_1)}}{2(\rho+\gamma)}-(x_1-x_2)\nonumber\\
    &\geq  \frac{\mu(x_2) + \sqrt{\mu^2(x_2)-2R(\rho+\gamma) \sigma^2(x_2)}}{2(\rho+\gamma)}, \label{eq:quad_perturb4}
    %\end{aligned}
    \end{align}
    where the second inequality follows from \eqref{eq:quad_perturb3} and  the last inequality follows from
    the fact that $\widetilde{\psi^+}(x)-x$ increases on $(x_2,x_1)$ since $(x_2,x_1)\subset [0,\underline{b}]$.

    \noindent Recalling the nonlinear ODE \eqref{eq:HJB_ODE1}, the inequality given in \eqref{eq:quad_perturb4} ensures that
    \[
        v''_{b^\gamma,\gamma}(x_2)=\frac{2}{\sigma^2(x_2)}\frac{1}{v_{b^\gamma,\gamma}(x_2)}\widetilde{Q}_{\gamma}(v_{b^\gamma,\gamma}(x_2);x_2)>0.
    \]
    This together with \eqref{eq:quad_perturb} contradicts Lemma~\ref{lem:ele}. (Lemma 2.2\;ii. requires $v''_{b^\gamma,\gamma}(x_1)>0$, but we can see it cannot be the case that $v''_{b^\gamma,\gamma}(x_2)>0$, $v''_{b^\gamma,\gamma}(x_1)=0$ and $v''_{b^\gamma,\gamma}(b^{\gamma})<0$).
    
    By deriving contradictions for the above two cases, we have \eqref{eq:claim_1_shooting}. 

    %\vspace{0.5em}
    
    {\bf Step 2.}  Recall $b^\gamma$ for $\gamma<0$ given in \eqref{eq:quad_perturb_0} and set $b^*:=\lim_{\gamma\uparrow 0} b^\gamma$ (see Lemma~\ref{lem:comparision}\;ii.). We claim that $v'_{b^*}(x)\ge 1$ for every $x\in (0,b^*]$. 

    To that end, let us consider an increasing sequence $(\gamma_n)_{n\in\mathbb{N}}$ such that $-\rho <\gamma_n<0$ for every $n\in\mathbb{N}$ and $\lim_{n\rightarrow \infty } \gamma_n = 0$. Since $\lim_{j\rightarrow \infty} b^{\gamma_n} = b^*$ and $(b^{\gamma_n})_{n\in\mathbb{N}}$ is a increasing sequence (see Lemma~\ref{lem:comparision}\;ii.), the following hold: $
    b^*\geq  b^{\gamma_n} \geq b^{\gamma_0}$ for every $n\in \mathbb{N}$. 

    Combining this with Step 1.~(along with $v_{b^{\gamma_n},\gamma_n}\in C^1(\mathbb{R}_+)$ for every $n\in\mathbb{N}$) ensures that for every $n\in\mathbb{N}$ and every $x\in (0,b^*]$, 
    \begin{align}\label{eq:quad_perturb5}
    v_{b^{\gamma_n},\gamma_n}'(x) \ge 1.
    \end{align}
    
    Fix $y\in(0,b^{\gamma_0}]\subset (0,b^*]$. For arbitrary $\varepsilon>0$, then Lemma~\ref{lem:pertub_cts}\;iii. ensures that there exists ${N}_\varepsilon\in \mathbb{N}$ such that $|v'_{b^{\gamma_n},\gamma_n}(x) - v'_{b^{\gamma_n}}(x)|<\varepsilon$ for every $n\geq {N}_\varepsilon$. Furthermore, Lemma \ref{lem:lowerbound_g}\;iii. ensures that there exists $\widetilde{N}_\varepsilon\in \mathbb{N}$ such that $
    |v'_{b^*}(x) - v'_{b^{\gamma_n}}(x)|<\varepsilon$ for every $n\geq \widetilde{N}_\varepsilon$. Therefore, from \eqref{eq:quad_perturb5}, it follows that for every $n\geq \widetilde{N}_\varepsilon\vee N_\varepsilon$, $v'_{b^*}(x) >v'_{b^{\gamma_n}}(x)  -\varepsilon> v'_{b^{\gamma_n},\gamma_n}(x)- 2\varepsilon\geq 1- 2\varepsilon.$ 
    By letting $\varepsilon\rightarrow0$, we conclude that $v_{b^*}'(x)\ge1$ for $y\in (0, b^{\gamma_0}]$.

    It remains to show $v_{b^*}'(y)\ge1$ on $(b^{\gamma_0},b^*]$. \\
    Assume the contrary: $\min_{y\in[b^{\gamma_0},b^*]} v_{b^*}'(y) <1$. Then $v_{b^*}'(y) -1$ takes its minimum at $x_0\in(b^{\gamma_0},b^*)$ because $v'_{b^*}(b^*) = 1$ and $v'_{b^*}(b^{\gamma_0})\ge1$ by above analysis. 
Then we have that $v_{b^*}''(x_0) = 0 $ and $v'_{b^*}(x_0)<1$. 
By the continuity of $v'_{b^*}$, let $\epsilon=(1-v'_{b^*}(x_0))/2$ there exists $x_0'$ in a local neighborhood of $x_0$ such that $v_{b^*}''(x_0') > 0$ and  $v_{b^*}'(x_0')< v'_{b^*}(x_0) + \epsilon<1$.
Hence $$g_{b^*}(x_0') = \dfrac{v'_{b^*}(x_0')}{v_{b^*}(x_0')} < \frac{1}{x_0'+\xi_0} < \frac{1}{\psi^{-}(x_0')},$$
where the first inequality follows Lemma~\ref{lem:comparision}\;i. that $v_{b^*}(x_0')\ge x_0'+\xi$ since $v_{b^*}(0)=\xi$, while the last inequality follows Assumption \ref{as:bdy_cond}\;iii. 

Recalling the ODE \eqref{eq:HJB_ODE1} of $v_{b^*}$, we have
\begin{align*}
    \frac{1}{2}\sigma^2(x_0) v''_{b^*}(x_0') = \frac{1}{v_{b^*}(x_0')}Q(v_{b^*}(x_0');x_0') > 0.
\end{align*}
Combined with the above arguments, we conclude that $g_{b^*}(x_0') < 1/\psi^+(x_0')$. This contradicts with Lemma~\ref{lem:lowerbound_g}\;i. which gives $g_{b^*}(x_0')\ge{1}/{\psi^+(x_0')}$ as $x_0>b^{\gamma_0}>\underline{b}$.
As a result, we see $v_{b^*}'(y)\ge1$ on $(b^{\gamma_0},b^*]$.

    %This completes the proof for the statement $v_{b^*}'(x)\ge 1$ on $[0,b^*]$.
    %\vspace{0.5em}
    Finally, we show that $b^*<\hat{b}$, where $\hat{b}$ defined in Assumption \ref{as:bdy_cond}\;iii. It suffices to show $b^* \neq \hat{b}$ since $b^*\le \hat{b}$ by Lemma \ref{lem:comparision}\;ii. When $b^* = \hat{b}$, by definition of $\hat{b}$, we have $v_{b^*}(0)=\xi_0$ while $v_{b^*}(b^*) = \psi^+(b^*) =  b^* + \xi_0$. Together with our previous result that $v_{b^*}'(x)\ge 1$ on $(0,b^*]$, we conclude $v_{b^*}'(x) \equiv 1$ on $(0,b^*]$, which cannot be the case; see \eqref{eq:HJB_ODE1}. Therefore, we have $b^*<\hat{b}$.  
    \end{proof}

    \subsubsection{Proof of Proposition~\ref{pro:shooting}}
    \begin{proof}
    From Proposition \ref{pro:shooting_sub}, define $v^*:{\mathbb{R}_+}\rightarrow \mathbb{R}$ by
    \begin{align}\label{eq:const_v}
    v^*(x):=\left\{
    \begin{aligned}
        &v_{b^*}(x)\quad &&\mbox{if}\;\;x\in [0,b^*];\\
        &x-b^*+\psi^+(b^*)\quad && \mbox{if}\;\;x\in (b^*,\infty).
    \end{aligned}
    \right.
    \end{align}
    We note that $v^*(= v_{b^{*}})$ on $[0,b^*]$ satisfies the first and third lines of \eqref{eq:fbdy}, and that $(v^*)^{\prime}=1$ on $(b^*,\infty)$. Thus, it suffices to show that ${\cal L}v^*\leq 0$ on $(b^*,\infty)$.

    The function $ v_{b^*} $ is in $\mathcal{C}^2(\mathbb{R}_+)$ and satisfies \eqref{eq:HJB_ODE1} with $\gamma = 0$. 
    In particular, $v_{b^*}''(b^*) = 0$ by construction. Indeed, substituting the boundary conditions $v_{b^*}'(b^*) = 1$ and $v_{b^*}(b^*) = \psi^+(b^*)$ into the ODE at $x = b^*$ gives
    \[
    \frac{\sigma^2(b^*)}{2} v_{b^*}''(b^*) + \mu(b^*) - R \frac{\sigma^2(b^*)}{2} \frac{1}{\psi^+(b^*)} - \rho \psi^+(b^*) = 0.
    \]
    Since $\psi^+(b^*)$ solves the quadratic equation $Q(\psi; b^*) = 0$, it follows that $v_{b^*}''(b^*) = 0$. 
        By the definition given in \eqref{eq:const_v}, the function $ v^* $ is in $\mathcal{C}^2({\mathbb{R}_+})$ and satisfies the free boundary problem \eqref{eq:fbdy} to the left of $ b^* $. It remains to show that for $x\in(b^*,\infty)$, ${\cal L} v^*\le 0$ (as $({v^*})^{\prime}(x)=1$ is obvious). 
    
    For $ x \in (b^*, \overline{b}) $, note that $\psi^+(x) - x$ is decreasing, as stated in Assumption~\ref{as:bdy_cond}~ii.. Consequently, we have $ v^*(x) = x - b^* + \psi^+(b^*) \ge \psi^+(x) $. Hence $\mathcal{L}v^*(x)$ is given by
    $$
    \mathcal{L}v^*(x) = -\frac{Q(v^*(x); x)}{v^*(x)} \le 0,
    $$
    where $ Q(\psi; x) $ is defined in Assumption \ref{as:bdy_cond}~ii..
    
    If $\overline{b} = \infty$, the proof concludes here. Otherwise, for $ x \in [\overline{b}, \infty) $, we have $ Q(v^*(x); x) < 0 $, which implies $\mathcal{L}v^*(x) < 0$. Thus, we conclude that $ v^*(x) $ satisfies $\mathcal{L}v^* \le 0$ on $(b^*, \infty)$. This means it solves \eqref{eq:fbdy} with the free boundary $ b^* \in (\underline{b}, \hat{b})$. This completes the proof. 
     
    \end{proof}
    
    %%%%%%%%%%%%%%%%%%%%%%%%%%%%%%%%%%%%%%%%%%%%%%%%%%%%%%%%%%%%%%%%%%%%%%%%%%%%%%%%%%%%%%%%%%%%%%%%%%%%%%%%%%%%%%%%%%%%%%%%%%%%%%%%%%%%%%%%%%
    \subsection{Verification theorem: Proof of Theorem~\ref{thm:verification}
    }\label{sec:thm:verification}
    To establish the boundary condition  ${\cal V}(0+)=\xi_0$  for value function ${\cal V}$ (defined in \eqref{eq:obj_maxmin}), we  first present a technical result for the process $X^x$.
    % We first present a technical result for $X^x$, $x\in {\cal X}$, given in \eqref{eq:uncontrolX} to establish the boundary condition ${\cal V}(0+)=\xi_0$ (see \eqref{eq:boundary_con_B} below) for ${\cal V}$ (defined in \eqref{eq:obj_maxmin}) in Lemma \ref{lem:v_bound_con}. 
    To that end, we define the following stopping time for each $x\in \mathcal{X}$%\mycomment{must be $\ell$, not~$-\ell$?}
        \begin{align}\label{dfn:zeta_y|x}
          \zeta_{y|x} = \inf\{t\ge 0, X_t^x = y \}, \quad y\in\mathcal{X} \cup \{\ell,\infty\}.
        \end{align}
        Assumption \ref{as:recurrent} ensures that for all $x\in\mathcal{X}/\{0\}$,
        $ \mathbb{P}(\zeta_{0|0} = 0) =\mathbb{P}(\zeta_{0|x}<\infty) = 1$ and $\mathbb{P}(\zeta_{\infty|x} <\infty) =0$.
    %Such a result together with  Proposition~\ref{lem:v_bound_con} extend \cite[Lemma~1, Proposition~1]{chakraborty2023optimal} from bounded drift and volatility diffusion model to a more general diffusion under Condition~\ref{con:stand}. 
    \begin{lem}\label{lem:X_bound} 
    Under Assumption \ref{as:recurrent}, for any $\delta>0$, we have 
    \begin{align} \label{eq:lem_bound_0}
    \lim_{x\rightarrow h} \mathbb{P}(\zeta_{h|x}<\delta) = 1, 
    \end{align}
    where stopping time $\zeta_{y|x}$ is defined in \eqref{dfn:zeta_y|x}. In particular, for fixed $y>0$,
     \begin{align}
         \lim_{x \downarrow 0} \mathbb{P}(\zeta_{0|x} < \delta)  = 1,
         \quad \lim_{x\downarrow 0}  \mathbb{E}\Big[\max_{0\le s\le \delta\wedge{\zeta}_{0|x}\wedge{\zeta}_{y|x}} X^x_s\Big] = 0.\label{eq:lem_bound_2}
     \end{align}
     \end{lem}
    \begin{proof}
    The proof involves a transformed process $U_t^x$ of $X_t^x$, which is a time-changed Brownian motion, and we prove that auxiliary versions of \eqref{eq:lem_bound_0}, \eqref{eq:lem_bound_2} hold with the process $U_t$. 
    We define the {\it scale function} of $X^x$ in \eqref{eq:uncontrolX} as 
    \begin{align*}%\label{def:scale_func}
    \mathfrak{s}(x): = \int_{0}^{x}\exp\Big(-2\int_0^z\frac{\mu(y)}{\sigma^2(y)}dy\Big)dz,
    \end{align*}
    which is strictly increasing, continuously differentiable bijection of the interval $\mathcal{X}$ onto the interval $I=(\iota_{l},\iota_{\infty})$ with endpoints $\iota_l:= \mathfrak{s}(\ell+)$ and $\iota_{\infty}:= \mathfrak{s}(\infty-) $; see, e.g., Borodin and Salminen~\cite[II.1.4]{borodin2015handbook}.%\mycomment{Section or theorem II.1.4?} 
    Denote by $\mathfrak{s}^{-1}$ the inverse mapping of $\mathfrak{s}$. We define $U_t^x : = \mathfrak{s}(X^x_t)$, with state space $I$, and its dynamic is given by
    \begin{align}\label{pf:lem:X_bound_0}
    U^x_t = \mathfrak{s}(X^x_0) + \int_0^t \varsigma(U^x_s) dW_s, 
    \end{align}
     where the dispersion function $\varsigma(u): = (\mathfrak{s}'\cdot\sigma)(\mathfrak{s}^{-1}(u))$; see e.g., Karatzas and Shreve~\cite[Section~5.5B]{KS1991}.
     To further remove the dependency on the function $\varsigma$ in \eqref{pf:lem:X_bound_0}, we consider a time change for transformation. 
     According to the Dambis–Dubins–Schwarz theorem (see Revuz and Yor~\cite[Theorem V.1.6]{revuz2013continuous}), one can rewrite the stochastic integral $\int_{0}^{\cdot} \varsigma(U^x_s) dW_s$ as a time-changed Brownian motion on (an extension of) the probability space, i.e., $\hat{W}_{\Lambda^{U^x}(\cdot)}$ that we define as follows. 
     Define a time scale $\Gamma^{U^x}(r)$ for $r \in{\mathbb{R}_+}$ as
    \begin{align*}
      \Gamma^{U^x}(r):=\int_0^{r} \frac{dt}{\varsigma^2\big(\mathfrak{s}(X_0^x)+\hat{W}_t\big)}
    \end{align*}
    in terms of a standard Brownian motion $\hat{W}$, and denote by $\Lambda^{U^x}$ the inverse mapping of $ \Gamma^{U^x}$, i.e. $\Lambda^{U^x}(t): = \inf\{r\ge 0: \Gamma^{U^x}(r)>t\}$. The representation 
    \begin{align}\label{pf:lem:X_bound_1}
    U_t^{x} =  \mathfrak{s}(X_0^x) + \hat{W}_{\Lambda^{U^x}(t)}
    \end{align}
    holds. Due to the strict monotonicity of the scale function $\mathfrak{s}$, it is sufficient to prove auxiliary properties of \eqref{eq:lem_bound_0}-\eqref{eq:lem_bound_2} for the process $U^x$ given by \eqref{pf:lem:X_bound_1}.
    For convenience of presentation, we adopt the notation 
    \begin{align*}
        &\zeta^U_{y|x} :=\inf\{t\ge 0: U_t^x = \mathfrak{s}(y)\}. %\text{ given }U_0 =\mathfrak{s}(x)\}.
    \end{align*}
    From the continuous path properties of standard Brownian motion together with the fact that $\Lambda^{U}(t)>0$ holds for all $t>0$, it is proved (see e.g.\cite[II.1.3]{borodin2015handbook})%\mycomment{Lemma V.46?} that for every $x \in \mathcal{X}$,
    \begin{align*}
    &\mathbb{P}\Big(\big \{\omega\in\Omega: \exists \epsilon(\omega)>0 \text{ s.t. } U_t^x \le \mathfrak{s}(x),\forall t\in[0,\epsilon(\omega)] \big \}\Big)=0,\\
    & \mathbb{P}\Big(\big \{\omega\in\Omega: \exists \epsilon(\omega)>0 \text{ s.t. } U_t^x \ge \mathfrak{s}(x),\forall t\in[0,\epsilon(\omega)] \big \} \Big)=0.
    \end{align*}
    Given the above statements, we have for any fixed $\delta>0$,
    \begin{align*}
       &\mathbb{P}(  \zeta^U_{a|x} \ge \delta) = \mathbb{P}\Big(\inf_{t\in[0,\delta]} U_t^x \ge \mathfrak{s}(a) \Big) \rightarrow \mathbb{P}\Big(\inf_{t\in[0,\delta]} U_t^x > \mathfrak{s}(x) \Big) = 0,\\
       &   \mathbb{P}(  \zeta^U_{b|x} \ge \delta) = \mathbb{P}\Big(\sup_{t\in[0,\delta]} U_t^x \ge \mathfrak{s}(b)\Big) \rightarrow \mathbb{P}\Big(\sup_{t\in[0,\delta]} U_t^x < \mathfrak{s}(x) \Big) = 0, 
    \end{align*}
    as $a\uparrow x$ and $b\downarrow x$ respectively.
    Hence, we conclude \eqref{eq:lem_bound_0}.
    
    To show \eqref{eq:lem_bound_2}, we recognize from the time change property that 
    \begin{align*}
        \max_{0\le s\le \delta\wedge{\zeta}^U_{0|x}\wedge{\zeta}^U_{y|x}} U_s^x &= \max_{0\le r \le \Lambda^{U^x}(\delta)\wedge\Lambda^{U^x}(\zeta^U_{0|x})\wedge\Lambda^{U^x}(\zeta^U_{y|x})} \mathfrak{s}(x)+ \hat{W}_r\\
        &= \max_{0\le r \le \Lambda^{U^x}(\delta)\wedge\hat{\zeta}_{0|x}\wedge\hat{\zeta}_{y|x}} \mathfrak{s}(x)+ \hat{W}_r,
    \end{align*}
    where 
    $
      \hat{\zeta}_{y|x} = \inf\{ r\ge 0: B_x(r)= \mathfrak{s}(y) \}$, with $B_x(r):= \mathfrak{s}(x)+ \hat{W}_r$. 
    
    By the monotonicity of $\mathfrak{s}$, $B_x(t,\omega)$ decreases as $x\downarrow 0$.
    By the law of iterated logarithm (see Revuz and Yor~\cite[Chapt.~II, Theorem~19]{revuz2013continuous}), we have $\hat{\zeta}_{0|x}\rightarrow0$ as $x\downarrow 0$. By definition, $\Lambda^{U^x}(\delta)$ is finite a.s. for $\delta\in[0,\infty)$. Since $B_x(r)$ has continuous path, considering a decreasing sequence $(x_n)_{n\in\mathbb{N}}$ with  $x_n<y$ and $\lim_{n\rightarrow\infty} x_n=0 $, we have 
    $
    \lim_{n\rightarrow\infty} \max_{0\le r \le \Lambda^{U^x}(\delta)\wedge\hat{\zeta}_{0|x_n}\wedge\hat{\zeta}_{y|x_n}} B_{x_n}(r) =  \mathfrak{s}(0),$
    and %this maximum is bounded by 
    $\max_{0\le r \le \Lambda^{U^x}(\delta)\wedge\hat{\zeta}_{y|x}}B_x(r)\le \mathfrak{s}(y)$; hence \eqref{eq:lem_bound_2} holds.  
    \end{proof}
    
    %Now we turn to the proof of Lemma~\ref{lem:v_bound_con}.
    \begin{lem}\label{lem:v_bound_con} Suppose that Assumptions \ref{as:recurrent} and \ref{as:default} and the condition \eqref{as:bound} in Assumption \ref{as:bdy_cond} hold. Then it holds that for $x>0$, 
    \begin{align}\label{eq:est_V}\xi_0 \le {\cal V}(x) \le x+ \xi_0 + 
     {\overline{\mu}}/{\rho},
     \end{align} %\mycomment{a bit confusing on the domain of the controlled process $X^D$. which one between ${\cal X}$ or $\mathbb{R}_+$?}
     where $\overline{\mu}$ is the constant that dominates function $\mu(x)-\rho x$ in \eqref{as:bound}. In addition, ${\cal V}$ is continuous on $(0,\infty)$ and non-decreasing subject to 
    		\begin{align}
    			{\cal V}(0+) = \xi_0. \label{eq:boundary_con_B} 
    		\end{align}
    \end{lem}

    \begin{proof}

{We prove \eqref{eq:boundary_con_B} in detail. The proofs of the remaining claims are standard and thus omitted. }
Fix $\delta>0$ and $y>0$, in view of Lemma~\ref{lem:X_bound}, for any $\epsilon\in(0,1)$, choose $x>0$ such that %\mycomment{what is $\mathbb{E}_x$?}
    		\begin{align*}
    			\mathbb{P}(\zeta_{0|x} < \delta)  \ge  1-\epsilon,\qquad \mathbb{E} \bigg[\max_{0\le s\le \delta\wedge{\zeta}_{{0|x}}\wedge\zeta_{y|x}} X_s^x\bigg] \le \epsilon.
    		\end{align*}
    	Set $\hat{\tau} = \tau^{x,D} \wedge \zeta_{y|x} \wedge \delta$. Since $0\le X^{x,D}\le X^x$, we have $\tau^{x,D}\le \zeta_0$ and 
        \begin{align}\label{eq:stopp_est}
        \mathbb{P}(\tau^{x,D} <\zeta_{y|x} \wedge \delta) \ge \mathbb{P}(\zeta_{0|x} <\zeta_{y|x} \wedge \delta) \ge 1-\epsilon.
        \end{align}
     
     Since $X^{x,D}_t \ge 0$ for all $t\ge 0$, we have 
     $D_{\hat{\tau}}\le X_{\hat{\tau}}^x \le \max_{0\le s\le \hat{\tau}}X_s^x\le \max_{0\le s\le \delta\wedge{\zeta}_{{0|x}}\wedge\zeta_{y|x}} X_s^x $.
     Therefore $\mathbb{E}_x[D_{\hat{\tau}}]\le \epsilon$. 

    Combining this with \eqref{eq:stopp_est} ensures that ${\cal V}_0^{x,D,\theta^0}$ in \eqref{dfn:rbs_V_theta} with $\theta^0= (\theta^0(t)\equiv 0 )_{t\ge0}\in \Theta^D$ satisfies
    		\begin{align*}
    			&{\cal V}_0^{x,D,\theta^0}
    			= \mathbb{E}\bigg[\int_0^{\hat{\tau}} e^{-\rho t} dD_t + \mathds{1}_{\hat{\tau}<\tau^{x,D}}\int_{\hat{\tau}}^{\tau^{x,D}} e^{-\rho t} dD_t + e^{-\rho \tau^{x,D}}\xi_0  \bigg]\\
    			&\le \mathbb{E}[D_{\hat{\tau}}] +  \mathbb{E}\bigg[  \mathds{1}_{\hat{\tau}<\tau^{x,D}}  \mathbb{E}_{\hat{\tau}}\bigg[\int_{\hat{\tau}}^{\tau^{x,D}} e^{-\rho t} dD_t \bigg]\bigg] +  \xi_0 \\
    			&\le \epsilon + \mathbb{E}[  \mathds{1}_{\hat{\tau}<\tau^{x,D}} {\cal V}(y)] + \xi_0 = \epsilon + {\cal V}(y)\mathbb{P}(\hat{\tau}<\tau^{x,D})+\xi_0 \le \xi_0 + \big(1+{\cal V}(y)\big)\epsilon.
    		\end{align*}
    		Due to the arbitrariness of $\epsilon$, $\inf_{\theta\in \Theta^D}{\cal V}_0^{0+,D,\theta}\le {\cal V}_0^{0+,D,\theta^0} \le  \xi_0$, thus by arbitrariness of $D$, ${\cal V}(0+)=\sup_{D\in{\cal A}^{0+}}\inf_{\theta\in \Theta^D}{\cal V}_0^{0+,D,\theta} \le  \xi_0$; together with  ${\cal V}(0+) \ge \xi_0$ by \eqref{eq:est_V}, we conclude the validity of \eqref{eq:boundary_con_B}.

            Lastly, we prove that ${\cal V}$ is continuous at any $x>0$.
      For arbitrary $\theta\in\Theta$, let the measure $\mathbb{Q}^{\theta}$ be defined in \eqref{eq:measureQ} and Brownian motion $W^{\theta}$ under $\mathbb{Q}^{\theta}$. Then the uncontrolled surplus $X^x$ in \eqref{eq:uncontrolX} satisfies
      \begin{align*}
          dX_t^x = \big(\mu(X_t^x) + \sigma(X_t^x)\theta(X_t^x) \big)dt + \sigma(X_t^x) dW^{\theta}_t,\quad X_0 = x.
      \end{align*}
      For arbitrary $x>0$, define a stopping time $\zeta^{\mathbb{Q}}_{0,y|x}:=\inf\{t\ge0: X_t^x \notin (0,y)\}$ under $\mathbb{Q}$. It is well-known from the classical theory of diffusion that
    		\begin{align}\label{pf:prop_v_cts_1}
    			\mathbb{Q}(X^x_{\zeta^{\mathbb{Q}}_{0,y|x}} = y) = \frac{\mathfrak{s}^{\mathbb{Q}}(x) -  \mathfrak{s}^{\mathbb{Q}}(0)}{ \mathfrak{s}^{\mathbb{Q}}(y) -  \mathfrak{s}^{\mathbb{Q}}(0) },
    		\end{align}
    		where $ \mathfrak{s}^{\mathbb{Q}}$ is the {\it scale function} of $X_t^x$ under $\mathbb{Q}$, that is, 
      $$\mathfrak{s}^{\mathbb{Q}}(x) : = \int_{0}^{x}\exp\bigg(-2\int_0^z\frac{\mu(y)+\sigma(y)\theta(y)}{\sigma^2(y)}dy\bigg)dz.$$
     \eqref{pf:prop_v_cts_1} implies that for arbitrary $\epsilon>0$, there exists $y({\epsilon}) > x$ such that 
     $$\mathbb{Q}(X^x_{\zeta^{\mathbb{Q}}_{0,y|x}} = y) \ge 1-\epsilon\quad \mbox{for}\quad x< y<y(\epsilon).$$ Then, following similar arguments for proving \eqref{eq:lem_bound_0} in Lemma \ref{lem:X_bound}, we have for any $\epsilon>0$, there exists $y'(\epsilon)>x$ such that $\mathbb{Q}(\zeta^{\mathbb{Q}}_{y|x}<\zeta^{\mathbb{Q}}_{0|x}\wedge \epsilon) \ge 1-\epsilon$ for $x<y<y'(\epsilon)$. 
     For arbitrary $D\in \mathcal{A}^y$ with $x<y<y'(\epsilon)$, we construct $D\in {\cal A}^x$ as follows: $D'_t = 0$ when $t< \zeta^{\mathbb{Q}}_{0,y|x}\wedge\epsilon$ and for $t\ge \zeta^{\mathbb{Q}}_{0,y|x}\wedge\epsilon$, $D'_t = X_{\epsilon}^x$  on $\{\epsilon< \zeta^{\mathbb{Q}}_{0,y|x}\}$,  $D'_t = D_{t-\zeta^{\mathbb{Q}}_{y|x}}$ on $\{\zeta^{\mathbb{Q}}_{y|x}<\epsilon\wedge\zeta^{\mathbb{Q}}_{0|x}\}$, $D_t \equiv 0$ otherwise. 
     
     It follows from \eqref{eq:thm:rob_rec:0} in Section \ref{proof:pro:Maenhoutdef} (noting that ${\cal Y}^\theta={\cal V}^{x,\hat D,\theta}$ under the setting $(x, D',\theta)$ therein) that for $x<y<y'(\epsilon)$,
    		\begin{align*}
    			{\cal V}_0^{x,D',\theta} &= \mathbb{E}^{\theta} \bigg[ \int_{0}^{\tau^{x,D'}} e^{\int_{0}^{s} \frac{\theta^2_u}{2R}du -\rho s} dD'_s + e^{\int_{0}^{\tau^{x,D'}} \frac{\theta^2_u}{2R}du}\xi_{\tau^{x,D'}}\bigg] \\
        &\ge \mathbb{E}^{\theta} \bigg[ \mathds{1}_{\zeta^{\mathbb{Q}}_{y|x}<\zeta^{\mathbb{Q}}_{0|x}}\bigg(\int_{0}^{\zeta^{\mathbb{Q}}_{y|x}} e^{\int_{0}^{s} \frac{\theta^2_u}{2R}du -\rho s} dD_s \\  &\qquad\qquad\qquad\qquad\quad 
        +\int_{\zeta^{\mathbb{Q}}_{y|x}}^{\zeta^{\mathbb{Q}}_{0|y}}   e^{\int_{0}^{s} \frac{\theta^2_u}{2R}du -\rho s} dD_s+ e^{\int_{0}^{\zeta^{\mathbb{Q}}_{y|x}} \frac{\theta^2_u}{2R}du}\xi_{\tau^{x,D'}}\bigg)\bigg]\\
        &\ge \mathbb{E}^{\theta} \bigg[ \mathds{1}_{\zeta^{\mathbb{Q}}_{y|x}<\zeta^{\mathbb{Q}}_{0|x}}\mathbb{E}^{\theta}_{\zeta^{\mathbb{Q}}_{y|x}}\bigg[ \int_{\zeta^{\mathbb{Q}}_{y|x}}^{\zeta^{\mathbb{Q}}_{0|y}}   e^{\int_{0}^{s} \frac{\theta^2_u}{2R}du -\rho s} dD_s + e^{\int_{0}^{\zeta^{\mathbb{Q}}_{y|x}} \frac{\theta^2_u}{2R}du}\xi_{\tau^{x,D'}}\bigg] \bigg]\\
       & \ge \mathbb{Q}(\zeta^{\mathbb{Q}}_{y|x}<\zeta^{\mathbb{Q}}_{0|x}\wedge \epsilon)e^{-\rho\epsilon}{\cal V}^{y,D,\theta}_0 \ge (1-\epsilon)e^{-\rho\epsilon}{\cal V}^{y,D,\theta}_0.
    		\end{align*}
      Hence,
      \begin{align*}
          &\inf_{\theta\in \Theta^D} {\cal V}^{y,D,\theta}_0- {\cal V}(x) \le  \inf_{\theta\in \Theta^D} {\cal V}^{y,D,\theta}_0 - \inf_{\theta\in \Theta^{D'}} {\cal V}^{x,D',\theta}_0 \\
         &\quad \le \inf_{\theta\in \Theta^D} {\cal V}^{y,D,\theta}_0  -(1-\epsilon)e^{-\rho\epsilon}{\cal V}^{y,D,\theta}_0 \le \big(1-(1-\epsilon)e^{-\rho \epsilon}\big) {\cal V}(h),
      \end{align*}
      for some $h>y$ sufficiently large. Therefore ${\cal V}(\cdot)$ is continuous on $(0,+\infty)$. 
      %Therefore for $x<y<y'(\epsilon)$, $J^*(y) - J^*(x) \le (1-(1-\epsilon)e^{-\rho \epsilon}) J^*(h)$, this shows the continuity of $J^*(x)$ at $x>0$.  
    \end{proof}

    \begin{rem}\label{rem:as:bound}
        In dividend optimization problems under model ambiguity, the regularity of the corresponding value function w.r.t. the state process is typically established only under specific models, such as constant-coefficient Brownian motion (Chakraborty et al. \cite{chakraborty2023optimal}).
    %\mycomment{\cite{Cadenillas2007MF} considers model ambiguity?} 
    Related works without ambiguity, such as Choulli et al.~\cite[Lemma 1]{choulli2003diffusion} and Alvarez and Virtanen~\cite[Lemma 3.1]{alvarez2006class}, assume boundedness of the coefficients of the It\^o process to guarantee such regularity. The condition \ref{as:bound} aligns with such conditions, and it enables to extend the regularity results to general processes even under model ambiguity.
    \end{rem}
Finally, we provide the proof of Theorem~\ref{thm:verification}.
\begin{proof} Theorem~\ref{thm:verification} is proved in two steps. As the proof for the case $x=0$ is obvious, we assume $x\in{\mathbb{R}_+}$. Recall that  ${\cal V}$ is defined in \eqref{eq:obj_maxmin}.  %For simplicity, we omit the superscript $x$ in $X^{x,D}$ and $\tau^{x,D}$, as defined in \eqref{eq:controlX} and \eqref{def:bankrupcytime}.

    \noindent{\bf Step 1.} We claim that for every $x\in{\mathbb{R}_+}$,
    \begin{align}\label{eq:claim1}
    v^*(x)\geq& \sup_{D\in {\cal A}^x}\inf_{\theta \in \Theta^D}\mathbb{E}^{\theta}\bigg[\int_{0}^{\tau^{x,D}}e^{-\rho t} \Big( dD_t +  \frac{\theta_t^2 }{2\cR} v^*(X_t^{x,D})  dt\Big) + \xi_{\tau^{x,D}} \bigg],\\
        v^*(x) \ge& {\cal V}(x).\label{eq:claim1-2}
    \end{align}
    Set an arbitrary $D\in {\cal A}^x$ and temporally define the stopping times $(\hat{\tau}_n)_{n\in \mathbb{N}}$ by 
    \begin{align}\label{eq:stop_bdd_X}
        \hat{\tau}_n:= \inf\{t\ge 0: X^{x,D}_t\notin [1/n,n] \}\wedge n.
    \end{align}
    Recall $\mathcal{L}$ given in \eqref{eq:nonlinear_hamilton}, and define the measure $\mathbb{Q}^*$ via \eqref{eq:measureQ} with the Girsanov kernel $\theta^*:=(\theta^*(D)_s)_{s\geq 0}$ given by 
    \begin{align}
        \theta^*_s:=%&\argmin_{\theta\in\mathbb{R}}\Big\{\frac{\sigma^2(X_s^{D})}{2}(v^*)''(X_s^{D})+\left(\mu(X_s^{D})+\sigma(X_s^{D}) \theta\right)(v^*)'(X_s^{D})+ \frac{\theta^2}{2R} v^*(X_s^{D}) -\rho v^*(X_s^{D})  \Big\}\nonumber \\=& 
        -R \frac{\sigma(X_s^{x,D})(v^*)'(X_s^{x,D})}{v^*(X_s^{x,D})}. \label{eq:worst_kernel}
    \end{align}
    Note that by \eqref{eq:const_v} and Proposition \ref{pro:shooting_sub}, $(v^*)'$ is bounded on $[0,b^*]$ and is identical to $1$ on $(b^*,\infty)$ while $v^*$ is bounded below by $\xi_0$ on ${\mathbb{R}_+}$. Following similar arguments of in Theorem~\ref{thm:rob_rec}, it follows that $\theta^*\in \Theta^D$.
    From this, we can denote by $\mathbb{Q}^*$ the corresponding probability measure obtained from $\theta^*$.
    
    %Let $T>0$ and 

    An application of Dynkin formula into the process $e^{-\rho t} v^*(X_t^{x,D})$ on the random time interval $[0,\tau^{x,D}\wedge\hat{\tau}_n]$ under $\mathbb{Q}^*$ ensures that 
    \begin{align}
    %\begin{aligned}
    &e^{-\rho (\tau^{x,D}\wedge \hat{\tau}_n )}  v^*(X_{\tau^{x,D}\wedge \hat{\tau}_n }^{x,D})=\int_0^{\tau^{x,D}\wedge \hat{\tau}_n } e^{-\rho s}\big(\hat{\cal L}^{\theta^*}v^*(X_s^{x,D})ds-(v^*)'(X_{s}^{x,D})  dD_s\big)\nonumber\\
    &\quad +\int_0^{\tau^{x,D}\wedge \hat{\tau}_n } e^{-\rho s} \sigma(X_{s}^{x,D})(v^*)'(X_{s}^{x,D}) d W^{\mathbb{Q}^*}_s+ v^*(x),\label{eq:thm_veri_pf_1}
    %\end{aligned}
    \end{align}
    where $\hat{\cal L}^{\theta^*}v^*(X_s^{x,D}):=({\sigma^2}(v^*)''/2+ (\mu+\sigma\theta_s^*)(v^*)' -\rho v^*)(X_{s}^{x,D}).$
    From the definitions of \eqref{eq:nonlinear_hamilton}, $\theta^*$ in \eqref{eq:worst_kernel}, and given that $v^*(X_s^{x,D})>0$ for $s\in[0,\tau^{x,D})$ and $v^*$ solves \eqref{eq:fbdy}, it follows that for every $s\in[0,\tau^{x,D})$:
    \begin{align*}%\label{eq:nonlinear_ineq}
        \hat{\cal L}^{\theta^*}v^*(X_s^{x,D})=  {\cal L} v^*(X_s^{x,D})-\frac{(\theta_s^*)^2}{2R}v^*(X_s^{x,D})\leq -\frac{(\theta_s^*)^2}{2R}v^*(X_s^{x,D}).
    \end{align*}
    Furthermore, Since $v^*\in C^2({\mathbb{R}_+})$ (see Proposition \ref{pro:shooting}), $\sigma\in C^1(\mathbb{R}_+)$ (see Assumption \ref{as:bdy_cond}\;i.) and $X^{x,D}$ is bounded in $[0,\tau^{x,D}\wedge \hat{\tau}_n ]$ (see \eqref{eq:stop_bdd_X}), the Brownian local martingale term given in~\eqref{eq:thm_veri_pf_1} is martingales.
    % , i.e.,
    % \begin{align}\label{eq:lcl_mtg}
    % &\mathbb{E}^{\theta^*}\left[\int_0^{\tau^D\wedge \hat{\tau}_n } e^{-\rho s} \sigma(X_{s}^D)(v^*)'(X_{s}^D) d W^{\mathbb{Q}^*}_s\right]= 0\\
    % &\mathbb{E}^{\theta^*}\left[\int_0^{\tau^D\wedge \hat{\tau}_n } e^{-\rho s+\frac{1}{2R}\int_0^{s} (\theta_u^*)^2 du}\sigma(X_{s}^D)(v^*)'(X_{s}^D) d W^{\mathbb{Q}^*}_s\right]=0, 
    % \end{align}
    %Furthermore, since $(v^*)'(x)\geq 1$ for every $x\in{\mathbb{R}_+}$, %(noting that $v^*$ is the solution of \eqref{eq:fbdy} and $v^*\in C^2({\mathbb{R}_+})$) 
    %and $X_{s-}^D-X_s^D=D_s - D_{s-}$ for every $s\geq 0$,the following hold:
    % \begin{align}\label{eq:singular_deriv}
    %     \begin{aligned}
    %     &\int_0^{\tau_D\wedge \hat{\tau}_n }e^{-\rho s} (v^*)'(X_{s}^D)  dD^{C}_s +\sum_{0\le s \le \tau_D\wedge \hat{\tau}_n } e^{-\rho s}(v^*(X_{s-}^D) - v^*(X_{s}^D))\\
    %      &\geq \int_0^{\tau_D\wedge \hat{\tau}_n }e^{-\rho s} dD^{C}_s+ \sum_{0\le s \le \tau_D\wedge \hat{\tau}_n } e^{-\rho s} (D_s - D_{s-}) = \int_0^{\tau_D\wedge \hat{\tau}_n }e^{-\rho s} dD_s.
    %     \end{aligned}
    % \end{align}
    Lastly, since $(v^*)'(x)\geq 1$ for $x\in {\mathbb{R}_+}$ and $v^*(0)=\xi_0$, it follows that $v^*(X_{\tau^{x,D}\wedge \hat{\tau}_n }^{x,D})\geq \xi_0$. %Combining this with \eqref{eq:nonlinear_ineq}-\eqref{eq:singular_deriv}, we 
    Taking the expectation over \eqref{eq:thm_veri_pf_1} under $\mathbb{Q}^*$ gives
    % \begin{align*}
    % v^*(x) \geq \mathbb{E}^{\theta^*}_x\left[e^{-\rho (\tau_D\wedge \hat{\tau}_n \wedge T)}  v^*(X_{\tau_D\wedge \hat{\tau}_n \wedge T}^D)+ \int_0^{(\tau_D\wedge \hat{\tau}_n \wedge T)} e^{-\rho s}\left(\frac{(\theta^*_s)^2}{2R}v^*(X^D_s) ds + d D_s \right)  \right]. 
    % \end{align*}
    % Furthermore, since $(v^*)'\geq 1$ on ${\mathbb{R}_+}$ and $v^*(0)=\xi_0$, 
    \begin{align*}
    v^*(x) \geq & \mathbb{E}^{\theta^*}\bigg[\int_0^{\tau^{x,D}\wedge \hat{\tau}_n } e^{-\rho s}\bigg(d D_s +\frac{(\theta^*_s)^2}{2R}v^*(X^{x,D}_s) ds\bigg) +e^{-\rho (\tau^{x,D}\wedge \hat{\tau}_n )}  \xi_0 \bigg],
    \end{align*}
    where we have used Assumption \ref{as:default} (i.e., $e^{-\rho (\tau^{x,D}\wedge \hat{\tau}_n )}  \xi_0=\xi_{\tau^{x,D}\wedge \hat \tau_n}$). 
    
    Since $\tau^{x,D}\wedge \hat{\tau}_n  \rightarrow \tau^{x,D}$ as $n\rightarrow \infty$ (see \eqref{eq:stop_bdd_X}), an application of Fatou's lemma %(noting that $\xi_0\geq0$; $v^*(x)\geq \xi_0$ for every $x\in{\mathbb{R}_+}$; and $dD_s\geq0$ for every $s\geq 0$) 
    together with $\theta^* \in \Theta^D$ (see~\eqref{eq:worst_kernel}) ensures that 
    \begin{align*}%\label{eq:proof_thm_veri}
    v^*(x) \geq & \mathbb{E}^{\theta^*}\bigg[\int_0^{\tau^{x,D}} e^{-\rho s}\bigg(d D_s+\frac{(\theta_s^*)^2}{2R}v^*(X^{x,D}_s) ds\bigg)+e^{-\rho \tau^{x,D}}  \xi_0  \bigg] \\
    \geq & \inf_{\theta \in \Theta^D} \mathbb{E}^{\theta}\bigg[\int_0^{\tau^{x,D}} e^{-\rho s}\bigg(d D_s+\frac{(\theta_s^*)^2}{2R}v^*(X^{x,D}_s) ds\bigg)+e^{-\rho \tau^{x,D}}  \xi_0  \bigg]. \nonumber
    \end{align*}
    As the inequality holds for every $D\in {\cal A}^x$, the claim \eqref{eq:claim1} holds.
    
    To show \eqref{eq:claim1-2}, the proof proceeds similarly by applying Dynkin formula to the process $e^{-\rho t + \int_0^{t}\frac{({\theta}^*_s)^2}{2R} ds} v^*(X_t^{x,D})$ on the random time interval $[0,\tau^{x,D}\wedge\hat{\tau}_n]$ and take expectation. 
    %With similar arguments applied to \eqref{eq:thm_veri_pf_1.5}, 
    Then we obtain that 
    \begin{align*}
    v^*(x) \geq \mathbb{E}^{\theta^*}\bigg[\int_0^{\tau^{x,D}} e^{-\rho s+\int_0^{s} \frac{(\theta_u^*)^2}{2R} du}d D_s+e^{-\rho \tau^{x,D}+\int_0^{\tau^{x,D}} \frac{(\theta_u^*)^2 }{2R}du}  \xi_0  \bigg]= {\cal V}^{x,D,\theta^*}_0,
    \end{align*}
    where the equality holds by \eqref{eq:thm:rob_rec:0} and Assumption \ref{as:default}. %(i.e., $e^{-\rho \tau^{x,D}}  \xi_0=\xi_{\tau^{x,D}}$).
    %$J^{\theta}$ is defined in \eqref{eq:object_J}. 
    The last term is greater than $\inf_{\theta\in \Theta^D} {\cal V}^{x,D,\theta}$. Taking supreme over ${D\in\mathcal{A}^x}$ on both sides of the above equation gives \eqref{eq:claim1-2}.
    
    %\vspace{0.5em}    
    \noindent {\bf Step 2.} Let $D^*=D(b^*)$ be the $b^*$-threshold strategy. We show that 
    \begin{align}\label{eq:claim2}
    v^*(x)\leq& \inf_{\theta \in \Theta^{D^*}}\mathbb{E}^{\theta}\bigg[\int_{0}^{\tau^{x,D^*}}e^{-\rho t} \Big(dD^*_t +  \frac{\theta_t^2}{2R}  v^*(X_t^{x,D^*}) dt \Big) + \xi_{\tau^{x,D^*}} \bigg],\\
    \label{eq:claim2-2}
        v^*(x) \le& {\cal V}(x).
    \end{align}
   We split the proof into two cases: $x\in(0,b^*]$ and $x\in(b^*,\infty)$. By the definition of $b^*$-threshold strategy (see Definition \ref{def:thres_divid}), for every $s\in(0,\tau^{x,D^*}]$, we have %Let $\theta\in \Theta(D^*)$ and denote by corresponding probability measure $\mathbb{Q}$. Then, 
    \begin{align}\label{eq:skorokhod_reflect}
        X^{x,D^{*}}_s \in [0, b^*],\quad dD^*_s>0\quad \mbox{only when $X_s^{x,D^*}=b^*$}.
    \end{align}
    Furthermore, if $x\in(0,b^*]$, the above properties hold for every $s\in[0,\tau^{x,D^*}]$. 
    %(see Definition~\ref{def:thres_divid})%, the {\it Skorokhod minimality condition} holds:
    % \begin{align}\label{eq:thm_veri_pf_2}
    %     \int_{0}^{\infty} \mathds{1}_{\{X^*_{t-} < \hat{x}_a \}}dD_t^* = \sum_{0\le t <\infty} \int_{0}^{D^*_{t} - D^*_{t-}}\mathds{1}_{\{X^*_{t-} - z <  \hat{x}_a\}} dz = 0.
    % \end{align}
    {\it Case 1.} $x\in(0,b^*]$. Let $\theta \in \Theta^{D^*}$, and $\mathbb{Q}^{\theta}$ be the corresponding probability measure induced by the kernel $\theta$. Furthermore, as in \eqref{eq:stop_bdd_X}, 
    we define for every $n\in \mathbb{N}$ by $\hat{\tau}_n:= \inf\{t\ge 0: X^{x,D^*}_t\notin [1/n,n] \}\wedge\tau^{x,D^*}.$
    %We note that since $x\in(0,b^*]$ and $D^*$ is a $b^*$-threshold strategy, $D^*$ has no jumps (i.e., $(D_s^*)^C=D_s^*$ for every $s\geq 0)$. Hence, 
    An application of Dynkin formula to the process $e^{-\rho t} v^*(X_t^{x,D^*})$ on the random time interval $[0,\hat{\tau}_n]$ under $\mathbb{Q}^{\theta}$ gives that \begin{align}\label{eq:thm_veri_pf_10}
    \begin{aligned}
    v^*(x)
    &=\mathbb{E}^{\theta}\bigg[e^{-\rho \hat{\tau}_n}  v^*(X_{\hat{\tau}_n}^{x,D^*})-\int_0^{\hat{\tau}_n} e^{-\rho s} (\hat{\cal L}^{\theta}v^*(X_s^{x,D^*})ds-dD^*_s ) \bigg],
    \end{aligned}
    \end{align}
    where $\hat{\cal L}^{\theta}v^*(X_s^{x,D^*}):=(\sigma^2(v^*)''/2 + (\mu+\sigma\theta_s)(v^*)' -\rho v^*)(X_{s}^{x,D^*})$
    % where $\hat{\cal L}^{\theta}$ is the infinitesimal operator under $\mathbb{Q}^{\theta}$, i.e.,
    % \begin{align*}
    %     \hat{\cal L}^{\theta}v^*(X_s^{x,D^*}):=\Big(\frac{1}{2}\sigma^2(v^*)'' + (\mu+\sigma\theta_s)(v^*)' -\rho v^*\Big)(X_{s}^{x,D^*}),
    % \end{align*}
    and we use the fact that $\mathbb{E}^{\theta}[\int_0^{\tau^{x,D^*}\wedge \hat{\tau}_n } e^{-\rho s} \sigma(X_{s}^{x,D^*})(v^*)'(X_{s}^{x,D^*}) d W^{\mathbb{Q}}_s]=0$ and that 
    $dD^*_s>0$ only when $(v^*)'(X_{s}^{x,D^*})=1$ for $s\in[0,\tau^{x,D^*}]$.
    %$\int_0^{\tau^{D^*}\wedge \hat{\tau}_n}e^{-\rho s}(v^*)'(X_{s}^D) dD^*_s=\int_0^{\tau^{D^*}\wedge \hat{\tau}_n}e^{-\rho s}dD^*_s$ (from \eqref{eq:skorokhod_reflect} and the fact that for every $s\in[0,\tau^{D^*}]$, $(v^*)'(X_s^{D^*})=1$ if $dD_s^*>0$; $(v^*)'(X_s^{D^*})\geq 1$ else). %is the solution of \eqref{eq:fbdy} with the free boundary $b^*$).
    
    From \eqref{eq:skorokhod_reflect} and \eqref{eq:fbdy}, it follows that ${\cal L}v^*(X_s^{x,D^*})=0$ for $s\in [0,\tau^{x,D^*}]$, where  ${\cal L}$  is defined in \eqref{eq:nonlinear_hamilton}. Thus, for $s\in [0,\tau^{x,D^*}]$,
    %From \eqref{eq:skorokhod_reflect} and the fact that ${\cal L}v^*(x)=0$ for all $x\in(0,b^*]$ (see \eqref{eq:fbdy} with the free boundary $b^*$), it follows that for all $s\in [0,\tau^{x,D^*}]$, ${\cal L}v^*(X_s^{x,D^*})=0$. Recall the definition of ${\cal L}$ in \eqref{eq:nonlinear_hamilton}, we obtain that for all $s\in [0,\tau^{x,D^*}]$ 
    \begin{align}\label{eq:claim3}
        \hat{\cal L}^{\theta}v^*(X_s^{x,D^*})\geq {\cal L}v^*(X_s^{x,D^*})- \frac{(\theta_s)^2}{2R} v^*(X_s^{x,D^*}) = - \frac{(\theta_s)^2}{2R} v^*(X_s^{x,D^*}).
    \end{align}
    By \eqref{eq:claim3} and \eqref{eq:thm_veri_pf_10}, we hence obtain that  \begin{align*}%\label{eq:thm_veri_pf_2}
    v^*(x)\leq \mathbb{E}^\theta \bigg[\int_0^{\hat{\tau}_n} e^{-\rho s} \bigg(dD^*_s +\frac{(\theta_s)^2}{2R} v^*(X_s^{x,D^*})ds\bigg)+e^{-\rho \hat{\tau}_n}  v^*(X_{\hat{\tau}_n}^{x,D^*}) \bigg].
    \end{align*}
    The claim \eqref{eq:claim2} holds for all $x\in(0,b^*]$ by passing the limit of $n\rightarrow \infty$ in the above equation and using the fact that $v^*(X_{\tau^{x,D^*}}^{x,D^*})=v^*(0)=\xi_0$ (see \eqref{eq:fbdy}).
   % By passing the limit of $n\rightarrow \infty$ in the above equation, we have 
    %From $v^*\in C^2({\mathbb{R}_+})$ and the fact that $X^{x,D^{*}}_s \in [0, b^*]$ for all $s\in [0,\tau^{x,D^*}]$ (see \eqref{eq:skorokhod_reflect}), it follows that 
    % $
    % \sup_{s\in [0,\tau^{x,D^*}]}|v^*(X_s^{x,D^*})|<\infty.
    % $
    % Hence, an application of the DCT (applicable to the last term of right-hand side of \eqref{eq:thm_veri_pf_2}) and MCT (applicable to the first term of right-hand side of \eqref{eq:thm_veri_pf_2}), together with $\hat{\tau}_n\rightarrow \tau^{x,D^*}$ as $n\rightarrow \infty$ into \eqref{eq:thm_veri_pf_2} ensures that 
    % \begin{align}\label{eq:thm_veri_pf_3}
    % v^*(x)
    %&\leq \mathbb{E}^\theta \left[\int_0^{\tau^{D^*}} e^{-\rho s} \left(\frac{(\theta_s)^2}{2R} v^*(X_s^{D^*})ds+dD^*_s \right) +e^{-\rho \tau^{D^*}}  v^*(X_{\tau^{D^*}}^{D^*})\right]\\
    % &\leq \mathbb{E}^\theta \bigg[\int_0^{\tau^{x,D^*}} e^{-\rho s} \left(dD^*_s +\frac{(\theta_s)^2}{2R} v^*(X_s^{x,D^*})ds\right)+e^{-\rho \tau^{x,D^*}}  \xi_0 \bigg],
    % \end{align}
    % where we have used that $v^*(X_{\tau^{x,D^*}}^{x,D^*})=v^*(0)=\xi_0$ (see \eqref{eq:fbdy}). As the inequality holds for all $\theta \in \Theta^{D^*}$, the claim \eqref{eq:claim2} holds for all $x\in(0,b^*]$.

     \noindent{\it Case 2.} $x\in (b^*,\infty)$. We have $v^*(x)=x-b^*+\psi^+(b^*)$ by \eqref{eq:const_v}. 
     
     As $D^*(=D(b^*))$ starts with an instantaneous increase with $x-b^*$, we have
    % It remains to prove the case $x\in (b^*,\infty)$. Note that for every $x\in (b^*,\infty)$, $v^*(x)=x-b^*+\psi^+(b^*)$ (see \eqref{eq:const_v}). Furthermore, since $D^*(=D(b^*))$ starts with an instantaneous increase with amount $x-b^*$, it follows that 
    \begin{align*}
    &\mathbb{E}^\theta \bigg[\int_0^{\tau^{x,D^*}} e^{-\rho s} \bigg(dD^*_s+\frac{(\theta_s)^2}{2R} v^*(X_s^{x,D^{*}})ds \bigg)+e^{-\rho \tau^{x,D^*}}  \xi_0 \bigg]\geq v^*(x)
  %  &=x-b^*+\mathbb{E}^\theta \bigg[\int_0^{\tau^{D^*}} e^{-\rho s} \Big(\frac{(\theta_s)^2}{2R} v^*(X_s^{b^*,D^{*}})ds+dD^*_s \Big)+e^{-\rho \tau^{D^*}}  \xi_0 \bigg]\geq x-b^* +v^*(b^*) = v^*(x)
    \end{align*}
    holds for every $\theta \in \Theta^{D^*}$, where we have used the inequality \eqref{eq:claim2} with $x=b^*$ (proved in Case 1) and $v^*(b^*)=\psi^+(b^*)$.
    %\comment{I will further fix some notations...} 
    % which ensures  
    % \[
    % \mathbb{E}^\theta \bigg[\int_0^{\tau^{D^*}} e^{-\rho s} \Big(\frac{(\theta_s)^2}{2R} v^*(X_s^{x,D^{*}})ds+dD^*_s \Big)+e^{-\rho \tau^{D^*}}  \xi_0\bigg]\geq  x-b^* +v^*(b^*) = v^*(x),
    % \]
    % where we have used that $v^*(b^*)=\psi^+(b^*)$ with $v^*\in C^2({\mathbb{R}_+})$. 
    %As the inequality holds for every $\theta \in \Theta_{D^*}$, the claim \eqref{eq:claim2} holds for every $x\in(b^*,\infty)$. 

    To show \eqref{eq:claim2-2}, similar to the analysis in {Step 1}, an application of Dynkin formula to the process $e^{-\rho t \int_0^{t}({\theta}_s)^2/(2 R) ds} v^*(X_t^{x,D^*})$ gives
    \begin{align*}
        v^*(x) \le \inf_{\theta\in{\Theta^{{D}^*}}} {\cal V}_0^{x,D^*,\theta} \le \sup_{D\in\mathcal{A}^x} \inf_{\theta\in{\Theta^{D}}} {\cal V}_0^{x,D,\theta} = {\cal V}(x),\quad x\in{\mathbb{R}_+}.
    \end{align*}
    This completes the proof.  
\end{proof}

    %%%%%%%%%%%%%%%%%%%%%%%%%%%%%%%%%%%%%%%%%%%%%%%%%%%%%%%%%%%%%%%%%%%%%%%%%%%%%%%%%%%%%%%%%%%%%%%%%%%%%%%%%%%%%%%%%%%%%%%%%%%%%%%%%%%%%%%%%%%%%%
    \subsection{Proof of results in Section \ref{sec:5}}\label{sec:10}
    We provide some observations used for the proof of Theorem \ref{thm:cts_VEZ}. 
   % Throughout, Assumption \ref{as:recurrent} is imposed without further mention.
    \begin{lem}\label{lem:10.1} 
    For any $(x,D)\in \mathbb{R}_+\times {\cal A}^x$, we assume that $\xi_{\tau^{x,D}}>C$ $\mathbb{P}$-a.s. with some constant $C>0$. Then the following hold: for every~${\cR}\in(0,1)$, %and $(x,D)\in \mathbb{R}_+\times {\cal A}^x$,
    %\mycomment{part i corresponds to Lemma \ref{lem:BSDE_exist_g_tau}. part ii corresponds to Lemma \ref{lem:rob_rec_strict}.}
\begin{enumerate}[label=\roman*.]
% \item $(V^{\EZ,D}(\cR), Z^D(\cR))$ is the unique $L^2$-solution to \eqref{dfn:BSDE1} under $(\tau,D,\xi_\tau({\cal R}))$; $V^{\EZ,D}_t(\cR) \ge C$ $\mathbb{P}$-a.s.\ for $t \ge 0$; and $\mathbb{E}[\sup_{t \ge 0} (V_t^{\EZ,D}(\cR))^{\frac{2}{1-\cR}}] < \infty$.
% \item $V_t^{\rob,D}(\cR) = (V_t^{\EZ,D}(\cR))^{\frac{1}{1-\cR}}$ for $t \ge 0$. In addition, $(V^{\rob,D}(\cR),{\mathbb Z}^D(\cR))$ with ${\mathbb Z}^D(\cR): = (V^{\rob,D}(\cR))^\cR Z^D(\cR)/(1-\cR) $ solves
% \begin{align}\label{eq:Vrob_BSDE}
%     V^{\rob,D}_t(\cR) = \xi^0_{\tau} + \int_{t\wedge\tau}^{\tau} e^{-\rho s} dD_s - \frac{\cR}{2} \int_{t\wedge\tau}^{\tau} \frac{({\mathbb Z}^D_s(\cR))^2}{V^{\rob,D}_s(\cR)} ds -  \int_{0}^{\tau} \mathbb{Z}^D_s(\cR) dW_s.
% \end{align}
\item ${\cal V}^{x,D}(\cR)$ admits the following representation: for $t\in[0,T]$
\begin{align*}
    {\cal V}^{x,D}_t(\cR) = \essinf_{\theta\in\Theta}  \mathbb{E}^{\theta}_t\bigg[ \int_{t\wedge\tau^{x,D}}^{\tau^{x,D}}\bigg( e^{-\rho s} dD_s+ \frac{{\cal V}^{x,D}_s(\cR) \theta_s^2}{2\cR} ds \bigg) + \xi_{\tau^{x,D}}\bigg],
\end{align*}
where $\Theta$ is defined in Section \ref{sec:notation} (which does not depend on $\cR$).
\item There exists some $\mathbb{F}$-adapted, positive process $(\overline{c}'_t)_{t\geq 0}$ (which does not depend on $\cR$) such that $\mathbb{E}[\sup_{t\ge0} \overline{c}'_t]<\infty$ and
\begin{align*}%\label{eq:Vrob_lb}
    \log\big({\cal V}_t^{x,D}(\cR)\big) \ge \frac{\log(K_t^{x,D}) - \cR \overline{c}'_t}{1-\cR},\quad  \text{ $\mathbb{P}$-a.s.}\quad \text{for }t \ge 0.
\end{align*}
\end{enumerate}
\end{lem}
\begin{proof} %Part {i.} follows from Theorem~\ref{thm:recurs_BSDE}~ii., while Part ii. follows from Theorem~\ref{thm:rob_rec}~ii.. 
%Hence, we start with proving Part iii..
%Let $t\ge0$. 
Fix $(x,D)\in \mathbb{R}_+\times {\cal A}^x$ and $\cR\in(0,1)$. It follows from \eqref{dfn:rbs_V} that for all $\theta\in\Theta^D(\cR)$, ${\cal V}^{x,D,\theta}(\cR)\ge {\cal V}^{x,D}(\cR)$, and that for all $t\geq 0$ 
\begin{align*}
    {\cal V}^{x,D}_t(\cR) \geq \essinf_{\theta\in\Theta^D(\cR)}  \mathbb{E}^{\theta}_t\bigg[ \int_{t\wedge\tau^{x,D}}^{\tau^{x,D}}\bigg( e^{-\rho s} dD_s+ \frac{{\cal V}^{x,D}_s(\cR) \theta_s^2}{2\cR} ds \bigg) + \xi_{\tau^{x,D}}\bigg].
\end{align*}
We get the inequality in the $\ge$ direction of part i as $\Theta^D(\cR)\subset\Theta$. 

For the reverse inequality, we % use the BSDE in \eqref{eq:Vrob_BSDE} to have 
note that for every $\theta \in \Theta$
\begin{align*}
    -\frac{\cR}{2} \frac{(\mathbb{Z}^{x,D}(\cR))^2}{{\cal V}^{x,D}(\cR)} = \inf_{\theta } \Big\{ \frac{{\cal V}^{x,D}(\cR)\theta^2}{2\cR} + \mathbb{Z}^{x,D}(\cR)\theta \Big\}\leq \frac{{\cal V}^{x,D}(\cR)\theta^2}{2\cR} + \mathbb{Z}^{x,D}(\cR)\theta.
\end{align*}

Substituting this into \eqref{eq:Vrob_BSDE}, taking conditional expectation under $\mathbb{Q}^{\theta}$ 
% gives
% \begin{align*}
%     V^{\rob}_t(\cR) \leq \mathbb{E}_t^{\theta}\bigg[ \int_{t\wedge\tau}^{\tau}\bigg( e^{-\rho s} dD_s+ \frac{V^{\rob,D}_s(\cR) \theta_s^2}{2\cR} ds \bigg) + \xi^0_{\tau}\bigg]
% \end{align*}
for each $\theta \in \Theta$, and then taking the essential infimum over $\theta\in\Theta$ yield the inequality in the $\le$ direction of part i.
This completes the proof of part i.

We now prove part ii. Set $R\equiv {\cal R}$. By \eqref{eq:EZ_V} and \eqref{dfn:rbs_V_theta}, we have for $t\geq 0$,
\begin{align*}
    &K_t^{x,D} - V^{x,D}_t(\cR)\nonumber\\
    %&\quad = \mathbb{E}_t\bigg[\xi_{\tau} - \xi_{\tau}^{1-\cR} + \int_{t\wedge\tau^{x,D}}^{\tau^{x,D}} e^{-\rho s} \big( 1 - (1-\cR) (V^{x,D}_s(\cR))^{\frac{-\cR}{1-\cR}}\big) dD_s \bigg] \nonumber\\
    &\quad \le \mathbb{E}_t\bigg[\xi_{\tau} - \xi_{\tau}^{1-\cR} + \int_{t\wedge\tau^{x,D}}^{\tau^{x,D}} e^{-\rho s} \cR V^{x,D}_s(\cR) \, dD_s \bigg] \nonumber\\
    &\quad \le \cR \bigg( \mathbb{E}_t[\xi_{\tau}^2] + \mathbb{E}_t\bigg[\sup_{s \ge t} (V_s^{x,D}(\cR))^2\bigg] \mathbb{E}_t\bigg[ \bigg(\int_0^{\tau^{x,D}} e^{-\rho s} dD_s \bigg)^2 \bigg] \bigg).%\label{eq:Vrob_lb_pr1}
\end{align*}
The first inequality uses the concavity of \(f(x) = 1 - (1-\cR)x^{\frac{-\cR}{1-\cR}}\), which implies \(f(x) \le f(1) + f'(1)(x-1) = \cR + \cR(x-1)\). The second inequality follows from \(x - x^{1-\cR} \le \cR x^2\) for \(x > 0\) and \(\cR \in (0,1)\). By assumption, the RHS is finite $\mathbb{P}$-a.s., so there is some ${\cal F}_t$-measurable \(\overline{c}_t\) (not depending on \(\cR\)) such that $K_t^{x,D} - V_t^{x,D}(\cR) \le \cR \overline{c}_t,$ $\mathbb{P}$-a.s.. %\label{eq:lem10.pf1}

Furthermore, \(\overline{c}_t\) can be chosen so that \(\mathbb{E}[\sup_{t\ge0} \overline{c}_t] < \infty\). Indeed, it holds that \(\mathbb{E}[\sup_{t\ge0} |K^{x,D}_t - V^{x,D}_t(\cR)|] < \infty\) (by \(\mathbb{E}[\sup_{t\ge0} |K^{x,D}_{t \wedge \tau^{x,D}}|^2] < \infty\) and \(\mathbb{E}[\sup_{t\ge0} |V^{x,D}_{t \wedge \tau^{x,D}}(\cR)|^2] < \infty\), with $K^{x,D}_t=V^{x,D}_t(\cR)=\xi_{\tau^{x,D}}$ on $t\ge\tau^{x,D}$.)

Using \(V^{x,D}({\cal R}) = ({\cal V}^{x,D}(\cR))^{1-\cR}\) (by Theorem \ref{thm:rob_rec} and $R\equiv {\cal R}$), we have 
\begin{align*}
    &\lambda \log(K_t^{x,D}) + (1-\lambda)(1-\cR) \log\big({\cal V}^{x,D}_t(\cR)\big)\\
    &\quad\le \log\Big( \lambda K^{x,D}_t + (1-\lambda) \big({\cal V}^{x,D}_t(\cR)\big)^{1-\cR} \Big)\quad \mbox{for any \(\lambda \in (0,1)\)}.
\end{align*}
Applying Taylor's expansion and taking the limit as \(\lambda \downarrow 0\), we deduce
\[
\log(K_t^{x,D}) - \log({\cal V}^{x,D}_t(\cR)) \le \frac{\cR \overline{c}_t}{V_t^{x,D}(\cR)} - \cR \log\big({\cal V}^{x,D}_t(\cR)\big),\quad \text{$\mathbb{P}$-a.s..}
\]
Setting \(\overline{c}'_t = \overline{c}_t/C\) (with the lower bound \(C\) of \(V_t^{x,D}(\cR)\)), we have part ii.
 
\end{proof}
%\subsection{Proof of Theorem \ref{thm:cts_VEZ} and Corollary \ref{coro:sec5}}\label{sec:proof:thm:cts_VEZ}

We now provide the proof of Theorem~\ref{thm:cts_VEZ}.
\begin{proof}%{\it of Theorem~\ref{thm:cts_VEZ}.}
%\noindent{\it Proof of Theorem~\ref{thm:cts_VEZ}.}
Fix $(x,D)\in \mathbb{R}_+\times {\cal A}^x$. We first prove the theorem under the assumption that $\xi_{\tau^{x,D}} > C$ $\mathbb{P}$-a.s. with some constant $C > 0$ (as in Lemma~\ref{lem:10.1}). The general case follows by considering the truncated sequence ${\xi}_{\tau^{x,D}}^{(n)} = \frac{1}{n} \vee \xi_{\tau^{x,D}}$ for each $n \in \mathbb{N}$, and then letting $n \to \infty$.
%Throughout the proof, the superscript $D$ is suppressed to simplify notation. 
We present the proof in two steps.

\noindent {\it Step 1. Monotonicity.} Let $0 < \cR_1 < \cR_2 < 1$. For any $\theta \in \Theta$, 
$$
\frac{1}{2\cR_1} \int_0^t \theta_s^2 \, ds \ge \frac{1}{2\cR_2} \int_0^t \theta_s^2 \, ds, \quad \forall t \ge 0.
$$
By the representation of ${\cal V}^{x,D,\theta}$ in~\eqref{eq:thm:rob_rec:0} (see Proposition~\ref{pro:Maenhoutdef}), the inequality ${\cal V}^{x,D,\theta}_t(\cR_1)  \ge {\cal V}^{x,D,\theta}_t(\cR_2),$ $t\ge0,$ holds for all $\theta\in\Theta$ (even if not in the admissible set $\Theta^D(\cR_1)$ of $\cR_1$). 
From this, we have for $t\geq 0$ 
\begin{align*}
    {\cal V}^{x,D}_t(\cR_1) =& \essinf_{\theta\in\Theta^D(\cR_1)}  {\cal V}_t^{x,D,\theta}(\cR_1) \\
\ge& \essinf_{\theta\in\Theta^D(\cR_2)}  {\cal V}_t^{x,D,\theta}(\cR_1) \ge \essinf_{\theta\in\Theta^D(\cR_2)}  {\cal V}_t^{x,D,\theta}(\cR_2) =  {\cal V}_t^{x,D}(\cR_2).
\end{align*}
The first inequality holds because \(\Theta^D(\cR_1) \subseteq \Theta^D(\cR_2)\); as \(\cR\) increases, the integrability conditions in \eqref{eq:ThetaD} become less restrictive.
On the other hand, since $\theta^0 := (\theta^0_t)_{t \ge 0} \equiv 0 \in \Theta^D(\cR)$ for all $\cR\in(0,1)$, it holds that
$$
{\cal V}^{x,D}_t(\cR) = \essinf_{\theta\in \Theta^D(\cR)} {\cal V}_t^{x,D,\theta}(\cR) \le {\cal V}_t^{x,D,\theta^0}(\cR) \equiv K_t^{x,D}, \quad t\ge0,\quad \cR\in(0,1).
$$

\noindent{\it Step 2. Continuity.}
Let \(0 < \cR_1 < \cR_2 < 1\) and let \(\theta^{*,2} \in \Theta^D(\cR_2) \subset \Theta\) denote the worst-case kernel for \(\cR_2\) so that \({\cal V}^{x,D}(\cR_2) = {\cal V}^{x,D,\theta^{*,2}}(\cR_2)\) (see Lemma~\ref{lem:rob_rec_strict}). Specifically,
\begin{align}\label{eq:theta*2}
    \theta^{*,2}_t := - \cR_2 \frac{{\mathbb Z}^{x,D}_t(\cR_2)}{{\cal V}^{x,D}_t(\cR_2)}  = -\frac{\cR_2}{1-\cR_2} \frac{Z^{x,D}_t(\cR_2)}{V^{x,D}_t(\cR_2)}, \quad t\ge0.
\end{align}
Let $\eta^{\theta^{*,2}}$ be defined as in \eqref{eq:eta} with $\theta=\theta^{*,2}$. By Lemma \ref{lem:10.1} i., %we have
\begin{align*}
    {\cal V}^{x,D}_t(\cR_1) \le & \mathbb{E}_t^{\theta^{*,2}}\bigg[ \int_{t\wedge\tau^{x,D}}^{ \tau^{x,D}} e^{-\rho s} dD_s + \xi_{\tau^{x,D}}+ \frac{1}{2\cR_1} \int_{t\wedge\tau^{x,D}}^{\tau^{x,D}} {\cal V}^{x,D}_s(\cR_1) (\theta^{*,2}_s)^2 ds\bigg]\\
     = &{\cal V}^{x,D}_t(\cR_2) +\frac{1}{2}\Big(\frac{1}{\cR_1}-\frac{1}{\cR_2}\Big) \mathbb{E}_t^{\theta^{*,2}}\bigg[ \int_{t\wedge\tau^{x,D}}^{\tau^{x,D}} (\theta^{*,2}_s)^2 {\cal V}^{x,D}_s(\cR_1) ds \bigg]\\
    &+ \mathbb{E}_t^{\theta^{*,2}}\bigg[ \int_{t\wedge\tau^{x,D}}^{\tau^{x,D}} \frac{(\theta^{*,2}_s)^2}{2\cR_2} \big( {\cal V}^{x,D}_s(\cR_1)-{\cal V}^{x,D}_s(\cR_2)\big)ds \bigg].
\end{align*}
Since ${\cal V}_t^{x,D}(\cR_1)\ge {\cal V}_t^{x,D}(\cR_2)$ (by Step 1), Gr\"onwall's inequality gives
\begin{align*}
    &{\cal V}^{x,D}_t(\cR_1)-{\cal V}^{x,D}_t(\cR_2) \\
    &\quad\le \frac{\cR_2-\cR_1}{2\cR_1\cR_2}\mathbb{E}_t^{\theta^{*,2}}\bigg[\int_{t\wedge\tau^{x,D}}^{\tau^{x,D}} e^{\int_{t\wedge\tau^{x,D}}^{s}\frac{(\theta^{*,2}_u)^2}{2\cR_2}du} (\theta^{*,2}_s)^2 {\cal V}^{x,D}_s(\cR_1) ds \bigg]\\
    &\quad = \frac{\cR_2-\cR_1}{2\cR_1\cR_2}\mathbb{E}_t\bigg[\int_{t\wedge\tau^{x,D}}^{\tau^{x,D}} \eta_s^{\theta^{*,2}} e^{\int_{t\wedge\tau^{x,D}}^s\frac{(\theta^{*,2}_u)^2}{2\cR_2}du} (\theta^{*,2}_s)^2 {\cal V}^{x,D}_s(\cR_1) ds \bigg]:=\mathrm{I}_t.
\end{align*}
Using It\^o's formula for \(\log( {\cal V}^{x,D}_s(\cR_2))\), a direct calculation gives 
\begin{align*}
    \eta_s^{\theta^{*,2}} e^{\int_{t}^s\frac{(\theta^{*,2}_u)^2}{2\cR_2}du} &= \exp\bigg\{\cR_2\bigg(\log\Big(\frac{{\cal V}^{x,D}_t(\cR_2))}{{\cal V}^{x,D}_s(\cR_2))}\Big)-\int_{t}^s\frac{e^{-\rho u}}{{\cal V}^{x,D}_u(\cR_2))}dD_u \bigg)\bigg\}\\
    &\le \bigg(\frac{ {\cal V}^{x,D}_t(\cR_2))}{{\cal V}^{x,D}_s(\cR_2))}\bigg)^{\cR_2}.
\end{align*}
% Using the form of $\theta^{*,2}$ in \eqref{eq:theta*2} and the expression of $\eta^{\theta^{*,2}}$ in \eqref{eq:pf_thm3.2_eta} and , we establish the following estimate for the RHS of above equation,
Hence, using the explicit form of \(\theta^{*,2}\) in \eqref{eq:theta*2}, we have
\begin{align*}
    \mathrm{I}_t\le&  \frac{\cR_2-\cR_1}{2\cR_1\cR_2} \mathbb{E}_t\bigg[\int_{t\wedge\tau^{x,D}}^{\tau^{x,D}}\bigg( \frac{\cR_2^2}{(1-\cR_2)^2}\frac{|Z_s^{x,D}(\cR_2)|^2}{V_s^{x,D}(\cR_2) \big({\cal V}^{x,D}_{s}(\cR_2)\big)^{1-\cR_2}} \\
    &\quad\quad\quad\quad\quad\quad\quad\quad\quad \times  \bigg(\frac{ {\cal V}^{x,D}_t(\cR_2))}{{\cal V}^{x,D}_s(\cR_2))}\bigg)^{\cR_2} {\cal V}^{x,D}_s(\cR_1)\bigg) ds \bigg]\\
    \le&\frac{(\cR_2-\cR_1)\cR_2}{2\cR_1(1-\cR_2)^2}  \frac{\big({\cal V}_t^{x,D}(\cR_2)\big)^{\cR_2}}{C^{\frac{1}{1-\cR_2}}} \mathbb{E}_t\bigg[ \sup_{t\ge 0}  \frac{ K^{x,D}_t}{ V^{x,D}_t(\cR_2)} \int_{t\wedge\tau^{x,D}}^{\tau^{x,D}} |Z^{x,D}_s(\cR_2)|^2 ds \bigg],
\end{align*}
% \begin{align*}
%     & \frac{\cR_2-\cR_1}{2\cR_1\cR_2} \mathbb{E}\bigg[\int_0^{\tau}e^{\cR_2\Big(\ln\big(\frac{V^{\rob}_0(\cR_2)}{V^{\rob}_t(\cR_2)}\big)-\int_{0}^t\frac{e^{-\rho s}}{V^{\rob}_s(\cR_2)}dD_s \Big)}  (\theta^{*,2}_t)^2 V^{\rob}_t(\cR_1) dt \bigg]\\
%     &\le \frac{(\cR_2-\cR_1)\cR_2}{2\cR_1(1-\cR_2)^2} \big(V_0^{\rob}(\cR_2)\big)^{\cR_2} \mathbb{E}\bigg[\int_0^{\tau} \frac{|Z^{(\cR_2)}_t|^2}{V^{\EZ}_t(\cR_2)} \frac{ V^{\rob}_t(\cR_1)}{ V^{\rob}_t(\cR_2)} dt \bigg]\\
%     &\le\frac{(\cR_2-\cR_1)\cR_2}{2\cR_1(1-\cR_2)^2}  \frac{\big(V_0^{\rob}(\cR_2)\big)^{\cR_2}}{C} \mathbb{E}\bigg[\sup_{t\ge0}  \frac{ V^{\rob}_t(\cR_1)}{ V^{\rob}_t(\cR_2)} \int_0^{\tau} |Z^{(\cR_2)}_t|^2 dt \bigg]
% \end{align*}
% where $(V^{\EZ}(\cR_2),Z^{(\cR_2)})$ is the unique $L^2$-solution to \eqref{dfn:BSDE1}.
where \(C^{\frac{1}{1-\cR_2}}\) is a lower bound for \({\cal V}_s^{x,D}(\cR_2)\), and we have used the fact that \({\cal V}^{x,D}_t(\cR_1) \le K^{x,D}_t\). To bound the expectation, we use Lemma \ref{lem:10.1} ii to have
\begin{align*}
    \sup_{t\ge0} \frac{ K_t^{x,D}}{ V^{x,D}_t(\cR_2)} \le \sup_{t\ge0} \Big(1+\cR_2 \frac{\overline{c}_t}{C}\Big) =\sup_{t\ge0} (1+\cR_2 \overline{c}_t')<\infty, \text{ ${\mathbb P}$-a.s.,}
    % &= \sup_{t\ge0} \exp\left(\ln(V^{\rob}_t(\cR_1)) - \ln(V^{\rob}_t(\cR_2))\right) \\
    % &\le \sup_{t\ge0} \exp\left(\ln(K_t) - \frac{ \ln(K_t) - \cR_2 \overline{c}'_t }{ 1 - \cR_2 } \right) \le \exp\left( \frac{ \cR_2 \overline{c}_t' }{ 1 - \cR_2 } \right),
\end{align*}
where \(\overline{c}_t'\) is the constant from Lemma \ref{lem:10.1} ii, and we use \(\mathbb{E}[\sup_{t\ge0} \overline{c}_t'] < \infty\).
Together with $\int_0^{\tau^{x,D}} |Z_t^{x,D}(\cR_2)|^2 dt \in L^1({\cal F}_{\tau})$, there exists \(C'_t>0\) such that
\begin{align*}
    0 \le {\cal V}^{x,D}_t(\cR_1) -  {\cal V}^{x,D}_t(\cR_2)
    \le C'_t (\cR_2 - \cR_1), \text{ ${\mathbb P}$-a.s..}
\end{align*}
This establishes the continuity of $(0,1) \ni \cR \mapsto {\cal V}_t^{x,D}(\cR)$. In addition, the limit $\lim_{\cR \downarrow 0} {\cal V}_t^{x,D}(\cR) = K^{x,D}_t$ follows from~Lemma \ref{lem:10.1} ii, so \eqref{eq:Vrob_lim} holds.
 
\end{proof}

We next provide the proof of Corollary~\ref{coro:sec5}.
\begin{proof} Since ${\cal V}_t^{x}({\cal R}) = \esssup_{D\in {\cal A}^x} {\cal V}^{x,D}_t(\cR)$, the monotonicity and lower semi-continuity from left follows directly from Theorem~\ref{thm:cts_VEZ}: 
$$\liminf_{\cR \uparrow \tilde{\cR}} {\cal V}_t^{x}({\cal R}) =  \liminf_{\cR \uparrow \tilde{\cR}} \esssup_{D\in {\cal A}^x}{\cal V}_t^{x,D}({\cal R}) \ge \esssup_{D\in {\cal A}^x}{\cal V}_t^{x,D}(\tilde {\cal R}) = {\cal V}_t^{x}(\tilde{\cR}).$$
For the right continuity, we observe that $\lim_{\cR \downarrow \tilde{\cR}} {\cal V}_t^{x}({\cal R}) =\sup_{\cR \le \tilde{\cR} }{\cal V}_t^{x}({\cal R}) $, and the interchange of supremum arguments yields the desired result. %\mycomment{pls check.}
     
\end{proof}

Lastly, we provide the proof of Theorem~\ref{thm:cts_b*}.
\begin{proof}
We denote by \(\underline{b}_{\cR}\), \(\overline{b}_{\cR}\), and \(\hat{b}_{\cR}\) the constants from Assumption~\ref{as:bdy_cond}, emphasizing their dependence on \(\cR\). We set \(\hat{\boldsymbol{b}} := \max_{\cR \in {\cal I}} \hat{b}_{\cR}\). By Theorem~\ref{thm:verification}, \({\cal V}(x;\cR)=v^*(x;\cR) \) for all \(x \in \mathbb{R}_+\), where \(v^*(x;\cR)\) is constructed in \eqref{eq:const_v} (by rewriting $\psi^+(b^*)$ as $\psi^+(b^*_{{\cal R}};{\cal R})$ for each ${\cal R}$ therein; see also Assumption~\ref{as:bdy_cond}\;ii.). 
Throughout, let \((v^*)'(\cdot;\cR)\) denote the derivative of \(v^*(\cdot;\cR)\) w.r.t. \(x\), and without loss of generality, we fix $\cR_1<\cR_2\in{\cal I}$.

For clarity, we organize the proof into the following steps.

\noindent{\it Step 1. Continuity of ${\cal I} \ni \cR \mapsto v^*(x;\cR)$.} %Consider $\cR_1<\cR_2\in{\cal I}$. 
Let $D_1^*$ and $\theta^{*,2}$ denote the optimal dividend for $\cR_1$ and the worst-case kernel for $\cR_2$ respectively.
From the proof of Theorem~\ref{thm:verification} in Section~\ref{sec:thm:verification}, we~have
\begin{align*}
& v^*(x ; \cR_1) \leq J^{\theta^{*, 2}}(x; D_1^*, v^*(\cdot; \cR_1), \cR_1),\\
&v^*(x,\cR_2) \geqslant J^{\theta^{*, 2}}(x; D_1^*,v^*(\cdot; \cR_2),\cR_2),%\\
\end{align*}
%& \text{where }J^{\theta}(x; D, f(\cdot), \cR):=\mathbb{E}^\theta\bigg[\int_0^{\tau^{x,D}} e^{-\rho s}\Big(d D_s+\frac{\theta_s^2}{2 \cR}f(X_s^{D}) d s\Big)+e^{-\rho\tau^{x,D}} \xi_0\bigg].
%\end{align*}
where for any Borel measurable $f:\mathbb{R}_+\to\mathbb{R}$,
\[
J^{\theta^{*, 2}}(x; D, f, \cR):=\mathbb{E}^{\theta^{*, 2}}\bigg[\int_0^{\tau^{x,D}} e^{-\rho s}\Big(d D_s+\frac{({\theta_s^{*, 2}})^2}{2 \cR}f(X_s^{D}) d s\Big)+e^{-\rho\tau^{x,D}} \xi_0\bigg].
\]
Following similar arguments as in the proof of Theorem~\ref{thm:cts_VEZ} (Step 2), we obtain
\begin{align*}
&0\le    v^*(x;\cR_1) -  v^*(x;\cR_2)\\ &\le\frac{(\cR_2-\cR_1)\cR_2}{2\cR_1}v^*(x;\cR_2)^{\cR_2}\mathbb{E}\bigg[\int_0^{\tau^{x,D_1^*}} e^{-\rho s}\tilde v_{{\cal R}_1,{\cal R}_2}(X_s^{D_1^*}) d s\bigg]\le C_{\cal I}(\cR_2-\cR_1),
 % &0\le    v^*(x;\cR_1) -  v^*(x;\cR_2) \le\frac{(\cR_2-\cR_1)\cR_2}{2\cR_1}v^*(x;\cR_2)^{\cR_2}\\
 % &\cdot \mathbb{E}\bigg[\int_0^{\tau^{x,D_1^*}} e^{-\rho s}\Big( \big((v^*)'(\cR_1)\sigma\big)^2\frac{v^*(\cR_1)}{v^*(\cR_2)} v^*(\cR_2)^{-(1+\cR_2)} \Big)(X_s^{D_1^*}) d s\bigg]\le C_{\cal I}(\cR_2-\cR_1),
\end{align*}
with $\tilde v_{{\cal R}_1,{\cal R}_2}(\cdot):=((v^*)'(\cdot;{\cal R}_1)\sigma(\cdot))^2\frac{v^*(\cdot;{\cal R}_1)}{v^*(\cdot;{\cal R}_2)}(v^*(\cdot;{\cal R}_2))^{-(1-{\cal R}_2)}$
and some constant \(C_{\cal I}>0\) (not depending on \(\cR_1\) and \(\cR_2\)). This uniform bound holds since \(X^{D_{\cR}^*}\) is valued in \([0, \hat{\boldsymbol{b}}]\) and \(\xi_0 + x \leq v^*(x; \cR) \leq \xi_0 + x + \overline{\mu}/\rho\) for all \(\cR \in {\cal I}\).
Therefore, \(v^*(x; \cR)\) is continuous in \(\cR\).

\noindent{\it Step 2. Continuity of ${\cal I} \ni \cR \mapsto (v^*)^{\prime}(0;\cR)$.} 
Note that \(v^*(0;\cR) = \xi_0\), and by Theorem~\ref{thm:cts_VEZ}, the mapping \(\cR \mapsto v^*(x;\cR)\) is decreasing. %Consequently, for any \(\cR_1 < \cR_2\in{\cal I}\), we have $(v^*)'(0;\cR_1) \geq (v^*)'(0;\cR_2)$.
Thus, $(v^*)'(0;\cR_1) \geq (v^*)'(0;\cR_2)$.
Suppose, to the contrary, that $\cR \mapsto (v^*)^{\prime}(0;\cR)$ is not continuous. Then there exists $\tilde{\cR} \in \operatorname{int}{\cal I}$ such that
\begin{align*}
    \alpha_1 := \liminf_{\epsilon \rightarrow 0+} (v^*)^{\prime}(0;\tilde{\cR} - \epsilon) > \limsup_{\epsilon \rightarrow 0+} (v^*)^{\prime}(0;\tilde{\cR} + \epsilon) := \alpha_2 \ge 1,
\end{align*}
where the inequality $\alpha_2 \ge 1$ follows from the fact that $(v^*)^{\prime}(0;\cR) \ge 1$, as it satisfies the HJB-VI \eqref{eq:HJB_VI}. For a sufficiently small $\delta_e>0$,%\mycomment{$\delta_e\to \delta_\epsilon$?}  
we have
\begin{align*}
    (v^*)^{\prime}(0;\tilde{\cR} - \epsilon) > (v^*)^{\prime}(0;\tilde{\cR} + \epsilon) + \frac{\alpha_1 - \alpha_2}{2},\quad \mbox{for all }\epsilon\in(0,\delta_e).
\end{align*}
Since $v^*(x;\cR)$ is twice continuously differentiable w.r.t. $x$ on $\mathbb{R}_+$, the derivative \((v^*)'(x;\cR)\) is continuous in \(x\). Thus, for a sufficiently small \(\delta_x > 0\), 
\begin{align*}
    (v^*)'(x;\tilde{\cR} - \epsilon) > (v^*)'(x;\tilde{\cR} + \epsilon) + \frac{\alpha_1 - \alpha_2}{2},\quad \mbox{for all }x \in [0, \delta_x], \epsilon \in (0, \delta_e).
\end{align*}
Integrating both sides from $x = 0$ to $x = \delta_x$ yields
\begin{align*}
    v^*(\delta_x;\tilde{\cR} - \epsilon) > v^*(\delta_x;\tilde{\cR} + \epsilon) + \frac{\alpha_1 - \alpha_2}{2} \delta_x, \quad \mbox{for all } \epsilon \in (0, \delta_{e}).
\end{align*}
Taking the limit as $\epsilon \rightarrow 0^+$, we obtain $v^*(\delta_x;\tilde{\cR}^{-}) > v^*(\delta_x;\tilde{\cR}^{+})$, which contradicts to the continuity of $\cR \mapsto v^*(x;\cR)$ in Corollary~\ref{coro:sec5}. We conclude that ${\cal I} \ni \cR \mapsto (v^*)^{\prime}(0;\cR)$ must be continuous.

%\medskip
\noindent{\it Step 3. Continuity of ${\cal I}\ni \cR \mapsto g_{b^*_{\cR}}(x)$.}
Recall $g_{b^*_{\cR}}(x) \equiv g_{b^*_{\cR},0}(x)$ satisfying
\begin{align*}
\left\{
\begin{aligned}
    &\frac{\sigma^2(x)}{2} (g_{b^*_{\cR}})^{\prime}(x) + \mu(x) g_{b^*_{\cR}}(x) + \frac{(1-\cR)}{2} \sigma^2(x) (g_{b^*_{\cR}})^2(x) = \rho, \\
    &g_{b^*_{\cR}}(b^*_{\cR}) = 1 / \psi^+(b^*_{\cR}),
\end{aligned}
\right.
\end{align*}
as in \eqref{eq:ODE_g}.
%Fix $\cR_1,\cR_2\in{\cal I}$, %
For notational simplicity,
let $g_i$ denote $g_{b^*_{\cR_i}}$ for $i=1,2$. Then, %by \eqref{eq:ODE_g}, 
\begin{align*}
    g_i(x) &= g_i(0) - \int_0^{x} g'_i(y) \, dy = \frac{(v^*)^{\prime}(0;\cR_i)}{v^*(0;\cR_i)} - \int_0^{x} g'_i(y) \, dy \\
    &= \frac{(v^*)^{\prime}(0;\cR_i)}{\xi_0} + \int_0^{x} \frac{2\mu(y)}{\sigma^2(y)} g_i(y) + (1-\cR_1) g_i^2(y) - \rho dy.
\end{align*}
Therefore, for $x \in [0, \hat{\boldsymbol{b}}]$,
\begin{align*}
    |g_1(x) - g_2(x)| \le & \frac{1}{\xi_0} \left| (v^*)^{\prime}(0;\cR_1) - (v^*)^{\prime}(0;\cR_2) \right| + C_{\hat{\boldsymbol{b}}} \int_{0}^{x} |g_1(y) - g_2(y)| dy \\
    &+ \int_{0}^{x} (\cR_2 - \cR_1) g_2^2(y) dy,
\end{align*}
where $C_{\hat{\boldsymbol{b}}}:= \sup_{0 \le y \le \hat{\boldsymbol{b}}} \left\{ \frac{2\mu(y)}{\sigma^2(y)} + (1-\cR_1) |g_1(y) + g_2(y)| \right\}$, for which we recall that $\hat{\boldsymbol{b}}$ is a uniform upper bound for $b^*_{\cR}$ by Proposition~\ref{pro:shooting}.
By Gr\"onwall's inequality, there is some constant $c>0$ such that
\begin{align*}
    \sup_{0 \le x \le \hat{\boldsymbol{b}}} |g_1(x) - g_2(x)| \le c e^{C_{\hat{\boldsymbol{b}}}} \big( |(v^*)^{\prime}(0;\cR_1) - (v^*)^{\prime}(0;\cR_2)| + |\cR_2 - \cR_1| \big).
\end{align*}
Since ${\cal I}\ni\cR \mapsto (v^*)^{\prime}(0;\cR)$ is continuous (by {Step 1}), we have
\begin{align}\label{eq:g_cts}
   \lim_{\substack{|\cR_2 - \cR_1| \to 0^+ \\ \cR_1,\,\cR_2 \in {\cal I}}} \sup_{0 \le x \le \hat{\boldsymbol{b}}} |g_1(x) - g_2(x)| = 0.
\end{align}
%We thus show that ${\cal I}\ni \cR \mapsto g_{b^*_{\cR}}(x)$ is continuous, uniformly in $x \in [0, \hat{\boldsymbol{b}}]$.
%\vspace{0.3em}

\noindent{\it Step 4. Continuity of ${\cal I}\ni \cR \mapsto b^*_{\cR}$.}
By definition of $b^*_{\cR}$ in Proposition~\ref{pro:shooting_sub}, 
%the relationship between $v^*(x;\cR)$ and $g_{b^*_{\cR}}(x)$ as given in 
\eqref{eq:equiv_deriv_2} and \eqref{eq:const_v}, we have
\begin{align}\label{eq:gR_relation}
   \frac{ \psi^+\big(b^*_{\cR_1};\cR_1\big) e^{-\int_0^{b^*_{\cR_1}} g_1(y)dy} }{\psi^+\big(b^*_{\cR_2};\cR_2\big) e^{-\int_0^{b^*_{\cR_2}} g_2(y)dy}} = \frac{v^*(0;\cR_1)}{v^*(0;\cR_2)} = 1.
\end{align}
To proceed, we recall several properties established: 
%in Section~\ref{sec:4}. By the definition of \(\psi^+\) in Assumption~\ref{as:bdy_cond}, 
the mapping \(\cR \mapsto \psi^+(x;\cR)\) is continuous and uniformly decreasing for \(x \geq 0\); since \(\psi^+(x;\cR) - x\) is decreasing for \(x > \underline{b}\), it follows that \((\psi^+)^{\prime}(x;\cR) \leq 1\) for \(x > \underline{b}\); Lemma~\ref{lem:lowerbound_g}~i. asserts that \(g_i(x) \geq 1 / \psi^+(x;\cR_i)\) for \(x > \underline{b}_{\cR_i}\).

Taking logarithms on both sides of \eqref{eq:gR_relation}, we consider the cases separately.

%\noindent 
\textit{Case 1:} \(b^*_{\cR_1} < b^*_{\cR_2}\).  
We have
\begin{align*}
   \ln \psi^+\big(b^*_{\cR_1};\cR_1\big) - \ln \psi^+\big(b^*_{\cR_2};\cR_1\big) 
     + \int_{b^*_{\cR_1}}^{b^*_{\cR_2}} g_1 (y) \, dy &\le \int_0^{b^*_{\cR_2}} \big(g_1(y)-g_2(y)\big) \, dy, \\
   \int_{b^*_{\cR_1}}^{b^*_{\cR_2}} \left( g_1 (y) -  \frac{(\psi^+)^{\prime}(y;\cR_1)}{\psi^+(y;\cR_1)} \right) \, dy &\le \int_0^{b^*_{\cR_2}} |g_1(y)-g_2(y)| \, dy, \\
   0 \le \int_{b^*_{\cR_1}}^{b^*_{\cR_2}}\left( g_1 (y) -\frac{1}{\psi^+(y;\cR_1)} \right) \, dy &\le \int_0^{b^*_{\cR_2}} |g_1(y)-g_2(y)| \, dy,
\end{align*}
where the first inequality follows from the monotonicity of \(\cR \mapsto \psi^+(x;\cR)\), and in the last line, we use \((\psi^+)^{\prime}(x;\cR_1) \le 1\) and \(g_1(x) \ge 1/ \psi^+(x;\cR_1)\) for \( x > b^*_{\cR_1}>\underline{b}_{\cR_1}\).

\textit{Case 2:} \(b^*_{\cR_1} > b^*_{\cR_2}\).  
By analogous reasoning,
\begin{align*}
    % &\ln \psi^+\big(b^*_{\cR_2};\cR_2\big) -  \ln \psi^+\big(b^*_{\cR_1};\cR_1\big)  + \int_{b^*_{\cR_2}}^{b^*_{\cR_1}} g_2 (y) \, dy = \int_0^{b^*_{\cR_1}} \big(g_2(y)-g_1(y)\big) \, dy, \\
    % &\ln \psi^+\big(b^*_{\cR_2};\cR_2\big) -  \ln \psi^+\big(b^*_{\cR_2};\cR_1\big) + \ln \psi^+\big(b^*_{\cR_2};\cR_1\big) -\ln \psi^+\big(b^*_{\cR_1};\cR_1\big) \\
    % &\qquad + \int_{b^*_{\cR_2}}^{b^*_{\cR_1}} g_2 (y) \, dy = \int_0^{b^*_{\cR_1}} \big(g_2(y)-g_1(y)\big) \, dy, \\
    0 &< \int_{b^*_{\cR_2}}^{b^*_{\cR_1}}  g_2 (y) -  \frac{1}{\psi^+(y;\cR_1)} \, dy \\
    &\le \left| \ln \psi^+\big(b^*_{\cR_2};\cR_1\big) - \ln \psi^+\big(b^*_{\cR_2};\cR_2\big) \right|  + \int_0^{b^*_{\cR_1}} |g_2(y)-g_1(y)| \, dy,
\end{align*}
where the strict inequality holds because \(\psi^+(x;\cR_2) < \psi^+(x;\cR_1)\), which implies \(g_2(x) \ge 1/\psi^+(x;\cR_2) > 1/\psi^+(x;\cR_1)\) for \(\underline{b}_{\cR_2} < b^*_{\cR_2} < x\).

For both cases, by \eqref{eq:g_cts} and the continuity of $\cR \mapsto \psi^+(x;\cR)$, we can have the desired continuity of $b^*_{\cal R}$.
% \[
% \lim_{\substack{|\cR_2 - \cR_1| \to 0^+ \\ \cR_1,\,\cR_2 \in {\cal I}}} \left| b^*_{\cR_1} - b^*_{\cR_2} \right| = 0,
% \]
% as claimed.
 
\end{proof}

\bibliographystyle{abbrv}
\bibliography{library}

@article{balter2025model,
  title={Model Ambiguity versus Model Misspecification in Dynamic Portfolio Choice},
  author={Maenhout, Pascal J. and Xing, Hao and Balter, Anne},
  journal={J. Finance},
  year={2025},
  note={Forthcoming}
}

@article{Musiela2011IJTAF,
author = {M. Musiela and T. Zariphopoulou},
title = {INITIAL INVESTMENT CHOICE AND OPTIMAL FUTURE ALLOCATIONS UNDER TIME-MONOTONE PERFORMANCE CRITERIA},
journal = {Int. J. Theor. Appl. Finance},
volume = {14},
number = {01},
pages = {61-81},
year = {2011}
}

@article{ElKaroui2013SIFIN,
author = {N. {El Karoui} and M. Mrad},
title = {An Exact Connection between Two Solvable {SDEs} and a Nonlinear Utility Stochastic {PDE}},
journal = {SIAM J. Financial Math.},
volume = {4},
number = {1},
pages = {697-736},
year = {2013}
}

@article{Musiela2009QF,
author = {M. Musiela and T. Zariphopoulou},
title = {Portfolio choice under dynamic investment performance criteria},
journal = {Quant. Finance},
volume = {9},
number = {2},
pages = {161--170},
year = {2009},
publisher = {Routledge}
}

@article{Musiela2010SIFIN,
author = {Musiela, M. and Zariphopoulou, T.},
title = {Portfolio Choice under Space-Time Monotone Performance Criteria},
journal = {SIAM J. Financial Math.},
volume = {1},
number = {1},
pages = {326-365},
year = {2010}
}

@article{Liang2017SIFIN,
author = {Liang, Gechun and Zariphopoulou, Thaleia},
title = {Representation of Homothetic Forward Performance Processes in Stochastic Factor Models via Ergodic and Infinite Horizon {BSDE}},
journal = {SIAM J. Financial Math.},
volume = {8},
number = {1},
pages = {344-372},
year = {2017}
}

@book{protter2003,
  title={Stochastic Integration and Differential Equations},
  author={Philip Protter},
  year={2003},
  edition={2},
  publisher={Springer Berlin, Heidelberg}
}

@book{nagle1996fundamentals,
  title={Fundamentals of Differential Equations And Boundary Value Problems},
  author={Nagle, R Kent and Saff, Edward B and Snider, Arthur David},
  year={2017},
  edition={7},
  publisher={Pearson}
}

@article{Cadenillas2007MF,
author = {Cadenillas, Abel and Sarkar, Sudipto and Zapatero, Fernando},
title = {OPTIMAL DIVIDEND POLICY WITH MEAN-REVERTING CASH RESERVOIR},
journal = {Math. Finance},
volume = {17},
number = {1},
pages = {81-109},
year = {2007}
}

@article{pardoux1990adapted,
  title={Adapted solution of a backward stochastic differential equation},
  author={Pardoux, Etienne and Peng, Shige},
  journal={Syst. Control Lett.},
  volume={14},
  number={1},
  pages={55--61},
  year={1990},
  publisher={Elsevier}
}

@article{popier2007backward,
  title={Backward stochastic differential equations with random stopping time and singular final condition},
  author={Popier, Alexandre},
  journal={Ann. Probab.},
  pages={1071--1117},
  year={2007},
  publisher={JSTOR}
}

@article{briand2003lp,
  title={{$L^p$} solutions of backward stochastic differential equations},
  author={Briand, Ph and Delyon, Bernard and Hu, Ying and Pardoux, Etienne and Stoica, Lucretiu},
  journal={Stochastic Process. Appl.},
  volume={108},
  number={1},
  pages={109--129},
  year={2003},
  publisher={Elsevier}
}

@book{fleming2006controlled,
  title={Controlled Markov processes and viscosity solutions},
  author={Fleming, Wendell H and Soner, Halil Mete},
  volume={25},
  year={2006},
  publisher={Springer New York, NY}
}

@article{DuffieEpstein1992,
 author = {Darrell Duffie and Larry G. Epstein},
 journal = {Econometrica},
 number = {2},
 pages = {353--394},
 publisher = {[Wiley, Econometric Society]},
 title = {Stochastic Differential Utility},
 urldate = {2023-09-25},
 volume = {60},
 year = {1992}
}

@article{EpsteinZin1989,
 author = {Larry G. Epstein and Stanley E. Zin},
 journal = {Econometrica},
 number = {4},
 pages = {937--969},
 publisher = {[Wiley, Econometric Society]},
 title = {Substitution, Risk Aversion, and the Temporal Behavior of Consumption and Asset Returns: A Theoretical Framework},
 urldate = {2023-09-25},
 volume = {57},
 year = {1989}
}

@article{Schroder1999,
title = {Optimal Consumption and Portfolio Selection with Stochastic Differential Utility},
journal = {J. Econ. Theory},
volume = {89},
number = {1},
pages = {68-126},
year = {1999},
author = {Mark Schroder and Costis Skiadas}
}

@article{xing2017consumption,
  title={Consumption--investment optimization with {Epstein--Zin} utility in incomplete markets},
  author={Xing, Hao},
  journal={Finance Stoch.},
  volume={21},
  pages={227--262},
  year={2017},
  publisher={Springer}
}

@article{skiadas2003robust,
  title={Robust control and recursive utility},
  author={Skiadas, Costis},
  journal={Finance Stoch.},
  volume={7},
  pages={475--489},
  year={2003},
  publisher={Springer}
}

@article{herdegen2021proper,
  title={Proper solutions for {Epstein--Zin} stochastic differential utility},
  author={Herdegen, Martin and Hobson, David and Jerome, Joseph},
  journal={Finance Stoch.},
  volume={29},
  number={},
  pages={885--932},
  year={2025},
  publisher={Springer}
}

@article{park2023irreversible,
  title={Irreversible Consumption Habit under Ambiguity: Singular Control and Optimal {$G$}-Stopping Time},
  author={Park, Kyunghyun and Chen, Kexin and Wong, Hoi Ying},
  journal={Ann. Appl. Probab.},
  volume={35},
  number={4},
  pages={2471--2525},
  year={2025},
  publisher={Institute of Mathematical Statistics}
}

@article{Aurand2023,
author = {Aurand, Joshua and Huang, Yu-Jui},
title = {{Epstein-Zin} utility maximization on a random horizon},
journal = {Math. Finance},
volume = {33},
number = {4},
pages = {1370-1411},
year = {2023}
}

@article{choulli2003diffusion,
  title={A diffusion model for optimal dividend distribution for a company with constraints on risk control},
  author={Choulli, Tahir and Taksar, Michael and Zhou, Xun Yu},
  journal={SIAM J. Control Optim.},
  volume={41},
  number={6},
  pages={1946--1979},
  year={2003},
  publisher={SIAM}
}

@article{herdegen2023infinite2,
  title={The infinite-horizon investment--consumption problem for {Epstein--Zin} stochastic differential utility. {II: Existence, uniqueness and verification for $\vartheta\in(0, 1)$}},
  author={Herdegen, Martin and Hobson, David and Jerome, Joseph},
  journal={Finance Stoch.},
  volume={27},
  number={1},
  pages={159--188},
  year={2023},
  publisher={Springer}
}

@article{herdegen2023infinite,
  title={The infinite-horizon investment--consumption problem for {Epstein--Zin} stochastic differential utility. {I: Foundations}},
  author={Herdegen, Martin and Hobson, David and Jerome, Joseph},
  journal={Finance Stoch.},
  volume={27},
  number={1},
  pages={127--158},
  year={2023},
  publisher={Springer}
}

@article{kobylanski2000backward,
  author = {Magdalena Kobylanski},
    title = {{Backward stochastic differential equations and partial differential equations with quadratic growth}},
    volume = {28},
    journal = {Ann. Probab.},
    number = {2},
    publisher = {Institute of Mathematical Statistics},
    pages = {558 -- 602},
    year = {2000}
}

@article{kruk2007explicit,
 ISSN = {00911798},
 author = {Lukasz Kruk and John Lehoczky and Kavita Ramanan and Steven Shreve},
 journal = {Ann. Probab.},
 number = {5},
 pages = {1740--1768},
 publisher = {Institute of Mathematical Statistics},
 title = {An Explicit Formula for the {Skorokhod} Map on {$[0, a]$}},
 urldate = {2023-06-20},
 volume = {35},
 year = {2007}
}

@article{chakraborty2023optimal,
  title={Optimal dividends under model uncertainty},
  author={Chakraborty, Prakash and Cohen, Asaf and Young, Virginia R},
  journal={SIAM J. Financial Math.},
  volume={14},
  number={2},
  pages={497--524},
  year={2023},
  publisher={SIAM}
}

@article{cohen2022optimal,
  title={Optimal ergodic harvesting under ambiguity},
  author={Cohen, Asaf and Hening, Alexandru and Sun, Chuhao},
  journal={SIAM J. Control Optim.},
  volume={60},
  number={2},
  pages={1039--1063},
  year={2022},
  publisher={SIAM}
}

@article{alvarez2006class,
  title={A class of solvable stochastic dividend optimization problems: on the general impact of flexibility on valuation},
  author={Alvarez, Luis HR and Virtanen, Jukka},
  journal={Econ. Theory},
  pages={373--398},
  year={2006},
  publisher={JSTOR}
}

@book{borodin2015handbook,
  title={Handbook of Brownian Motion -- Facts And Formulae},
  author={Borodin, A.N. and Salminen, P.},
  year={2015},
  publisher={Springer Science \& Business Media}
}

@article{maenhout2004robust,
  title={Robust portfolio rules and asset pricing},
  author={Maenhout, Pascal J},
  journal={Rev. Financ. Stud.},
  volume={17},
  number={4},
  pages={951--983},
  year={2004},
  publisher={Oxford University Press}
}

@Book{KS1991,
	author ={I. Karatzas and S.E. Shreve},
	title =	 "{Brownian Motion and Stochastic Calculus}",
	year =	 {1991},
	publisher= {Springer-Verlag},
}

@article{Black1976,
	author = {Black, F.},
	title = "{The dividend puzzle.}",
	journal = {J. Portf. Manag.},
	volume = {2},
	number = {2},
	pages = {5--8},
	year = {1976},
	month = {},
}

@article{Miller1961,
	author = {Miller, M. and Modigliani, F.},
	title = "{Dividend policy, growth, and the valuation of shares.}",
	journal = {The Journal of Business},
	volume = {34},
	number = {4},
	pages = {411--433},
	year = {1961},
	month = {},
}

@article{Bhattacharya1979,
	author = {Bhattacharya, S.},
	title = "{Imperfect information, dividend policy, and ``the bird in the hand'' fallacy.}",
	journal = {Rand J. Econ.},
	volume = {10},
	number = {1},
	pages = {259--270},
	year = {1979},
	month = {},
}

@article{John1985,
	author = {John, K. and Williams, J.},
	title = "{Dividends, dilution, and taxes: A signalling equilibrium.}",
	journal = {J. Finance},
	volume = {40},
	number = {4},
	pages = {1053--1070},
	year = {1985},
	month = {},
}

@article{Miller1985,
	author = {Miller, M.H. and Rock, K.},
	title = "{Dividend Policy under Asymmetric Information.}",
	journal = {J. Finance},
	volume = {40},
	number = {4},
	pages = {1031--1051},
	year = {1985},
	month = {},
}

@article{Black1986,
	author = {Black, F.},
	title = "{Noise.}",
	journal = {J. Finance},
	volume = {41},
	number = {3},
	pages = {528--543},
	year = {1986},
	month = {},
}

@article{Brav2005,
	author = {Brav, A. and Graham, J.R. and Harvey, C.R. and Michaely, R.},
	title = "{Payout policy in the 21st century.}",
	journal = {J. Financ. Econ.},
	volume = {77},
	number = {3},
	pages = {483--527},
	year = {2005},
	month = {},
}

@article{Ham2020,
	author = {Ham, C.G. and Kaplan, Z.R. and Leary, R.T.},
	title = "{Do dividends convey information about future earnings?}",
	journal = {J. Financ. Econ.},
	volume = {136},
	number = {2},
	pages = {547--570},
	year = {2020},
	month = {},
}

@article{Michaely2021,
	author = {Michaely, R. and Rossi, S. and Weber, M.},
	title = "{Signaling safety.}",
	journal = {J. Financ. Econ.},
	volume = {139},
	number = {2},
	pages = {405--427},
	year = {2021},
	month = {},
}

@article{Reppen2020,
	author = {Reppen, A. and Rochet, J.C. and Soner, H.},
	title = "{Optimal dividend policies with random profitability.}",
	journal = {Math. Finance},
	volume = {30},
	number = {},
	pages = {228--259},
	year = {2020},
	month = {},
}

@article{jeanblanc1995optimization,
  title={Optimization of the flow of dividends},
  author={Jeanblanc-Picqu{\'e}, Monique and Shiryaev, Albert Nikolaevich},
  journal={Russian Mathematical Surveys},
  volume={50},
  number={2},
  pages={25--46},
  year={1995},
  publisher={Russian Academy of Sciences, Steklov Mathematical Institute of Russian~…}
}

@article{shreve1984optimal,
  title={Optimal consumption for general diffusions with absorbing and reflecting barriers},
  author={Shreve, Steven E and Lehoczky, John P and Gaver, Donald P},
  journal={SIAM J. Control Optim.},
  volume={22},
  number={1},
  pages={55--75},
  year={1984},
  publisher={SIAM}
}

@article{kraft2017optimal,
  title={Optimal consumption and investment with {Epstein--Zin} recursive utility},
  author={Kraft, Holger and Seiferling, Thomas and Seifried, Frank Thomas},
  journal={Finance Stoch.},
  volume={21},
  pages={187--226},
  year={2017},
  publisher={Springer}
}

@article{matoussi2018convex,
  title={Convex duality for {Epstein--Zin} stochastic differential utility},
  author={Matoussi, Anis and Xing, Hao},
  journal={Math. Finance},
  volume={28},
  number={4},
  pages={991--1019},
  year={2018},
  publisher={Wiley Online Library}
}

@article{melnyk2020lifetime,
  title={Lifetime investment and consumption with recursive preferences and small transaction costs},
  author={Melnyk, Yaroslav and Muhle-Karbe, Johannes and Seifried, Frank Thomas},
  journal={Math. Finance},
  volume={30},
  number={3},
  pages={1135--1167},
  year={2020},
  publisher={Wiley Online Library}
}

@article{soner1991free,
  title={A free boundary problem related to singular stochastic control: the parabolic case},
  author={Soner, H Mete and Shreve, Steven E and El~Karoui, N},
  journal={Commun. Partial Differ. Equ.},
  volume={16},
  number={2-3},
  pages={373--424},
  year={1991},
  publisher={Taylor \& Francis}
}

@article{ferrari2015integral,
  title={On an integral equation for the free-boundary of stochastic, irreversible investment problems},
  author={Ferrari, Giorgio},
  journal={Ann. Appl. Probab.},
  volume={25},
  number={1},
  year={2015}
}

@article{de2017optimal,
  title={Optimal boundary surface for irreversible investment with stochastic costs},
  author={De Angelis, Tiziano and Federico, Salvatore and Ferrari, Giorgio},
  journal={Math. Oper. Res.},
  volume={42},
  number={4},
  pages={1135--1161},
  year={2017},
  publisher={INFORMS}
}

@article{shreve1991free,
  title={A free boundary problem related to singular stochastic control},
  author={Shreve, Steven E and Soner, Halil M},
  journal={Applied stochastic analysis (London, 1989)},
  volume={5},
  pages={265--301},
  year={1991},
  publisher={Gordon and Breach New York}
}

@article{bank2005optimal,
  title={Optimal control under a dynamic fuel constraint},
  author={Bank, Peter},
  journal={SIAM J. Control Optim.},
  volume={44},
  number={4},
  pages={1529--1541},
  year={2005},
  publisher={SIAM}
}

@article{bank2004stochastic,
  title={A stochastic representation theorem with applications to optimization and obstacle problems},
  author={Bank, Peter and {El Karoui}, Nicole},
  journal={Ann. Probab.},
  volume={32},
  number={1B},
  pages={1030--1067},
  year={2004},
  publisher={Institute of Mathematical Statistics}
}

@article{bank2017convex,
  title={CONVEX DUALITY FOR STOCHASTIC SINGULAR CONTROL PROBLEMS},
  author={Bank, Peter and Kauppila, Helena},
  journal={Ann. Appl. Probab.},
  pages={485--516},
  year={2017},
  publisher={JSTOR}
}

@article{bank2001optimal,
  title={Optimal consumption choice with intertemporal substitution},
  author={Bank, Peter and Riedel, Frank},
  journal={Ann. Appl. Probab.},
  volume={11},
  number={3},
  pages={750--788},
  year={2001},
  publisher={Institute of Mathematical Statistics}
}

@article{ferrari2022optimal,
  title={Optimal consumption with {Hindy--Huang--Kreps} preferences under nonlinear expectations},
  author={Ferrari, Giorgio and Li, Hanwu and Riedel, Frank},
  journal={Adv. Appl. Probab.},
  volume={54},
  number={4},
  pages={1222--1251},
  year={2022},
  publisher={Cambridge University Press}
}

@article{ferrari2022knightian,
  title={A Knightian irreversible investment problem},
  author={Ferrari, Giorgio and Li, Hanwu and Riedel, Frank},
  journal={J. Math. Anal. Appl.},
  volume={507},
  number={1},
  pages={125744},
  year={2022},
  publisher={Elsevier}
}

@article{bayraktar2013multidimensional,
	title={On the multidimensional controller-and-stopper games},
	author={Bayraktar, Erhan and Huang, Yu-Jui},
        journal={SIAM J. Control Optim.},
	volume={51},
	number={2},
	pages={1263--1297},
	year={2013},
	publisher={SIAM}
}

@article{bayraktar2011optimal_a,
	title={Optimal stopping for non-linear expectations—{Part I}},
	author={Bayraktar, Erhan and Yao, Song},
  journal={Stochastic Process. Appl.},
	volume={121},
	number={2},
	pages={185--211},
	year={2011},
	publisher={Elsevier}
}

@article{bayraktar2011optimal_b,
	title={Optimal stopping for non-linear expectations—{Part II}},
	author={Bayraktar, Erhan and Yao, Song},
  journal={Stochastic Process. Appl.},
	volume={121},
	number={2},
	pages={212--264},
	year={2011},
	publisher={Elsevier}
}

@article{park2023robust,
	title={Robust retirement and life insurance with inflation risk and model ambiguity},
	author={Park, Kyunghyun and Wong, Hoi Ying and Yan, Tingjin},
  journal={Insur. Math. Econ.},
	volume={110},
	pages={1--30},
	year={2023},
	publisher={Elsevier}
}

@article{PW23,
	author = {Park, K. and Wong, H. Y. },
	title = {Robust Retirement with Return Ambiguity: Optimal {$G$}-Stopping Time in Dual Space},
        journal={SIAM J. Control Optim.},
	volume = {61},
	number ={3},
	year =	 {2023},
	pages ={1009--1037}
}

@article{maenhout2006robust,
  title={Robust portfolio rules and detection-error probabilities for a mean-reverting risk premium},
  author={Maenhout, Pascal J},
  journal={J. Econ. Theory},
  volume={128},
  number={1},
  pages={136--163},
  year={2006},
  publisher={Elsevier}
}

@article{jin2018dynamic,
  title={Dynamic asset allocation with uncertain jump risks: a pathwise optimization approach},
  author={Jin, Xing and Luo, Dan and Zeng, Xudong},
  journal={Math. Oper. Res.},
  volume={43},
  number={2},
  pages={347--376},
  year={2018},
  publisher={Informs}
}

@article{anderson2003quartet,
  title={A quartet of semigroups for model specification, robustness, prices of risk, and model detection},
  author={Anderson, Evan W and Hansen, Lars Peter and Sargent, Thomas J},
  journal={J. Eur. Econ. Assoc.},
  volume={1},
  number={1},
  pages={68--123},
  year={2003},
  publisher={Oxford University Press}
}

@article{laeven2014robust,
  title={Robust portfolio choice and indifference valuation},
  author={Laeven, Roger JA and Stadje, Mitja},
  journal={Math. Oper. Res.},
  volume={39},
  number={4},
  pages={1109--1141},
  year={2014},
  publisher={INFORMS}
}

@article{csiszar1975divergence,
  title={{$I$}-divergence geometry of probability distributions and minimization problems},
  author={Csisz{\'a}r, Imre},
  journal = {Ann. Probab.},
  pages={146--158},
  year={1975},
  publisher={JSTOR}
}

@article{ben1985entropic,
  title={The entropic penalty approach to stochastic programming},
  author={Ben-Tal, Aharon},
  journal={Math. Oper. Res.},
  volume={10},
  number={2},
  pages={263--279},
  year={1985},
  publisher={INFORMS}
}

@article{branger2013robust,
  title={Robust portfolio choice with uncertainty about jump and diffusion risk},
  author={Branger, Nicole and Larsen, Linda Sandris},
  journal={J. Bank. Finance},
  volume={37},
  number={12},
  pages={5036--5047},
  year={2013},
  publisher={Elsevier}
}

@article{yi2013robust,
  title={Robust optimal control for an insurer with reinsurance and investment under Heston’s stochastic volatility model},
  author={Yi, Bo and Li, Zhongfei and Viens, Frederi G and Zeng, Yan},
  journal={Insur. Math. Econ.},
  volume={53},
  number={3},
  pages={601--614},
  year={2013},
  publisher={Elsevier}
}

@article{nutz2015optimal,
  title={OPTIMAL STOPPING UNDER ADVERSE NONLINEAR EXPECTATION AND RELATED GAMES},
  author={Nutz, Marcel and Zhang, Jianfeng},
  journal={Ann. Appl. Probab.},
  volume={25},
  number={5},
  pages={2503--2534},
  year={2015}
}

@article{riedel2009optimal,
  title={Optimal stopping with multiple priors},
  author={Riedel, Frank},
  journal={Econometrica},
  volume={77},
  number={3},
  pages={857--908},
  year={2009},
  publisher={Wiley Online Library}
}

@article{monoyios2024stability,
  title={Stability of the {Epstein--Zin} problem},
  author={Monoyios, Michael and Mostovyi, Oleksii},
  journal={Math. Finance},
  year={2024},
  publisher={Wiley Online Library}
}

@article{bartl2023sensitivity,
  title={Sensitivity of robust optimization problems under drift and volatility uncertainty},
  author={Bartl, Daniel and Neufeld, Ariel and Park, Kyunghyun},
  journal={Finance Stoch., arXiv:2311.11248},
  year={2026+},
  note={forthcoming}
}

@article{herrmann2017model,
  title={Model uncertainty, recalibration, and the emergence of delta--vega hedging},
  author={Herrmann, Sebastian and Muhle-Karbe, Johannes},
  journal={Finance Stoch.},
  volume={21},
  number={4},
  pages={873--930},
  year={2017},
  publisher={Springer}
}

@article{herrmann2017hedging,
  title={Hedging with small uncertainty aversion},
  author={Herrmann, Sebastian and Muhle-Karbe, Johannes and Seifried, Frank Thomas},
  journal={Finance Stoch.},
  volume={21},
  number={1},
  pages={1--64},
  year={2017},
  publisher={Springer}
}

@article{bartl2024numerical,
  title={Numerical method for nonlinear {Kolmogorov PDEs} via sensitivity analysis},
  author={Bartl, Daniel and Neufeld, Ariel and Park, Kyunghyun},
  journal={Appl. Math. Optim.},
  volume={93},
  number={75},
  year={2026}
}

@book{revuz2013continuous,
  title={Continuous martingales and Brownian motion},
  author={Revuz, Daniel and Yor, Marc},
  volume={293},
  year={2013},
  publisher={Springer Science \& Business Media}
}

@article{arthur2025sensitivity,
  title={Sensitivity Analysis of Distributionally Robust {BSDEs} and {RBSDEs}},
  author={Arthur, Compoint and Nathan, Sauldubois and Nizar, Touzi},
  journal={arXiv preprint arXiv:2511.01828},
  year={2025}
}

\end{document}